\pdfoutput=1

\newcommand{\eqnref}[1]{Eq.\ (\ref{#1})}
\newcommand{\Eqnref}[1]{Equation (\ref{#1})}
\newcommand{\eqnrefs}[2]{Eqs.\ (\ref{#1}) and (\ref{#2})}

\newcommand{\step}[1]{step ($#1$)}

\newcommand{\stepd}[2]{steps ($#1$)--($#2$)}

\newcommand{\secref}[1]{Sec.\ \ref{#1}}
\newcommand{\Secref}[1]{Section \ref{#1}}
\newcommand{\secrefs}[2]{Secs.\ \ref{#1} and \ref{#2}}

\newcommand{\Appref}[1]{Appendix \ref{#1}}
\newcommand{\Apprefs}[2]{Appendices \ref{#1} and \ref{#2}}
\newcommand{\figref}[1]{Fig.\ \ref{#1}}

\newcommand{\subfigref}[2]{Fig.\ \ref{#1}(#2)} 
\newcommand{\subFigref}[2]{Figure \ref{#1}(#2)} 
\newcommand{\Figref}[1]{Figure \ref{#1}}

\newcommand{\ErdosRenyi}[0]{Erd\H os-R\'enyi}
\newcommand{\bigO}[1]{$O(#1)$}
\newcommand{\kpowavg}[0]{\langle k\rangle}

\newcommand{\widesubfigsize}[0]{0.68\columnwidth}

\newcommand{\INN}[0]{U}
\newcommand{\INAB}[0]{U(A,B)}
\newcommand{\INAA}[0]{U(A,A)}
\newcommand{\vab}[0]{v(a,b)}
\newcommand{\vsupab}[1]{v^{#1}(a,b)}
\newcommand{\vba}[0]{v(b,a)}
\newcommand{\vaa}[0]{v(a,a)}
\newcommand{\voab}[0]{\overline{v}(a,b)}
\newcommand{\Inab}[0]{u(a,b)}

\newcommand{\Inaa}[0]{u(a,a)}
\newcommand{\Ioab}[0]{\overline{u}(a,b)}

\newcommand{\Ham}{\mathcal{H}}
\newcommand{\etal}{\emph{et al}.}

\newcommand{\ie}{\emph{i.e.}}

\documentclass[aps,pre,twocolumn,showpacs,amsmath,amssymb]{revtex4-1}
\usepackage{graphicx,dcolumn,bm,verbatim,hyperref,subfigure,amsthm}

\theoremstyle{plain}  \newtheorem{theorem}{Theorem}

\begin{document} 


\title{Local multiresolution order in community detection}

\author{Peter Ronhovde}
\affiliation{Department of Physical Sciences, The University of Findlay, 
1000 N. Main St., Findlay, Ohio 45840, USA}%
\author{Zohar Nussinov}
\affiliation{Department of Physics, Washington University in St. Louis, 
Campus Box 1105, 1 Brookings Drive, St. Louis, Missouri 63130, USA}%

\date{\today}

\begin{abstract}
Community detection algorithms attempt to find the best clusters of nodes 
in an arbitrary complex network.
Multi-scale (``multiresolution'') community detection extends the problem
to identify the best network scale(s) for these clusters.
The latter task is generally accomplished by analyzing community stability
simultaneously for all clusters in the network.
In the current work, we extend this general approach to define local 
multiresolution methods, 
which enable the extraction of well-defined local communities even if the 
global community structure is vaguely defined in an average sense.
Toward this end, we propose measures analogous to variation of information 
and normalized mutual information that are used to quantitatively identify 
the best resolution(s) at the community level based on correlations between 
clusters in independently-solved systems.
We demonstrate our method on two constructed networks as well as a real network 
and draw inferences about local community strength.
Our approach is independent of the applied community detection algorithm save 
for the inherent requirement that the method be able to identify communities 
across different network scales, with appropriate changes to account for how 
different resolutions are evaluated or defined in a particular community detection 
method.
It should, in principle, easily adapt to alternative community comparison measures.
\end{abstract}

\pacs{89.75.Fb, 64.60.aq, 89.65.--s}

\maketitle{}

\section{Introduction} \label{sec:introduction}

Applications of complex network analysis span a wide range 
of seemingly unrelated fields.
In these networks, elements of the model system are abstracted as nodes
(\ie{}, people, atoms, etc.), and edges represent known relationships 
between them (\ie{}, friendships, energies, etc.).
As depicted in \figref{fig:communities}, community detection (CD)
\cite{ref:motternetworks,ref:fortunatophysrep} seeks 
to identify natural groups of related nodes in a network. 
This structure can take the form of social groups \cite{ref:lanccharacter}, 
clusters of atoms \cite{ref:RCHNstructure}, proteins \cite{ref:palla}, 
and much more.
Several categories of common real-world networks are characterized 
in Ref.\ \cite{ref:lanccharacter}.

Conceptually speaking, communities in a network are groups of nodes that are 
strongly connected inside a community but weakly connected between communities.
This basic idea is well established in the literature;
it seems to be easily quantifiable and perhaps even sufficient to rigorously 
define a community if a few small clarifications are specified.  
However, the amazing variety of CD algorithms as well as a limited consensus 
in the field contradicts this na{\" i}ve assessment.

Multiresolution community detection extends the CD concepts to find 
the most natural resolution(s) for a network partition. 
It endeavors to identify the network scales that best represent 
the community structure of a network, effectively distinguishing between 
densely or sparsely connected community members.
Similar to single resolution CD, multiresolution methods must quantitatively 
assess this description in order to obtain an objective measure of the best 
candidate partition and resolution.
A common approach, which is implemented in various ways, is to search for regions 
with stable partitions 
\cite{ref:arenasmultires,ref:kumpulamultires,ref:fenndynamic,ref:rosvallmultires,ref:chengshen},
where the community structure does not change significantly in terms of 
various applied measures of the candidate partitions (e.g., number of communities $q$, 
dynamic flow across the network, information, etc.).
Our global multiresolution algorithm \cite{ref:rzmultires} 
(MRA) asserts that the most natural resolution for a network may be identified 
based on how well independently solved replicas agree on the partition 
as evaluated by information measures \cite{ref:danon,ref:vi}.

Our local multiresolution algorithm (LMRA) 
quantitatively identifies the most natural resolution(s) for individual 
communities regardless of the weak or strong community correlations present 
in the rest of the network.
That is, the LMRA method is able to select optimal CD resolution parameter(s) 
independently for each cluster in a graph.
Our use of the term local implies that the communities are determined 
with respect to parameters defined ``near'' the individual communities 
or nodes (\ie{}, community size, relations among neighboring nodes, etc.).
Here, we solve the full network partition for every resolution, but the 
algorithm trivially adapts to CD algorithms which can identify local 
communities without partitioning the entire network, which is important 
for immense networks such as the World Wide Web.

\section{Background} \label{sec:background}

One of the most popular CD methods defines a cost function that attempts 
to quantitatively encapsulate the essential features for a ``good'' division 
of nodes, thus evaluating the best community structure in an 
objective fashion.
Regardless of the specific form, the task is to optimize the function
for a particular graph to determine the optimal node division(s).
Newman and Girvan \cite{ref:gn} introduced the most common approach by far 
with ``modularity.''
CD methods based on Potts model cost functions, or methods that may be cast 
as such \cite{ref:LPA,ref:barberLPA}, are also common.
Reichardt and Bornholdt (RB) wrote a Potts model \cite{ref:smcd} which they 
specialized into two main cases utilizing null models.
Null models are auxiliary graphs which are selected to evaluate the quality 
of a candidate partition, thus implicitly selecting the ``correct'' scale for a graph.

The choice of a null model inherently, often implicitly, selects a pre-determined scale 
for a network.  The most common null models by far are:
the ``configuration null model'' which sets edge connection probabilities based 
on the current graph, encompassing modularity as a special case,
and the \ErdosRenyi{} null model \cite{ref:reichardt} which defines the connection
probability of all edges to be equally likely based on the graph's average edge density.
Optimization of CD quality functions using these null models was shown to suffer 
from an inherent resolution limit \cite{ref:gn,ref:fortunato,ref:smcd,ref:kumpulaResLim,ref:zhangmodweak}, 
which cannot be resolved by varying the network scale \cite{ref:lancfortunatomod,ref:xiangmultireslimit}. 
This feature hinders the proper identification of some communities in large graphs.

An important general network model is the stochastic block model (SBM) \cite{ref:holland1983stochastic},
which provides a descriptive and generative model of network structure.
Such models can then be used by various CD techniques to identify community structure 
\cite{ref:snijders1997estimation,ref:newman2007mixture,ref:latouche2012variational}.
Decelle et. al. \cite{ref:decelleKMZPT,ref:decelle2011asymptotic} 
and Hu et. al. \cite{ref:hu2012phase,ref:huCDPTlong}
studied phase transitions for SBM-type networks, and Darst et. al. examined related bounds 
on well-defined communities \cite{ref:darstCDSBMdef}.
Extensions have moderated the internally homogenous nature of SBM graphs
to improve its performance when modeling realistic networks with more varied
degree distributions \cite{ref:karrer2011stochastic,ref:zhu2014oriented,ref:yan2014model}. 
Often, network models imply or impose an expected structure, but an adaptive method 
based on mixture models \cite{ref:newman2007mixture} allows for detection of unspecified 
types of structure in a variety of network classes.

Potts model and related CD approaches include \cite{ref:blatt,ref:ispolatov,
ref:hastings,ref:barberLPA,ref:traagPRE,ref:traaglocalscope,ref:rzmultires,
ref:rzlocal}, and Refs.\ \cite{ref:traagPRE,ref:traaglocalscope} generalized 
the RB Potts models in \cite{ref:smcd,ref:reichardt}, respectively.
Our previous work \cite{ref:rzmultires,ref:rzlocal} advanced a local Potts 
model, and local models were studied in more detail in \cite{ref:traaglocalscope}.
Other local methods include \cite{ref:bagrowboltlocal,ref:traaglocalscope,
ref:LPA,ref:palla,ref:lanc,ref:rzlocal,ref:havemannlocalmultires,ref:zhao2011community},
including variants of modularity \cite{ref:clausetlocal,ref:muff}.
Potts systems in CD can experience disorder from thermal effects 
\cite{ref:huCDPTsgd,ref:huCDPTlong}, extraneous edges (noise) 
\cite{ref:rzlocal,ref:huCDPTsgd,ref:goodMC,ref:nadakuditiSBM,ref:huCDPTlong}, 
and system size \cite{ref:huCDPTlong,ref:rhzglobaldisorder}. The selected 
model can also exacerbate disorder effects \cite{ref:danonhetero,ref:goodMC}.

Some CD methods, such as modularity, implicitly select a single ``objective'' 
scale for a candidate community division (e.g., Refs.\ \cite{ref:gn,ref:LPA}),
but certain networks such as hierarchical systems inherently possess multiple 
natural scales.
Hierarchical clustering is an early multiscale method \cite{ref:everittHC},
but it \emph{forces} hierarchical structure on every system without evaluating 
the relevance of the solved partitions.
That is, it assigns but does not quantititatively evaluate whether the hierarchical 
structure is a good multi-scale partitioning scheme for the graph.
More recent hierarchical approaches include 
\cite{ref:clausetmissinglinks,ref:rosvallmultires,ref:salespardo,ref:blondel,
ref:shenmodhier,ref:ahnlinkmultiscale},
and Ref.\ \cite{ref:ravaszbarabasi} relates the presence of hierarchical 
features to a scale-free-network property.

\begin{figure}
\centering
\includegraphics[width=0.6\columnwidth]{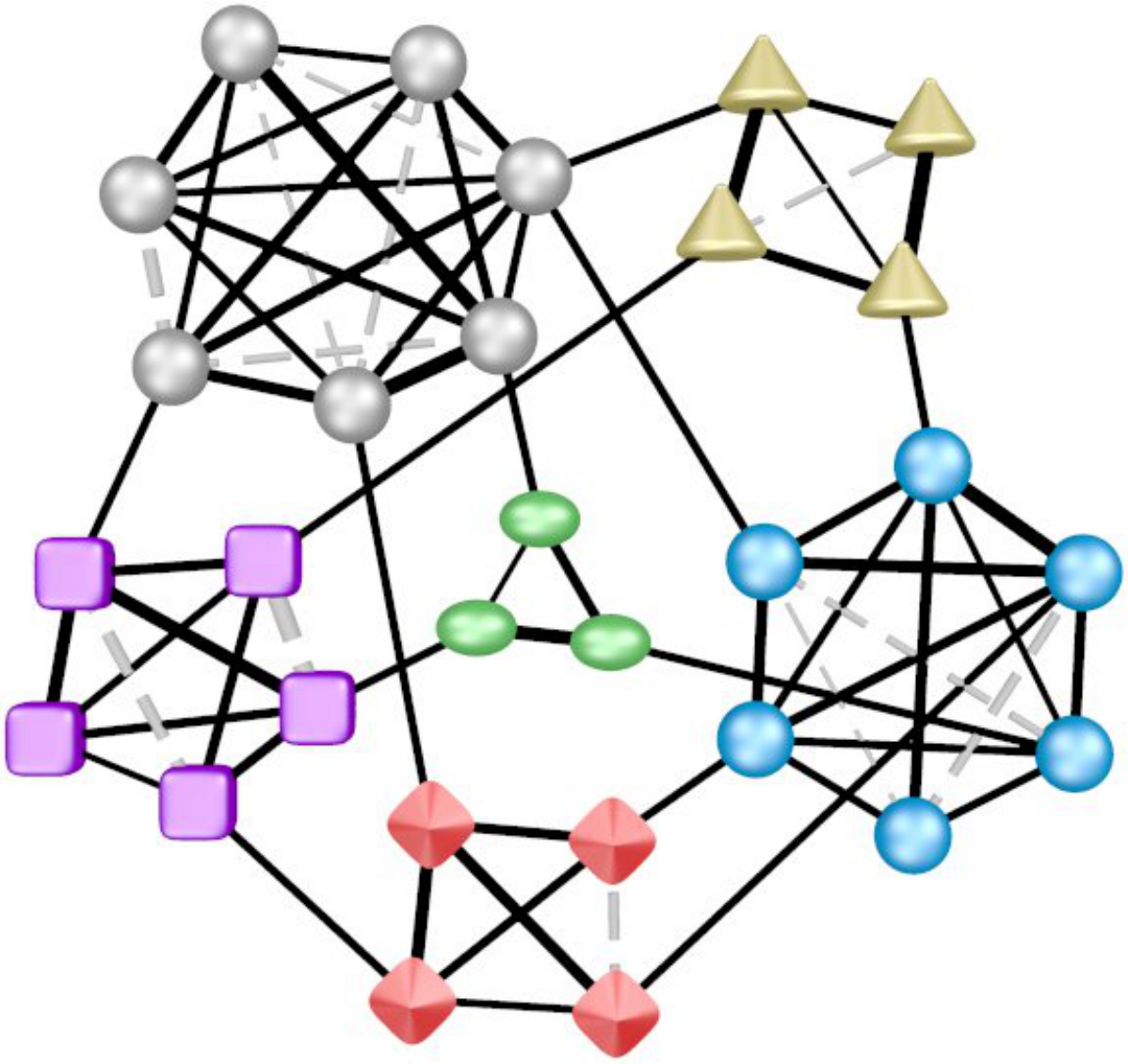}
\caption{(Color online) The figure illustrates a network partition where 
communities are represented by distinct node shapes and colors.
``Friendly'' or ``cooperative'' relations are depicted as solid, black lines
which are modeled as ferromagnetic interactions with $w_{ij}>0$ in \eqnref{eq:ourmodel}. 
``Adversarial'' or ``neutral'' relations (some are omitted for clarity) are depicted 
as gray, dashed lines.
These are modeled as antiferromagnetic interactions with $u_{ij}>0$.
In both cases, the line thickness indicates a relative interaction strength.
With \eqnref{eq:ourmodel}, neutral interactions (unconnected or unspecified 
relations) are repulsive in nature since they work like adversarial relations 
that break up well-defined communities.}
\label{fig:communities}
\end{figure}

Ideally, a CD algorithm should be able to determine all relevant scales 
of a network. 
This problem is the impetus for developing quantitative multiresolution 
network analysis.
Multiscale capable methods that utilize cost functions include 
\cite{ref:arenasmultires,ref:lanc,ref:fenndynamic,ref:smcd,ref:reichardt,
ref:rzmultires,ref:muchaSCI}.
The RB Potts model weighs the contribution of the null model \cite{ref:smcd}, 
allowing the cost function to span different network scales.
Other methods encompass varied forms of analysis \cite{ref:zhangsmall,ref:chengshen,
ref:lancstatCD,ref:shenccmmultiscale} to attack the problem.

Even with tunable CD cost function parameters, the question of which resolutions 
are the most \emph{natural} scales for a network is not necessarily answered.
Thus, multiresolution methods sought to identify the best scale(s)
\cite{ref:arenasmultires,ref:kumpulamultires,ref:rzmultires} for a network 
without imposing, or arbitrarily selecting, a preferred network scale.
The most common method detects stable resolutions in terms of network 
and model resolution parameters \cite{ref:lanc,ref:arenasmultires,ref:rzmultires}.
Our multiresolution replica algorithm calculated information-based 
correlations \cite{ref:rzmultires} among independent copies of the same system 
to quantitatively compare the partition strength across all relevant network 
scales.

To our knowledge, all current multiresolution approaches analyze the network 
robustness in an ``average'' sense across all communities in a network 
(see \secrefs{app:CDalgorithm}{app:MRA}), but \emph{the best local communities 
will not necessarily coincide at the same resolution in general}.
For example, communities in large networks may experience a ``lost-in-a-crowd'' 
effect which can obscure locally well-defined communities and limit the 
ability of global multiresolution methods (see \Appref{app:localglobal}) 
to accurately isolate their structure.
In some models, the effect can be exacerbated by heterogeneously-sized 
community structure \cite{ref:danonhetero,ref:shenlaplacianPRE} depending 
on the network scale.
Conversely, a global partition may be strongly defined for most communities, 
but a given cluster may still be weakly defined.

The LMRA method combines the benefits of multiresolution analysis with the local 
identification of community structure.
While each community exists and is defined in the context of its own network
neighborhood, we ideally prefer to identify strong communities independent 
of the global system.
That is, we allow each community to stand on its own in terms of the evaluation 
of its community structure.
Somewhat related efforts to the current work include detecting ``unbalanced'' communities 
in a network partition \cite{ref:zhangunbalanced} and an efficient seed-expansion 
method by Havemann \etal{} \cite{ref:havemannlocalmultires} which could, in principle, 
be modified for other local cost functions.

Community detection methods utilizing quality functions generally include, 
or have been extended by subsequent work to include, weight parameters that serve 
vary the target resolution, so these methods will naturally adapt to the algorithm 
described in this work.
Other disparate approaches include flow analysis \cite{ref:chengshen,ref:rosvallmultires}, 
spectral partitioning \cite{ref:shenchengspectral}, 
Bayesian analysis \cite{ref:Hofman2008}, 
dynamics \cite{ref:gudkov}, network synchronization \cite{ref:boccaletti}, 
$k$-means \cite{ref:kanungo2002}, and others.

Some of the above methods do not easily or naturally incorporate an explicit 
resolution parameter---relying rather on input community parameters, dynamics, 
or other measures of stability to identify the best resolution(s).
Several of these require the number of communities $q$ as an input parameter, 
or equivalently for our analysis, they may detect $q$ based on the network 
(e.g., eigenvalue gaps). 
Then, $q$ is passed 
to a clustering algorithm such as $k$-means.
In either of these cases, $q$ can serve as the resolution parameter for our 
LMRA algorithm, particularly for larger networks where small changes in $q$ will 
represent correspondingly small changes in the overall community structure.
For other cases, the precise implementation will be more model dependent, but 
the current LMRA algorithm only needs to receive the community partitions over 
a range of network scales, regardless of how these scales are detected or defined 
in a particular CD algorithm.

The remainder of the work is organized as follows:
we introduce our community detection Potts model in \secref{sec:hamiltonian}.
\Secref{sec:communitydefs} elaborates on concepts of community definitions, and 
\secref{sec:resolution} describes the notion of a partition resolution.
We suggest a local, community-based analogy to the 
variation of information (VI) and normalized mutual information (NMI) measures 
in \secref{sec:clusterVI} which we apply in \secref{sec:LMRA} for our local 
multiresolution algorithm.
\Secref{sec:examples} illustrates the approach with three examples, and we 
conclude in \secref{sec:conclusion}.
\Appref{app:localglobal} explains the context of local and global terminology 
used in this paper.
Finally, \Apprefs{app:semimetricproof}{app:altCVI} comment on the semi-metric 
property of our cluster measure as well as a couple alternative approaches 
to local cluster comparisons in an information-theoretic analogy.

\section{Potts Model Hamiltonian} \label{sec:hamiltonian}

Regardless of the underlying solution method, the ultimate goal of any 
community detection partitioning algorithm is a Potts-type assignment 
$i \to \sigma_{i}$  for each node $i$ into one of $q$ different 
clusters where $\sigma_i$ may be regarded as a Potts-type variable.
Toward this end, we focus directly on Potts variables.
Some methods extend this notion to include overlapping memberships
(e.g., Refs.\ \cite{ref:palla,ref:lanc,ref:ballMLlinks,ref:havemannlocalmultires}) 
where nodes may be shared between, or fractionally assigned to,  different 
communities.
In these cases, the community assignment becomes a vector quantity
for each node as opposed to a single integer value.

We identify community partitions by minimizing (see \secref{app:CDalgorithm}) 
a general CD Potts model
\begin{equation}
	\Ham(\{ \sigma \} ) = -\frac{1}{2} \sum_{i\neq j}
	    \big[ w_{ij} A_{ij} - \gamma u_{ij} \left( 1 - A_{ij} \right) \big] 
	    \delta ( \sigma_i,\sigma_j )
	\label{eq:ourmodel}
\end{equation}
which we refer to as an ``absolute'' Potts model (APM) since it is not 
defined relative to a random null model.
Assuming $N$ nodes, $\{A_{ij}\}$ is the adjacency matrix where $A_{ij}=1$ 
if nodes $i$ and $j$ are connected and is $0$ if they are not connected.
The spin variable $\sigma_{i}$ identifies the community 
membership of node $i$ in the range $1 \le \sigma_{i} \le q$ 
where node $i$ is a member of community $k$ if $\sigma_{i}=k$.
The Kronecker delta $\delta (\sigma_i,\sigma_j) = 1$ 
if $\sigma_i = \sigma_j$ and $0$ when $\sigma_i \neq \sigma_j$.
By virtue of the Kronecker delta, interactions are 
limited to spins in the same community, and they are ferromagnetic 
in nature if nodes $i$ and $j$ are connected and antiferromagnetic if they 
are not connected.  
The global resolution parameter $\gamma$ 
scales the relative effects of the ferromagnetic $\{w_{ij}\}$ and antiferromagnetic 
$\{u_{ij}\}$ interactions, effectively allowing the model to vary the network 
resolution.
A network resolution roughly corresponds to the typical community 
size, but a better characterization may be the typical community edge density 
(see \secref{sec:resolution}).

In \eqnref{eq:ourmodel}, $\{w_{ij}\}$ and $\{u_{ij}\}$ are the edge weights 
for ``cooperative'' and ``neutral'' or ``adversarial'' relations, respectively,
that are defined by the graph under consideration.  
These weights are based on known or estimated 
relations between the elements as recorded or defined by the person constructing 
the network.
We refer to edges defined by $\{w_{ij}\}$ as cooperative since these lower
the community energy (\ie{}, they reinforce the community).
Relations described by $\{u_{ij}\}$ raise the energy (\ie{}, they work 
to break up the community). In unweighted graphs, $u_{ij}=w_{ij}=1$.

Both adversarial and neutral relations serve to break up community structure, 
so the APM \cite{ref:rzmultires,ref:rzlocal} penalizes neutral relations 
much like one would expect for adversarial relations
(as opposed to zero energy contributions in a purely ferromagnetic Potts 
model \cite{ref:LPA,ref:blatt}).
This property avoids a trivial ground state solution (\ie{}, a completely 
collapsed system for every graph) present in the purely ferromagnetic Potts model.
In essence, the energy penalty for adversarial relations provides a
``penalty function'' as an alternative to how modularity resolved the trivial-ground-state
problem \cite{ref:gn} (\ie{}, by comparing a community to an average, random distribution 
of edges in the graph).
Ref.\ \cite{ref:traagPRE} generalized a common Potts model variant 
\cite{ref:smcd} to include negative link weights.

Despite the global energy sum in \eqnref{eq:ourmodel}, the model is a local
measure of community structure (see \Appref{app:localglobal}) because all node 
assignments are made strictly by evaluating local network parameters 
\cite{ref:rzlocal,ref:traaglocalscope}.
For simplicity, our current analysis will focus on undirected, static networks; 
but both \eqnref{eq:ourmodel} and the LMRA method 
in this work are suitable for general weighted, directed, 
and dynamic (time-dependent) networks.

\section{Community detection concepts} \label{sec:CDconcepts}

A precise definition of community structure in networks is still not agreed upon 
in the literature.  Generally speaking, communities consist of nodes
which are strongly connected within communities, in terms of the number or weight of edges, 
but nodes in different communities are more sparsely connected.
When constructing the quantitative community evaluation, there is also a question as to whether 
the ``inside'' versus ``outside'' degree comparison is summed over \emph{all} external 
communities \cite{ref:radicchi,ref:cafieristrong} or is evaluated only between 
\emph{individual pairs} of communities \cite{ref:rzlocal,ref:rzmultires,ref:darstCDSBMdef}.
Our model applies the latter case.

\subsection{Community definitions} \label{sec:communitydefs}

Communities in social networks are the prototypical CD model.
People often have many more ``external'' relationships of varying strengths
than they do within a local group in which they are a member.
For example, an individual may associate with a chess club, but his network 
of friendships may extend to dozens or even hundreds of people beyond this 
local group.

In many network approximations (e.g., the ubiquitous Zachary karate club 
network \cite{ref:zachary}), these ``extra'' edges are omitted as extraneous 
for the reduced-size network (\ie{}, no need to solve a large graph if we are 
only interested in the local club).
If we were to create a more comprehensive, expanded network and re-partition the system, 
the additional noise induced by including these previously external relations 
should not intuitively disturb the original communities, provided they are still 
strongly defined relative to any new structure(s) in the expanded system.
This intuitive concept is overlooked by some CD methods
because the quantitative evaluation of community structure in the expanded system 
directly changes by virtue of the size (nodes, edges) increase alone, regardless 
of whether the local relations in and around a given community are affected.

Ref.\ \cite{ref:radicchi} proposed definitions for ``strong'' and ``weak'' 
communities: in a strong community, \emph{all} nodes have more internal 
than external edges, and a weak community is one where the \emph{sum} over 
all internal edges exceeds the sum of the external edges.
A large social network may not have ``strong'' or even ``weak'' communities 
in the sense of the proposed definitions, but the communities are still 
well-defined empirically.
Thus, these community definitions \cite{ref:radicchi} neglect certain important 
(high noise) and intuitive \cite{ref:zhangunbalanced,ref:zhangmodweak} cases.

In fact, several CD methods were compared by Lancichinetti and Fortunato 
\cite{ref:lancLFRcompare} where most, if not all, communities were weakly-defined
in the sense proposed by Ref. \cite{ref:radicchi}, but many of the algorithms 
were nevertheless able to easily identify these purported weak communities.
That is, the best methods easily solved the benchmark graphs \cite{ref:lancbenchmark} 
well into regions where \emph{all} nodes (on average) have more external than internal 
edges.
This definition is apparently not characteristic enough to describe weakly-defined 
communities based on the capabilities of some CD algorithms.
Otherwise, we would intuitively think that the detection boundary would lie somewhere
near this threshold.
The crux is that the proposed definition considers the sum of external community 
edges, but identifiable communities in many CD algorithms seem to be more aligned 
toward a less restrictive definition of what a weak community is.

With these examples in mind, it seems appropriate, at least in social and 
related networks, to favor cost functions or analysis methods that utilize 
\emph{pairwise} community comparisons when evaluating node membership 
robustness.
This assumption inherently affects the notion of well-defined partitions,
communities, and individual node memberships \cite{ref:rzlocal,ref:darstCDSBMdef}.
Another approach that may be fruitful is to pursue a community definition
based on \emph{edge density} as opposed to inner and outer community edge 
\emph{counts}, but further quantitative analysis is beyond the scope of the current 
work.

\subsection{Resolution} \label{sec:resolution}

Intuitively, the resolution of a community partition is the typical 
strength of intracommunity connections. 
This concept can be quantified by the typical edge density $p$ of the 
communities in the partition. 
Communities with significantly different edge densities are qualitatively 
different.
For example, social networks may naturally display communities of ``close friends'' 
or ``acquaintances.'' 
Close friends are generally very likely to know most or all members of the 
same group ($p$ is high) where acquaintances are much less likely to know 
each other ($p$ is lower).

As a specific example, a community where each person has five friendships 
in a group of six is a clique. 
That is, every node is connected to all others in the group.
However, if we consider the same five friendships in a group of $100$, 
it may not even qualify as a community of social acquaintances.
These two clusters have an identical edge count, but they represent drastically
different \emph{types} of communities (\ie{}, different network scales).
As mentioned above, the inner and outer edge count is not sufficient to quantitatively 
describe a cluster.
This distinction highlights the importance of a penalty term in various CD 
quality functions.

In practice, a partition will contain communities with a range of edge densities, 
but intuitively, the differences should not be drastic at a given resolution since 
the partition should manifest communities with similar levels of association.
Continuing with the social network example, mixing communities of close friends 
and acquaintances in the same partition makes less sense than a partition
that indicates close friendships in most communities.
Given this argument, it is reasonable that a given $\gamma$ in \eqnref{eq:ourmodel}
could be applied to the whole graph and provide meaningful partition information 
in general, but this manuscript illustrates a method to enhance the analysis 
of complex networks by finding locally optimal resolutions at the community level.

We specialize the edge density analysis below to unweighted graphs for clarity, 
but Ref.\ \cite{ref:rzlocal} discusses weighted graphs in the same context.
The edge density of community $a$ is $p_a = \ell_a/\ell_a^\mathrm{max}$ 
where $\ell_a$ is the number of edges in the community;
$\ell_a^\mathrm{max}=n_a(n_a-1)/2$ is the maximum number of possible
edges in community $a$ with $n_a$ nodes.
The global resolution parameter $\gamma$ in \eqnref{eq:ourmodel} requires 
a \emph{minimum} edge density for each community in the partition,
\begin{equation} 
  p_\mathrm{min}\ge \frac{\gamma}{\gamma+1},
  \label{eq:gammadensity}
\end{equation}
which we calculated by determining the minimum density configuration that 
yields an energy of zero or less.
Without $\gamma$, the model can only solve a particular implicit resolution 
for all systems, $p_\mathrm{min}^{\gamma=1} \ge 1/2$. 
Other models implement similar weight parameters 
\cite{ref:reichardt,ref:smcd,ref:arenasmultires,ref:lanc,ref:barberLPA,
ref:traagPRE,ref:traaglocalscope}
which allow them to solve distinct network scales.

While \eqnref{eq:gammadensity} provides a convenient lower bound on the 
minimum community edge density, optimizing \eqnref{eq:ourmodel} implements 
the constraint implied in \eqnref{eq:gammadensity} by enforcing a stronger 
requirement.
That is, it merges network elements (a node to a community or two 
communities) if the edge density between them exceeds $p_\mathrm{min}$.
Thus, one is assured that \emph{all sub-elements of a community are connected 
by at least} $p_\mathrm{min}$.
This effectively avoids resolution-limit-type effects by acting locally 
\cite{ref:rzlocal}.

\subsection{Heuristic multiresolution method}  \label{sec:globaldiscussion}

The local and global multiresolution methods are discussed in detail later, 
but it is relevant to take a moment consider the basic function of our multiresolution 
approach and its connection to the underlying CD solver.
Briefly, the global method independently solves for the community structure 
in a set of $r$ replicas of the network.  
It then uses information-based measures (see \secref{sec:clusterVI}) to evaluate 
the partition similarities, asserting that the best resolutions have strongly
correlated partitions among the replicas.

The exact partitions determined among the replicas depend on the efficacy 
of the particular CD algorithm used to determine the individual partition solutions.
While \eqnref{eq:ourmodel} 
has a well-defined ground state of partitioned nodes that depends on the weight
parameter $\gamma$, detecting transitions with the MRA global multiresolution 
method still depends in some sense on imperfect CD solutions provided by our robust, 
but nevertheless greedy, CD algorithm.

More specifically, the MRA algorithm takes advantage of the fact it is significantly
more difficult for a heuristic CD algorithm to navigate the energy landscape 
when competing partitions have comparable energies.
This condition could occur when two or more good partitions exist at a single 
resolution, but it appears to occur more often in transition regions between different 
levels or types of structure.
In terms of CD model parameters, these different levels of multi-scale community
structure have a minimum energy at different resolutions; but in the 
transition region, the minimum for each distinct division is comparable, 
causing different replicas of the CD algorithm to be more easily trapped in unrelated 
areas of the energy landscape.
The MRA algorithm attempts to detect this variation in partitions 
and uses it to identify the best resolutions.

If we consider a perfect CD algorithm which always finds the ground state 
of the cost function regardless of any model parameters or the complexity 
of the problem, most of the variation among partition solutions in the replicas 
mentioned above will disappear (see comment on degeneracies below).
That is, the more perfect the solver, the sharper the detected transition would become 
in terms of any variation in the CD model's resolution parameter(s).
With a perfect solver or an easy CD problem, the transition is essentially discrete.

The primary partition similarity measures of our MRA algorithm, variation 
of information $V$ and the normalized mutual information $U$, would only observe 
strong agreement among the replicas across this transition
(\ie{}, no observable extrema indicating structural shifts or preferred community 
structure), but changes in the secondary measures 
(e.g., Shannon entropy $H$, mutual information $I$, and the number of clusters $q$) 
would still mark the transition, albeit without quantitatively indicating the quality 
of the candidate partitions (the length of the stable region can still be a qualitative
indicator of partition stability).
Thus, a perfect CD algorithm illustrates that the MRA algorithm, as currently employed, 
relies to some extent on the imperfect solutions provided by the underlying CD solver used.
A similar argument would apply to the local multiresolution method discussed in the current work.

When using a perfect CD solver, our global multiresolution method would measure the 
degeneracy of the ground state energy at a given resolution.
One could extend our global multiresolution method by comparing partitions at nearby 
resolutions to better evaluate the stability of the partitions over a range of resolutions.
Under this construction, the MRA algorithm would still be able to evaluate partition stability 
even with a perfect CD solver at the cost of increased computational effort based on time
spent comparing replicas across a range of resolution values.

\section{Information Measures} \label{sec:clusterVI}

Information measures have received broad acceptance for comparing 
candidate CD partitions.
Commonly used measures include the variation of information  
\cite{ref:vi} and normalized mutual information \cite{ref:danon}.
We leveraged the measures in \secref{sec:NMVI} to identify the best 
global network scales via a multiresolution replica method \cite{ref:rzmultires} 
(see \Appref{app:localglobal} and \secref{app:MRA}).

\subsection{Partition correlations} \label{sec:NMVI}

To define VI and NMI, we select a random node from partition $A$ and note 
that it has a probability $n_k/N$ of being in community $k$ where 
$n_k$ is the number of nodes in the community and $N$ is the total 
number of nodes in the system.  The Shannon entropy is
\begin{equation}
  H(A) = -\sum_{a=1}^{q_A} \frac{n_k}{N}\log\frac{n_k}{N}
  \label{eq:shannonentropy}
\end{equation}
where $q_A$ is the number of communities in partition $A$.
The mutual information $I(A,B)$ between two partitions $A$ and $B$ evaluates 
how much we learn about $A$ if we know $B$.
For our application, contending partitions ($A$, $B$, $\ldots$, $Z$)
are independently solved copies of the system.

We define a ``confusion matrix'' for partitions $A$ and $B$ which specifies 
how many nodes $n_{ab}$ in community $a$ of partition $A$ are also in community 
$b$ of partition $B$. Mutual information is 
\begin{equation}
  I(A,B) = \sum_{a=1}^{q_A}\sum_{b=1}^{q_B} \frac{n_{ab}}{N} 
  \log\left(\frac{n_{ab} N}{n_a n_b}\right)  
  \label{eq:mi}
\end{equation}
where $n_a$ ($n_b$) is the number of nodes in community $a$ ($b$) of partition 
$A$ ($B$).  Also, $I(A,B)=0$ when $n_{ab}=0$.
The variation of information $V(A,B)$ metric is then
\begin{equation} 
  V(A,B) = H(A) + H(B) - 2I(A,B)
  \label{eq:vi} 
\end{equation}
which measures the information ``distance'' between partitions $A$ 
and $B$ with a range of $0\le V(A,B)\le\log N$.
We use base $2$ logarithms.

A normalized information measure \cite{ref:danon} 
of partition similarity is
\begin{equation}
  \INAB = \frac{2I(A,B)}{H(A)+H(B)}.
  \label{eq:NMI}
\end{equation}
NMI and VI are closely related,
$\INAB = 1 - V(A,B)/\left[ H(A)+H(B) \right]$.
While NMI is a valuable measure of partition similarity, it is not a formal 
metric (see \Appref{app:semimetricproof}) on partitions $A$ and $B$ in part 
because $\INAA = 1$ not $0$.

Some researchers prefer NMI as a normalized measure where it maps uncorrelated 
partitions to $0$ and perfectly correlated partitions to $1$.
Others prefer the stricter metric properties of VI, but VI can also be normalized,
if desired.
In most cases, the two measures provide the same information, 
but occasionally distinctions between them can be observed.
For example, we previously showed \cite{ref:rzlocal} that maxima in VI and 
minima in NMI mark structural transitions in partitions, that provides information 
about the average intercommunity edge density, but only VI clearly shows an extremum 
before the system collapses into disjoint communities as $\gamma$ is lowered
in \eqnref{eq:ourmodel}.

\subsection{Local information analogies} \label{sec:CNMVI}

When defining a cluster comparison measure, we wish to maintain consistency 
with the trend in CD towards information-theoretic partition evaluations.
Toward this end, we consider the cluster embedded in the full system of $N$ 
nodes (see also later comments), where $N$ is the number of nodes in the network.
This comparison gives our measure a context for the resulting cluster-level 
entropy or information calculations based on associated partition-of-unity 
probabilities.

\begin{figure}
\centering
\includegraphics[width=0.975\columnwidth]{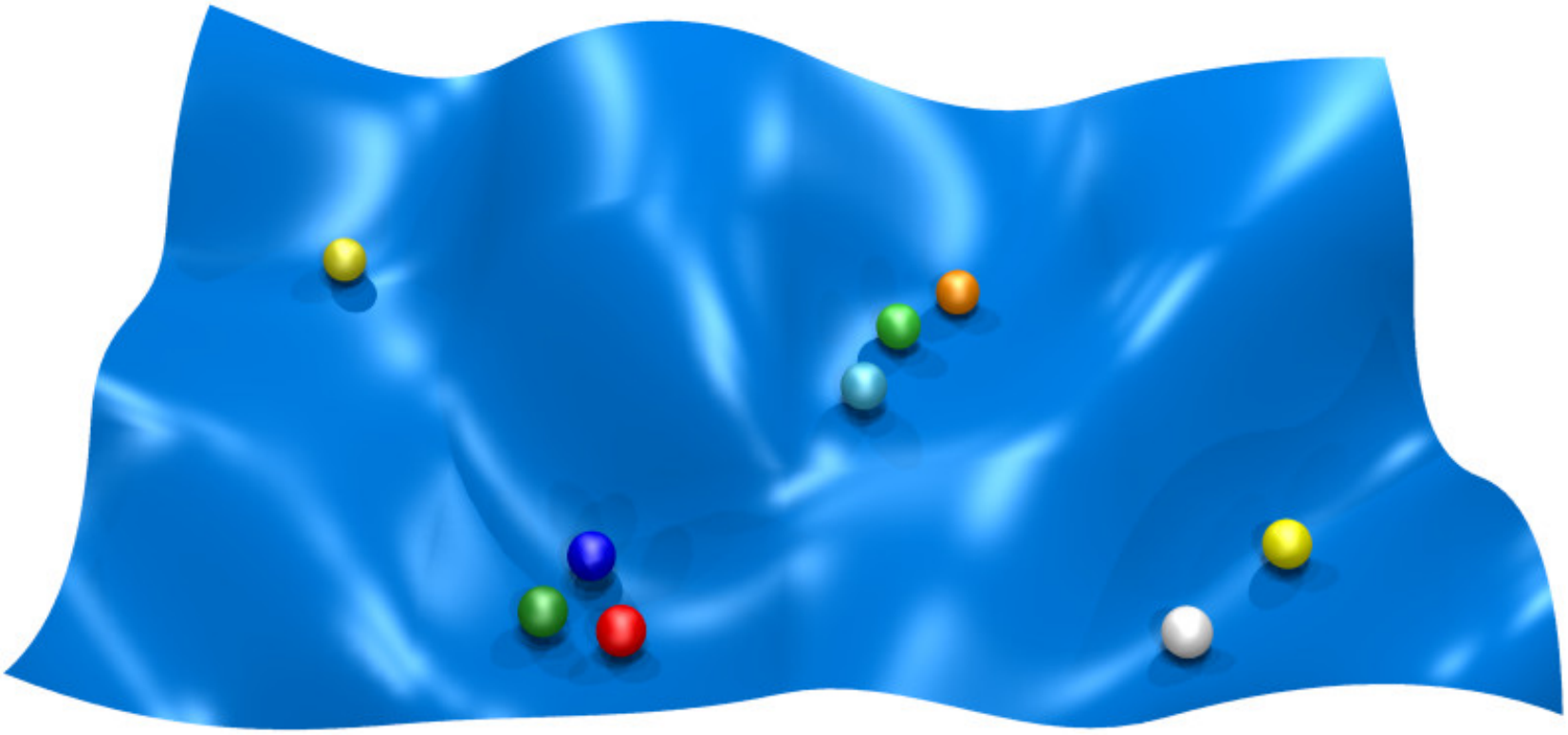}
\caption{(Color online) The figure schematically depicts $r$ independent
solvers (``replicas'') as spheres navigating the energy landscape of
\eqnref{eq:ourmodel}.
Stronger agreement among the replicas, as measured by information correlations
in \secref{sec:NMVI},
indicates a more stable, well-defined global partition (see text).
In this manuscript, we demonstrate that local communities may be strongly defined 
even if all the communities in the global system are weakly correlated 
(see \figref{fig:LMRA}).}
\label{fig:MRAlandscape}
\end{figure}

From \eqnref{eq:shannonentropy}, the entropy contribution of community 
$a$ in partition $A$ is
\begin{equation}
  H_a(A)\equiv -\frac{n_a}{N}\log\left(\frac{n_a}{N}\right)
  \label{eq:Ha}
\end{equation}
where $n_a$ is the number of nodes in community $a$.
Similarly, \eqnref{eq:mi} indicates the mutual information contribution 
when comparing cluster $a$ in partition $A$, $(a,A)$, to cluster $b$ in partition 
$B$, $(b,B)$,
\begin{equation}
  I_{ab}(A,B)\equiv \frac{n_{ab}}{N} \log\left(\frac{n_{ab} N}{n_a n_b}\right).
  \label{eq:IabAB}
\end{equation}
In analogy with \eqnref{eq:vi}, we introduce the \emph{cluster variation 
of information} (CVI) $\vab$
\begin{equation} 
  \vab\equiv H_a(A) + H_b(B) - 2\Inab.
  \label{eq:CVI}
\end{equation}
CVI exhibits appealing ``distance-like'' properties of a semi-metric 
for comparing clusters $(a,A)$ and $(b,B)$ (see \Appref{app:semimetricproof}
for a trivial proof).
Summing over all pairs of clusters $a$ and $b$, VI is related to CVI by 
\begin{equation}
  V(A,B) = \sum_{a}^{q_A}\sum_{b}^{q_B} \vab - (q_B-1)H(A) - (q_A-1)H(B).
  \label{eq:vitoviab}
\end{equation}
\Appref{app:altCVI} provides additional remarks.

From \eqnref{eq:NMI}, we introduce the natural \emph{cluster normalized 
mutual information} (CNMI) analogy
\begin{equation}
  \Inab\equiv \frac{ 2 n_{ab}\log\left(\frac{n_{ab} N}{n_a n_b}\right)}
                            { n_a\log\left(\frac{N}{n_a}\right) 
                             +n_b\log\left(\frac{N}{n_b}\right) }.
  \label{eq:CNMI}
\end{equation}
While CNMI is not a metric
[in part because $\Inaa=1$ not $0$], it has the same intuitive property 
of cluster similarity that makes NMI attractive for partition comparisons.
\Eqnref{eq:CNMI} is essentially a normalized variant of CVI, 
$\Inab = 1 - \vab /\left[ H_a + H_b \right]$.
On smaller networks, CVI provides a clearer picture of transitions
with its distance-like semi-metric properties, but CNMI is more easily 
evaluated for larger networks because variations in CVI become small as 
$N$ becomes large.

If we were to compare larger (multi-cluster) sub-graphs, a natural 
approach is to cut the subgraph from the whole network and compare the 
reduced-size partitions. 
This breaks down at the cluster level because there is no partition-of-unity
associated with an individual cluster as used to define NMI \cite{ref:danon}
or VI \cite{ref:vi} for community detection.
Implementing an arbitrary measure for clusters is difficult, 
so we chose to consider the cluster comparisons in the context of a larger
network of nodes.
Strictly speaking we do not need to use the true size of the network for our 
cluster comparisons.
Rather, we could use some other $N'\neq N$, but it is conceptually appealing
to evaluate a cluster in the context of the full network.

\begin{figure*}[t]
\centering
\includegraphics[width=1.4\columnwidth]{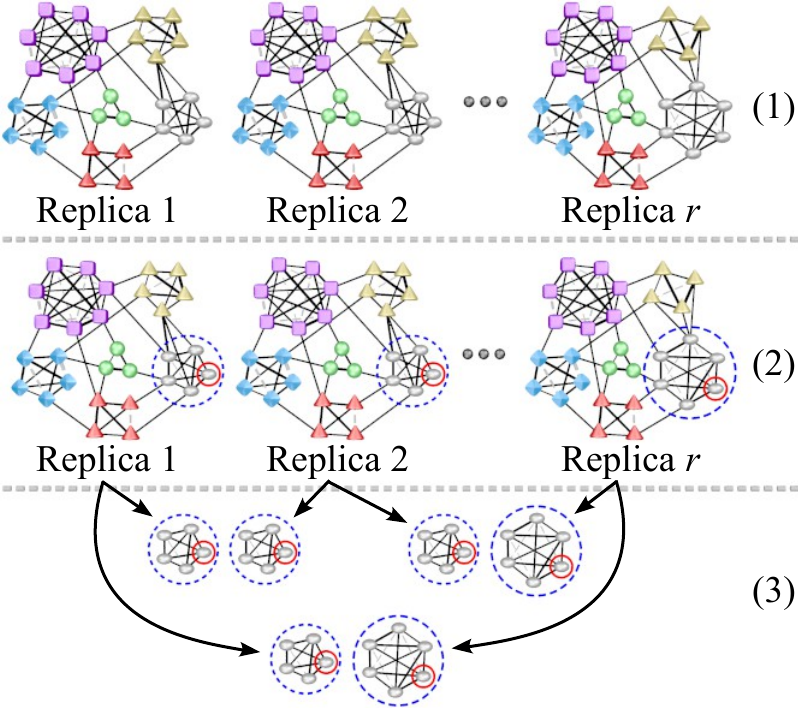}
\caption{(Color online) The figure illustrates our local multiresolution
algorithm discussed in detail in \secref{sec:LMRA}.
The graphs include ferromagnetic [``cooperative'' with $w_{ij}>0$ 
in \eqnref{eq:ourmodel}] relations depicted by solid, black lines and 
antiferromagnetic (``neutral'' or ``adversarial'' with $u_{ij}>0$) 
interactions depicted by gray, dashed lines.
The line thickness indicates the relative interaction strength, and we 
omit intercommunity adversarial and neutral relations for clarity.
The following steps are completed for each $\gamma$ used in \eqnref{eq:ourmodel}:
In step ($1$), we independently solve a series of $r$ ``replicas'' of the 
community detection problem (although we could, in general, improve the 
efficiency by solving only the local communities embedded in the network).
Step ($2$) identifies the target node(s) of interest (solid red circles) 
and their corresponding parent clusters (blue dashed circles).
Step ($3$) uses \eqnrefs{eq:CVI}{eq:CNMI} to calculate correlations among 
all pairs of parent clusters in order to determine the community robustness 
at the current resolution specified by $\gamma$ in \eqnref{eq:ourmodel}.}
\label{fig:LMRA}
\end{figure*}

\section{Local multiresolution algorithm} \label{sec:LMRA}

Our local multiresolution algorithm identifies relevant local 
multiresolution order, meaning well-defined local communities.
We invoke $\vab$ in \eqnref{eq:CVI} and $\Inab$ in \eqnref{eq:CNMI} 
to compare local clusters $a$ and $b$ across $r$ ``replicas'' 
(independent solutions).
\Figref{fig:MRAlandscape} depicts solutions with the global MRA \cite{ref:rzmultires} 
algorithm given in \secref{app:MRA}.
The LMRA method depicted in \figref{fig:LMRA} extends the MRA method 
by incorporating comparisons between specific clusters.

\subsection{Community detection algorithm} \label{app:CDalgorithm}

We begin by introducing our greedy CD algorithm which dynamically moves 
nodes into the community that best lowers the local energy according 
to \eqnref{eq:ourmodel} given the current state of the system $\{\sigma_i\}$.
The process iterates through the nodes until no further nodes are available.
Typically, $O(10)$ iteration cycles through all $N$ nodes are required 
except in rare instances that lie in or near the ``hard'' (or ``glassy'') 
phase \cite{ref:rzlocal,ref:huCDPTsgd,ref:huCDPTlong}.\\

The CD steps are:

(0) \emph{Initialize the system}. 
Initialize the connection matrix $A_{ij}$ and edge weights $w_{ij}$ and 
$u_{ij}$.  Determine the number of optimization trials $t$.

(1) \emph{Initialize the clusters}. 
The initial partition is usually a ``symmetric'' state wherein each node 
is the lone member of its own community (\ie{}, $q_0=N$).

(2) \emph{Optimize the node memberships}. 
Sequentially select each node, traverse its neighbor list,
and calculate the energy change that would result if it were moved 
into each connected cluster (or an empty cluster).
Immediately move it to the community which best lowers the energy 
(optionally allowing zero energy changes).

(3) \emph{Iterate until convergence}. 
Repeat \step{2} until a (perhaps local) energy minimum is reached 
where no nodes can move. 

(4) \emph{Test for a local energy minimum}. 
Merge any connected communities if the combination lowers the summed 
community energies. 
If any merges occur, return to \step{2} and attempt additional node-level 
refinements.
     
(5) \emph{Repeat for several trials}. 
Repeat \stepd{1}{4} for $t$ independent trials and select the 
lowest energy result as the best solution. 
By a trial, we refer to a copy of the network in which the 
initial system is randomized in a symmetric state with a different
node order.\\

The optimal $q$ is usually dynamically determined by the lowest energy state
although the algorithm can also fix $q$ during the dynamics.
Empirically, the computational effort scales as \bigO{tL^{1.3}\log k} 
where $k$ is the average node degree and $\log k$ is from a binary search 
implemented on large sparse matrix systems.
This greedy variant can accurately scale to at least \bigO{10^9} edges 
\cite{ref:rzlocal}.
We can extend it with a stochastic heat bath \cite{ref:huCDPTsgd} solver 
or a simulated annealing algorithm \cite{ref:smcd} at the cost of significantly 
increased computational effort, but the greedy variant performs exceptionally
well on many systems.

\subsection{Global multiresolution algorithm} \label{app:MRA}


In order to set the stage for introducing our local multiresolution algorithm,
we first discuss the global algorithm and some of its features.
As depicted in \figref{fig:MRAlandscape}, our multiresolution algorithm 
iteratively applies the CD algorithm in \secref{app:CDalgorithm} 
to quantitatively evaluate the most stable community partitions over 
a range of network scales.
After convergence, these replicas sample the local energy minima of the 
energy landscape, giving an estimate of its associated complexity.

In its basic form, the global MRA algorithm iteratively solves the CD problem 
for a graph over a range of $\gamma$ in \eqnref{eq:ourmodel} and evaluates 
the average strength of the partition correlations to find more stable partitions.
This process quantitatively estimates the robustness of the identified 
partitions by sampling the complexity of the energy landscape.
Previous work by the authors \cite{ref:rzmultires} on the Lancichinetti-Fortunato-Kert{\'e}sz
\cite{ref:lancbenchmark} as well as other synthetic and real benchmarks \cite{ref:rzlocal} 
show that these strongly correlated regions regularly correspond to the known, 
accurate partitions.

Generally speaking, poorer correlations occur when there are contending 
partitions of comparable strength [\ie{}, the energy difference of the applied 
cost function is near zero], the resolution is inside a glassy phase 
(extraneous intercommunity edges obscure the dynamic process of locating the 
best solution), or the graph is more random in nature.  
In the case of contending partitions, local multiresolution methods, such as 
the one presented in the current work, may be able to reliably extract the 
well-defined communities.

We quantify the partition correlations using information theoretic (or other 
appropriate) measures (see \secref{sec:NMVI}).
If most or all solvers (replicas) agree on the best solution, then we rate 
the partition as strongly correlated, but if the partitions have large 
variations, we say the solution is weak.
In either case, we select the lowest energy replica solution to represent 
the best answer at a given resolution $\gamma_i$, but one could also construct 
a consensus partition \cite{ref:LPA,ref:fredjain,ref:topchyconsensus}, 
particularly in the latter case of weak solutions \cite{ref:topchyweak}.

As a function of the resolution parameter $\gamma$ in \eqnref{eq:ourmodel} 
(or any relevant CD scale parameter for another model \cite{ref:smcd,ref:arenasmultires}), 
the best resolutions may be identified by peaks or plateaus in NMI \cite{ref:rzmultires}, 
minima or plateuas in VI \cite{ref:rzmultires,ref:fenndynamic}, and/or plateaus 
in the number of clusters $q$ \cite{ref:arenasmultires} or other measures 
\cite{ref:rzmultires,ref:fenndynamic}. 
Plateaus in these measures (\ie{}, NMI, VI, $H$, $q$, etc.) as a function of $\gamma$
imply more stable features of the network,
although caution must be exercised when interpreting some measures \cite{ref:rzmultires}.
Sharper peaks in NMI or narrow troughs in VI indicate strongly defined but more 
transient features.
Significant peaks in VI or troughs in NMI generally indicate transitions between 
dominant structures.
More generally, we can further extract pertinent details of the network from 
other extrema in NMI and VI (e.g., Ref.\ \cite{ref:HRNimages} also analyzed 
peaks in VI to perform image segmentation using CD concepts).

Correlations among the replicas evaluate the level of agreement on stable, low-energy 
solutions. 
If a problem has two equally viable partitions (\ie{}, with the same energy) that 
are located by two replicas, the one with the highest overlap among all other replicas 
would be preferred.
This corresponds to the volume of a configuration space basin associated with this 
preferred partition where 
the other partition with the same energy has an associated smaller basin size 
(the number of states or exponentiated entropy).
This is like an ``order by disorder'' effect 
\cite{ref:villain1980JPF,ref:shender1982SP,ref:henley1989PRL,ref:nussinov2004EPL} 
present in various systems (e.g., entropic contributions
to the free energy in finite temperature systems) which lifts the degeneracy 
between equally viable partitions and favors one partition (or a subset of partitions) 
over the others.\\

The MRA algorithm is:

(0) \emph{Initialize the algorithm}.  
Select the number of independent replicas $r$.
Identify the set of resolutions $\{\gamma_i\}$ to analyze using \eqnref{eq:ourmodel}
along with a starting $\gamma_0$.
It is often convenient to begin at high gamma and step downward, stopping
if the system completely collapses.

(1) \emph{Initialize the system}. 
For the current $\gamma_i$, initialize each replica with a unique 
set of $N$ spin indices (\ie{}, $q_{0}^{(j)} = N$ for each replica $j$).

(2) \emph{Solve each replica}.
Independently solve each replica according to the CD algorithm 
in \secref{app:CDalgorithm}.

(3) \emph{Compare all replicas}. 
Calculate the Shannon entropy for every replica and compare all pairs 
of replicas using the mutual information $I(A,B)$, normalized mutual 
information $\INAB$, and variation 
of information $V(A,B)$ measures in \secref{sec:NMVI}.

(4) \emph{Iterate to the next resolution}. 
Increment to the next resolution $\gamma_{i+1}$. 
A geometric step size $\Delta\gamma = 10^{1/s}$ is often convenient where 
$s\approx O(10)$ is an integer number of $\gamma_i$'s per decade of $\gamma$.
Repeat \stepd{1}{3} until the system is fully collapsed (if stepping down 
in $\gamma_i$) or no $\gamma_i$'s remain.\\

The information correlations in steps (3) and (4) allow the determination 
of the best global network scale(s) \cite{ref:rzmultires} 
(see \Appref{app:localglobal}) based upon regions of $\gamma$ with high NMI 
or low VI.
Plateaus in $H$, $I$, and $q$ may also provide supplemental information
regarding partition stability.
The solution cost scales linearly in $r$ with the CD algorithm 
in secref{app:CDalgorithm}, \bigO{rtL^{1.3}\log k}.
We have solved systems with \bigO{10^7} edges on a single processor 
\cite{ref:rzmultires} in a few hours.

The algorithm may detect, but does not impose, a hierarchical 
community structure.
That is, as shown in \secref{sec:staggered}, the MRA algorithm will show strongly
correlated regions at the well-defined hierarchical levels, but it is also able 
to analyze non-hierarchical multiresolution structure. 
This approach is preferable to \emph{forcing} a hierarchical structure 
on every analyzed network \cite{ref:everittHC} since some networks may not naturally 
possess this type of organization.
Once the preferred resolutions are identified, the specific hierarchical nature 
can be analyzed and evaluated by other means 
\cite{ref:trusinahiermeasure,ref:moneshiermeasure}.

\subsection{LMRA replica method} \label{sec:LMRAalgorithm}

Clusters naturally change as the resolution is varied, so 
how do we identify the appropriate target clusters for comparison? 
Two immediate approaches include: ($i$) compare clusters for ``nearby'' resolutions 
as specified by a particular $\gamma$ in \eqnref{eq:ourmodel} 
or ($ii$) compare targeted parent clusters for specific node(s) of interest 
across the replicas. 
Another natural approach is to mix the above methods: ($iii$) select a node 
of interest but further compare the target clusters at neighboring resolutions.
This is an extension of case ($ii$), so we focus on former cases, 
leaving more complicated implementations to future work.

In case ($i$), if one deviates too far from $\gamma_i$, the cluster 
will change substantially and the evaluation will be less useful.
That is, at some point, the cluster changes enough that it is no longer the 
``same'' community. 
We could quantitatively define this comparison based on the relevant CVI values, 
but the practical question of identifying the appropriate community 
for comparison across all partitions becomes increasingly difficult.

In case ($ii$), the node may be selected \emph{a priori} based on a 
known identity in the real network, or it may be randomly selected.
(One may also first analyze the global system and use any communities 
with interesting features to identify important nodes.)
This option has two advantages: it is simpler to implement, but more importantly,
the studied clusters are always well-defined, enabling comparisons of 
community robustness across all relevant resolutions.
That is, at a given $\gamma_i$, we only need to know to which cluster 
node $i$ belongs, regardless of any structural changes in its network
neighborhood as $\gamma$ is varied.
Cluster correlations are quantitatively evaluated at a given $\gamma_i$, 
but the average $\voab$ or $\Ioab$ values over the replica pairs can still 
be compared across different $\gamma_i$'s to evaluate the relative strength 
of the parent communities.

Option ($ii$) is used in the current work.
We select a \emph{node} of interest (e.g., a person in a terrorist network, 
see \secref{sec:terrornetwork}), 
and trace the parent clusters among the replicas across a range of network 
scales [\ie{}, different $\gamma_i$'s in \eqnref{eq:ourmodel}].
As depicted in \figref{fig:LMRA}, the LMRA algorithm is:

(0) \emph{Initialize the algorithm}. 
Select the number of replicas $r$ and the number of independent optimization 
trials $t$ per replica (see \secrefs{app:CDalgorithm}{app:MRA}).
Select a set of nodes $\{\eta\}$ to track based on problem parameters (e.g., 
a person of interest). 
Identify the set of resolutions $\{\gamma_i\}$ to analyze (often selected 
to sample all relevant network scales, see step $4$ in \secref{app:MRA})
by minimizing \eqnref{eq:ourmodel}.
Select a starting $\gamma_0$.

(1) \emph{Solve $r$ independent replicas}.
For the current $\gamma_i$ in \eqnref{eq:ourmodel}, apply \stepd{1}{3} 
of the global MRA algorithm in \secref{app:MRA}.

(2) \emph{Identify parent clusters}. 
Identify the parent cluster $a_{ij}$ corresponding to each target node 
$\eta$ at the current $\gamma_i$ in each replica $j$.

(3) \emph{Compare clusters}. 
For each parent cluster $a_{ij}$, calculate CVI $v\left(a_{ij},a_{ik}\right)$ 
in \eqnref{eq:CVI} and CNMI $u\left(a_{ij},a_{ik}\right)$ 
in \eqnref{eq:CNMI} with the corresponding parent cluster $a_{ik}$ in replica $k$.
Calculate the average of measure $S_i$ [generically referring to measures 
$\vab$, $\Inab$, etc.{} in \secref{sec:CNMVI}] over all replica pairs 
at $\gamma_i$ by 
\begin{equation}
  \overline{S_i}(a,b)= \frac{2}{r(r-1)}\sum_{k>j} S_{ijk}(a,b),
  \label{eq:CNMVIaverage}
\end{equation}
where $i$ refers to a particular resolution parameter index for $\gamma_i$ 
in \eqnref{eq:ourmodel}, and $j$ and $k$ refer to replica summations.

(4) \emph{Identify the best resolutions}. 
For each parent cluster $a_{ij}$, find the lowest CVI values $\vab$ 
or the highest CNMI values $\Inab$ 
and their corresponding resolution(s) $\{\gamma^\mathrm{Best}_i\}\subset 
\{\gamma_i\}$.
These are the best resolutions for each cluster $a_{ij}$.\\

Step (4) is the final result for the algorithm, indicating which 
resolutions and candidate partitions are mostly likely to be useful.
As with the global MRA approach in \secref{app:MRA}, we are interested
in extrema or plateaus in the pertinent measures in \secref{sec:clusterVI}.
Empirically, $r\approx O(10)$ or less appears to be sufficient for most problems.
We estimate the cost to be \bigO{L r^2} which is comparable to the base
MRA algorithm cost in \secref{app:MRA}.

\begin{figure}
\centering
\includegraphics[width=0.9\columnwidth]{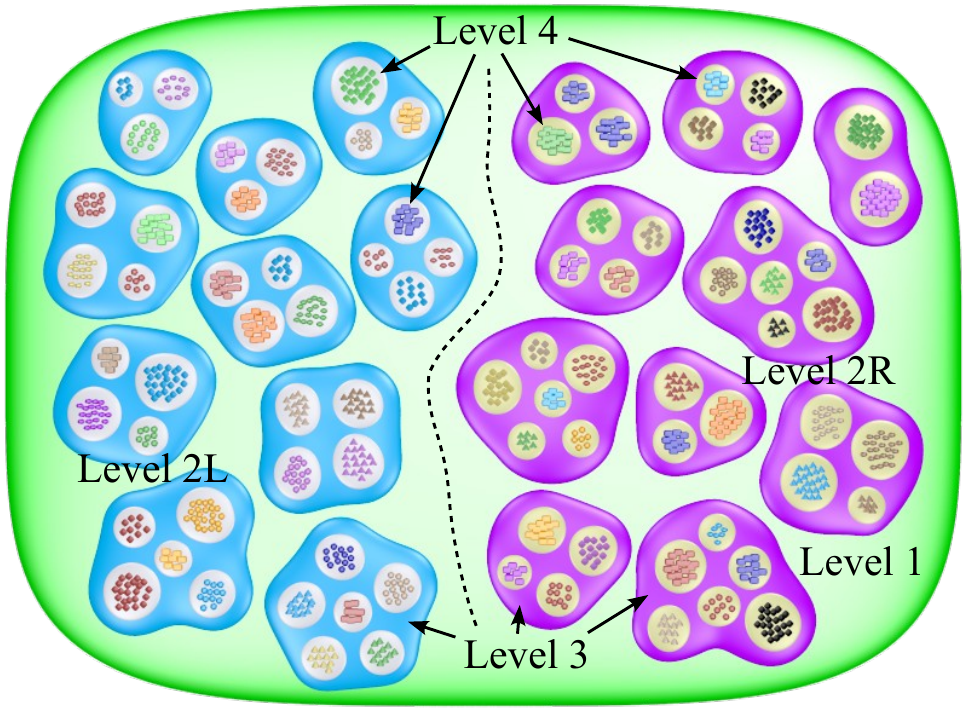}
\caption{(Color online) The figure depicts a constructed $N=1024$ node 
four-level hierarchy. 
Level $1$ is the complete network with two sides of supercommunities 
that are randomly connected at a low edge density between them.
Level $2$ consists of two roughly equal sized branches ($N_L = 502$ and 
$N_R = 522$) which we denote by left (L, blue or darker tone) and 
right (R, purple or medium tone) as the figure indicates.
Level $3$ is the set of supercommunities, and level $4$ is the set 
of smallest communities strictly contained within the supercommunities.
At levels $3$ and $4$, elements of the left branch are connected at higher 
internal and intercommunity edge densities than the corresponding right 
branch elements.
See the text for a more detailed description of the network.
This construction results in a more ``blurred'' global multiresolution 
signature in \subfigref{fig:staggeredplotglobal}{a} where level $4$L 
is lost in the global MRA plot at feature ($iv$).
The corresponding LMRA plot for node $951$ in \subfigref{fig:staggeredplotlocal}{c} 
is nevertheless able to clearly identify level $4$L as a strongly defined 
resolution.}
\label{fig:staggeredhierarchy}
\end{figure}

\begin{figure*}
\centering
\subfigure[\ Global MRA: Branch Hierarchy -- Combined]
  {\includegraphics[width=\widesubfigsize]{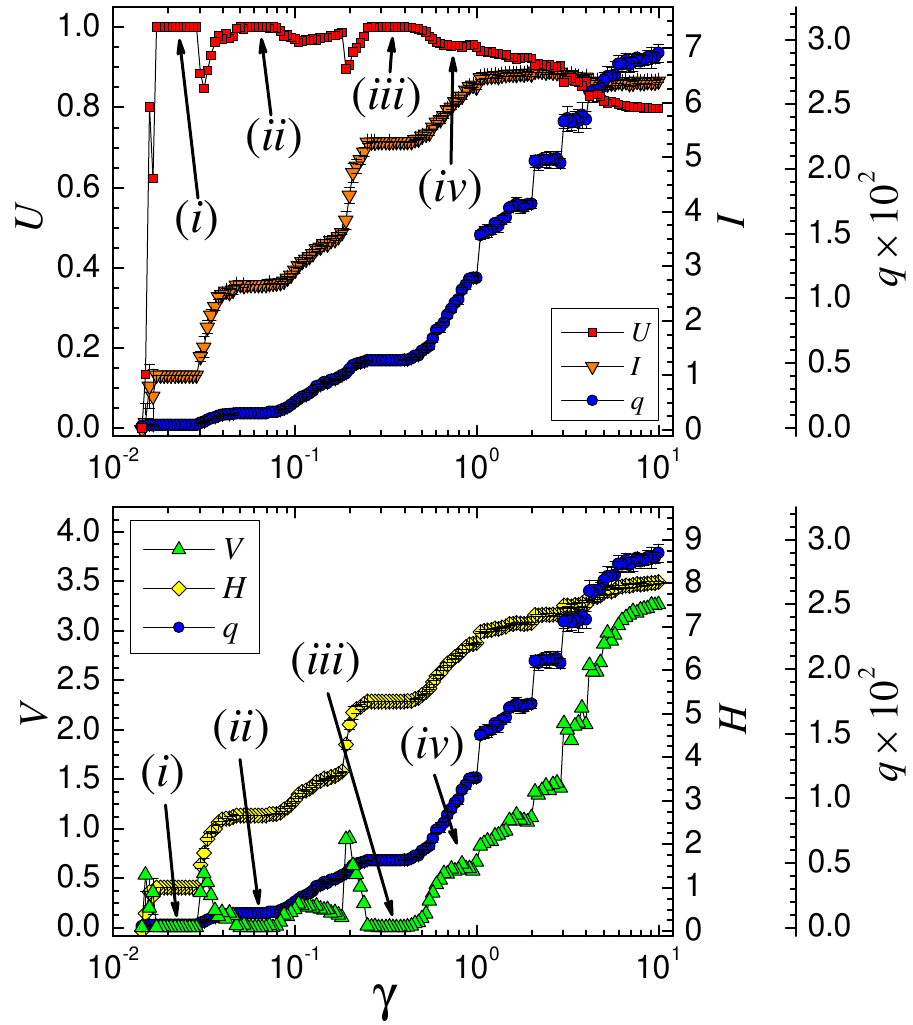}}
\subfigure[\ Global MRA: Branch Hierarchy -- Left Side]
  {\includegraphics[width=\widesubfigsize]{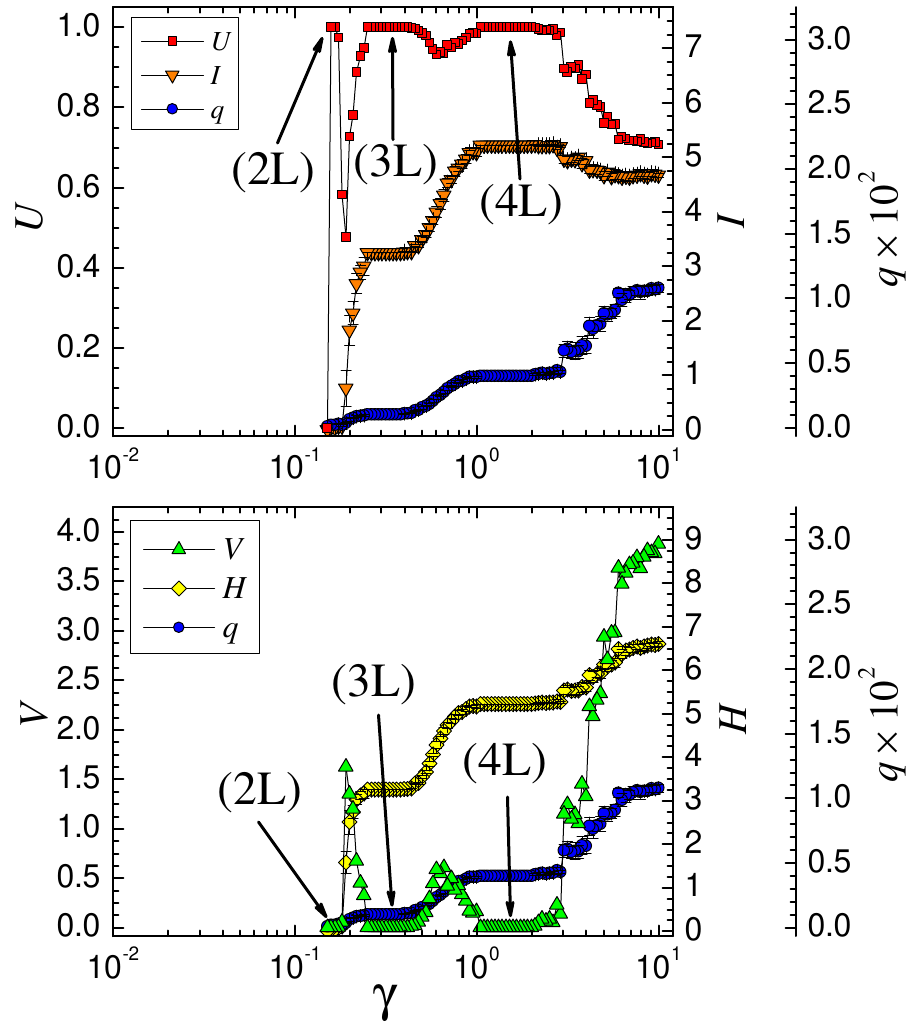}}
\subfigure[\ Global MRA: Branch Hierarchy -- Right Side]
  {\includegraphics[width=\widesubfigsize]{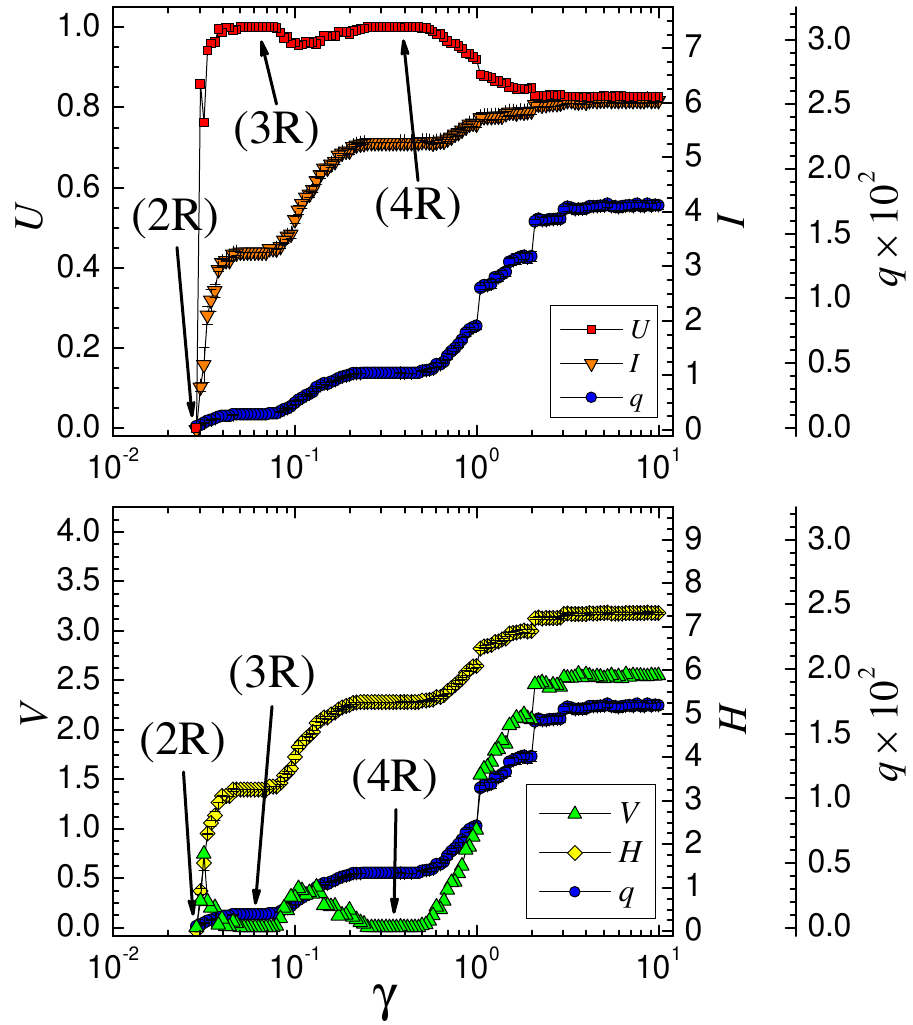}}
\caption{(Color online) In panels (a), we apply our global multiresolution 
algorithm (MRA, see \Appref{app:localglobal} and \secref{app:MRA}) to the $N=1024$ node, 
four-level, branched hierarchy depicted in \figref{fig:staggeredhierarchy}.
Panels (b) and (c) show the MRA method applied separately to the left and 
right level $2$ hierarchy branches, respectively.
In the top sub-panels (a--c), we compare replica \emph{partitions} using 
normalized mutual information $\INN$ (left axes, see \secref{sec:NMVI}) 
and  mutual information $I$ (right axes).
In the corresponding bottom sub-panels, we plot variation of information $V$ 
(left axes) and the Shannon entropy $H$ (right axes).
We also plot the average number of communities $q$ (offset right axes) in top 
and bottom sub-panels.
Features ($i$)--($iii$) demonstrate that the global MRA algorithm can detect 
network-wide stable partitions \cite{ref:rzmultires}.
Feature ($iv$) in panel (a) shows that the level $4$ community 
structure on the left side, known to be present at feature ($4$L) in panel (b), 
\emph{is almost completely obscured} because the right branch is significantly 
more random at the same network scale [\ie{}, value of $\gamma$ in \eqnref{eq:ourmodel}, 
see also \secref{sec:resolution}].
In \figref{fig:staggeredplotlocal}, we compare parent communities using 
the local multiresolution algorithm in \secref{sec:LMRA} where we demonstrate 
that the method can accurately extract level $4$L for the targeted nodes.}
\label{fig:staggeredplotglobal}
\end{figure*}

\begin{figure*}
\centering
\subfigure[\ Local MRA: Node 116 -- Member of right side]
  {\includegraphics[width=\widesubfigsize]{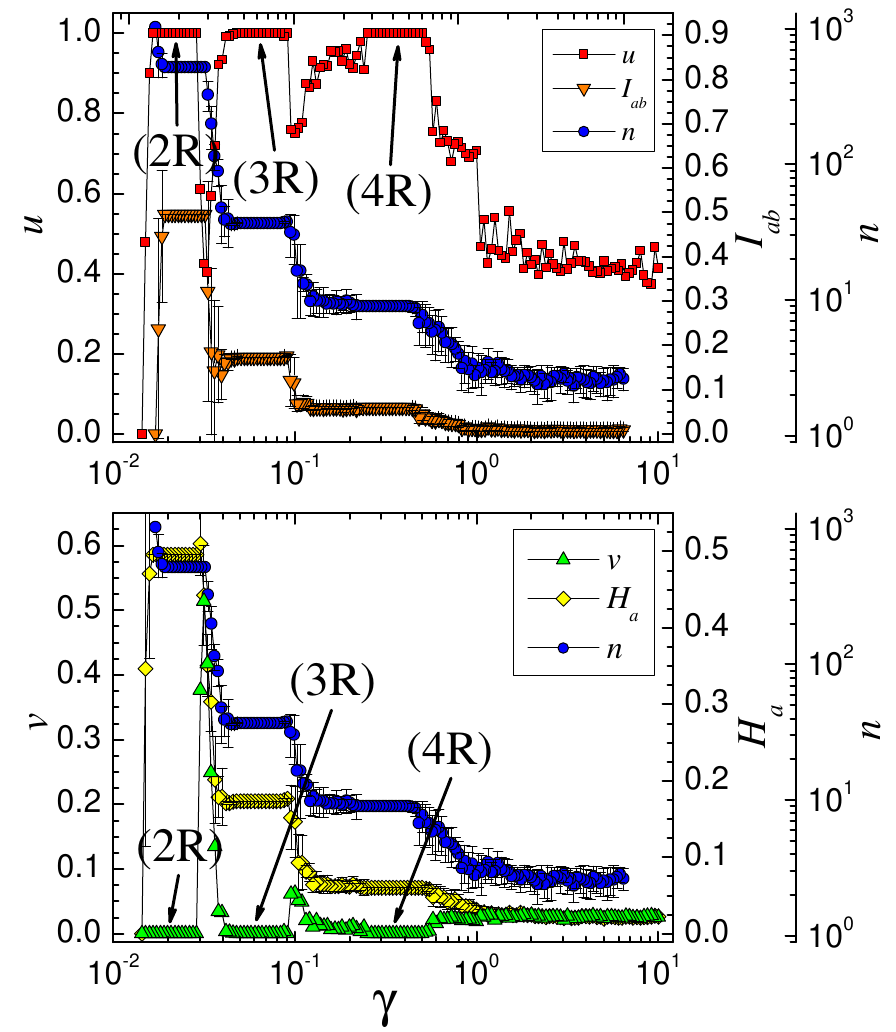}}
\subfigure[\ Local MRA: Node 661 -- Member of left side]
  {\includegraphics[width=\widesubfigsize]{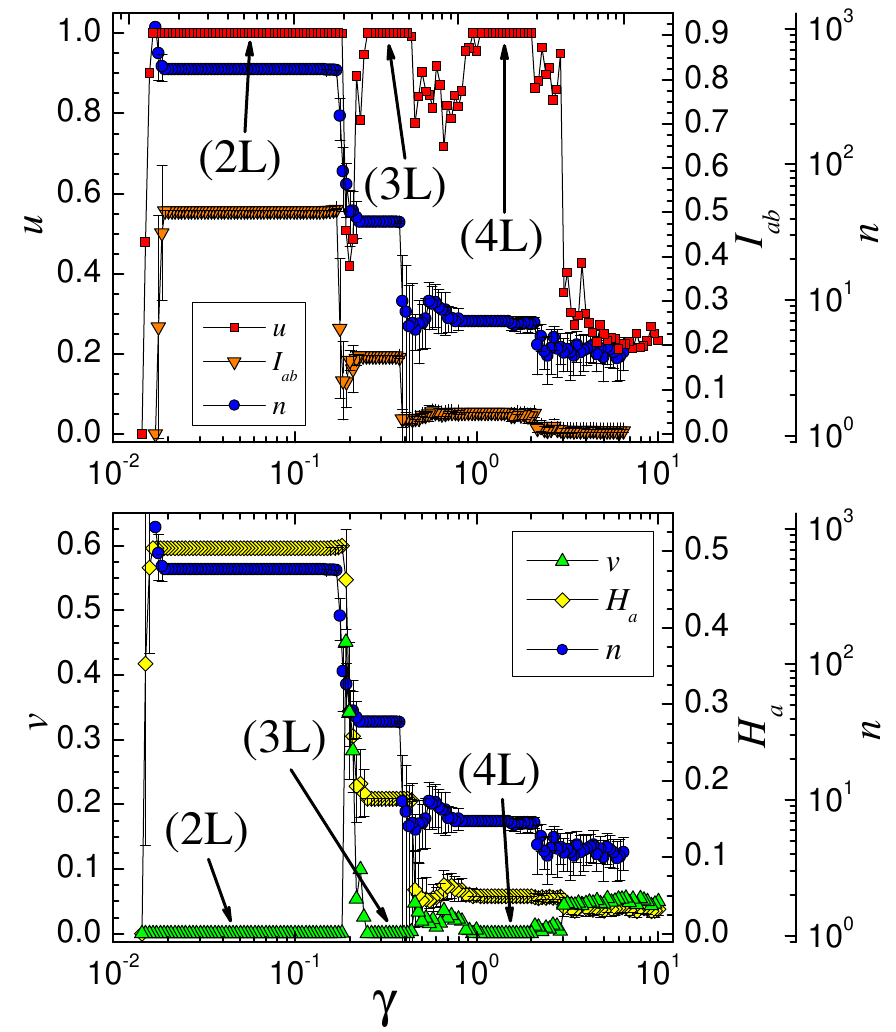}}
\subfigure[\ Local MRA: Node 951 -- Member of left side]
  {\includegraphics[width=\widesubfigsize]{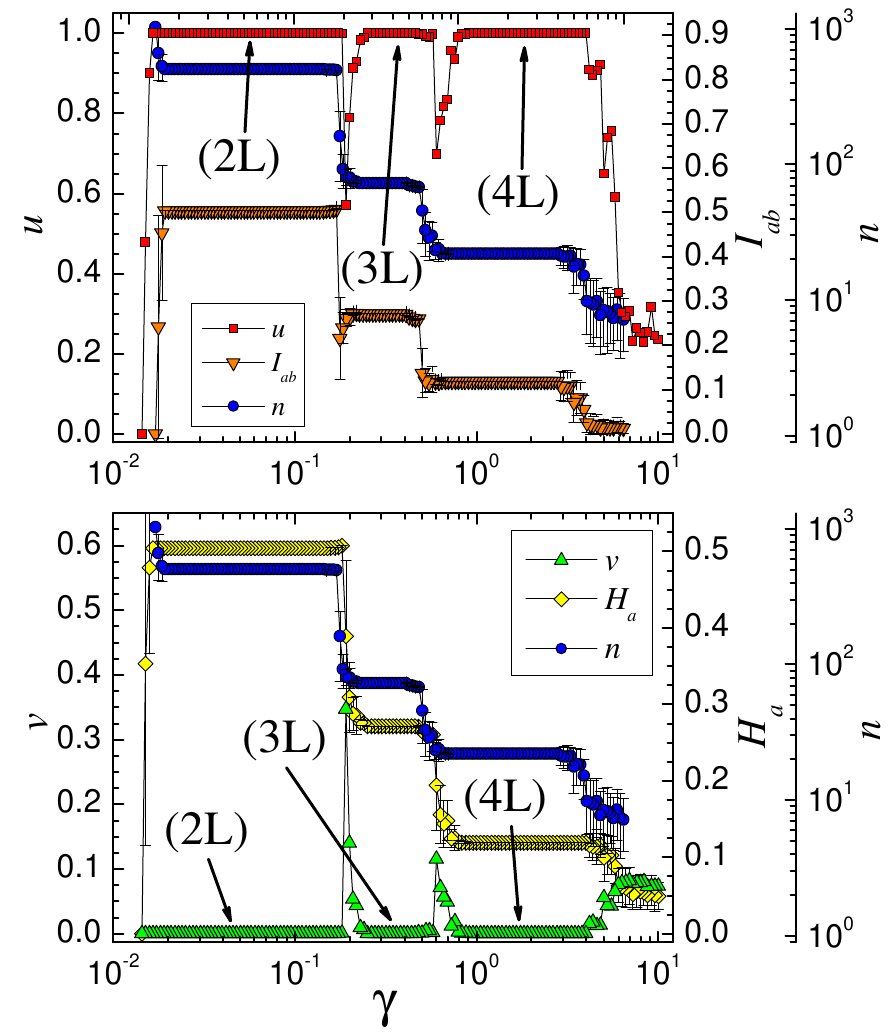}}
\caption{(Color online) In panels (a--c), we apply our local multiresolution 
algorithm (LMRA) in \secref{sec:LMRA} to randomly selected nodes of the branched 
hierarchy depicted in \figref{fig:staggeredhierarchy}.
The top sub-panels compare targeted \emph{communities} in the solved replicas 
(independent solutions) using the cluster normalized mutual information 
$\Inab$ (left axes, see \secref{sec:CNMVI}) and the mutual information 
contribution $I_{a b}$.
The corresponding bottom sub-panels plot the cluster variation of information 
$\vab$ (left axes) and the Shannon entropy contribution $H_a$ (right axes).
Both top and bottom sub-panels also plot the average number of nodes $n$ in the 
respective parent communities on the offset right axes.
The LRMA method is easily able to extract the relevant levels $3$ and $4$ for 
the target nodes as evidenced by regions of low CVI (or high CNMI) even though 
level $4$L of the hierarchy is almost completely obscured at feature ($iv$) 
in the combined global MRA plot in \subfigref{fig:staggeredplotglobal}{a}.}
\label{fig:staggeredplotlocal}
\end{figure*}

\subsection{Alternative implementations} \label{sec:alternateLMRAs}

In the current work, we contrast local, community-level analysis with 
global multiresolution correlations.
Thus, in this algorithm, we solve for all communities in the full system 
and then select the appropriate parent clusters for the community-level 
analysis.
Since the only global parameter that we need to evaluate CVI or CNMI is 
the system size $N$, a more efficient approach could take advantage of our 
local cost function in \eqnref{eq:ourmodel}
(see also Ref.\ \cite{ref:havemannlocalmultires} for a more efficient 
method applied a different fitness function \cite{ref:lanc}).
Specifically, we would solve for the target communities around a particular
node of interest $a_i$ by examining community membership opportunities 
strictly for the neighbors of nodes in or connected to $a_i$'s local 
neighborhood.
The remainder of the graph partition need not be specified in detail
to apply \eqnrefs{eq:CVI}{eq:CNMI}.

A more comprehensive alternative to \step{3} could be useful if there are 
no \emph{a priori} nodes of interest to study.
That is, we could compare all pairs of clusters and identify the best 
matching cluster $b_{ik}$ for $a_{ij}$ based on the minimum
$\vsupab{(jk)}$ at the current $\gamma_i$.
Then we would average CVI over all cluster matches for each best-cluster 
pair.
In this scenario, we could further pursue the relative cluster comparisons 
among the replicas by evaluating whether the best clusters match among
themselves. 
That is, we would determine if $b_{ik}$ of partition $A$ also matches the 
parent cluster $d_{il}$ in partition $B$, repeating the process to the 
desired depth.

With this alternative to \step{3}, individual community matches among the 
$r$ replicas (see \figref{fig:LMRA}) are not necessarily symmetric.
That is, while \eqnref{eq:CVI} is symmetric in $(a,A)$ and $(b,B)$, this 
does not require that the best matching clusters in the respective partitions 
necessarily agree.
Consequently, it would provide an additional measure of community robustness 
based on the level of mutual agreement (number of agreed matches compared 
to the total possible matches among all replicas).

\section{Examples} \label{sec:examples}

As discussed in \secref{app:MRA}, we calculate the global MRA algorithm
for the network and concurrently apply the LMRA algorithm in \secref{sec:LMRA} 
to targeted nodes by tracking the respective parent clusters across a full range 
of network scales.
Comparing explicit values of VI and CVI is difficult, so we evaluate relative 
values of VI or CVI for a given network.
We demonstrate the LMRA method with a constructed network example and a small, real
terror network.

\begin{figure*}
\centering
\subfigure[\ Terrorist network -- $\gamma=0.1$]
  {\includegraphics[width=1.0\columnwidth]{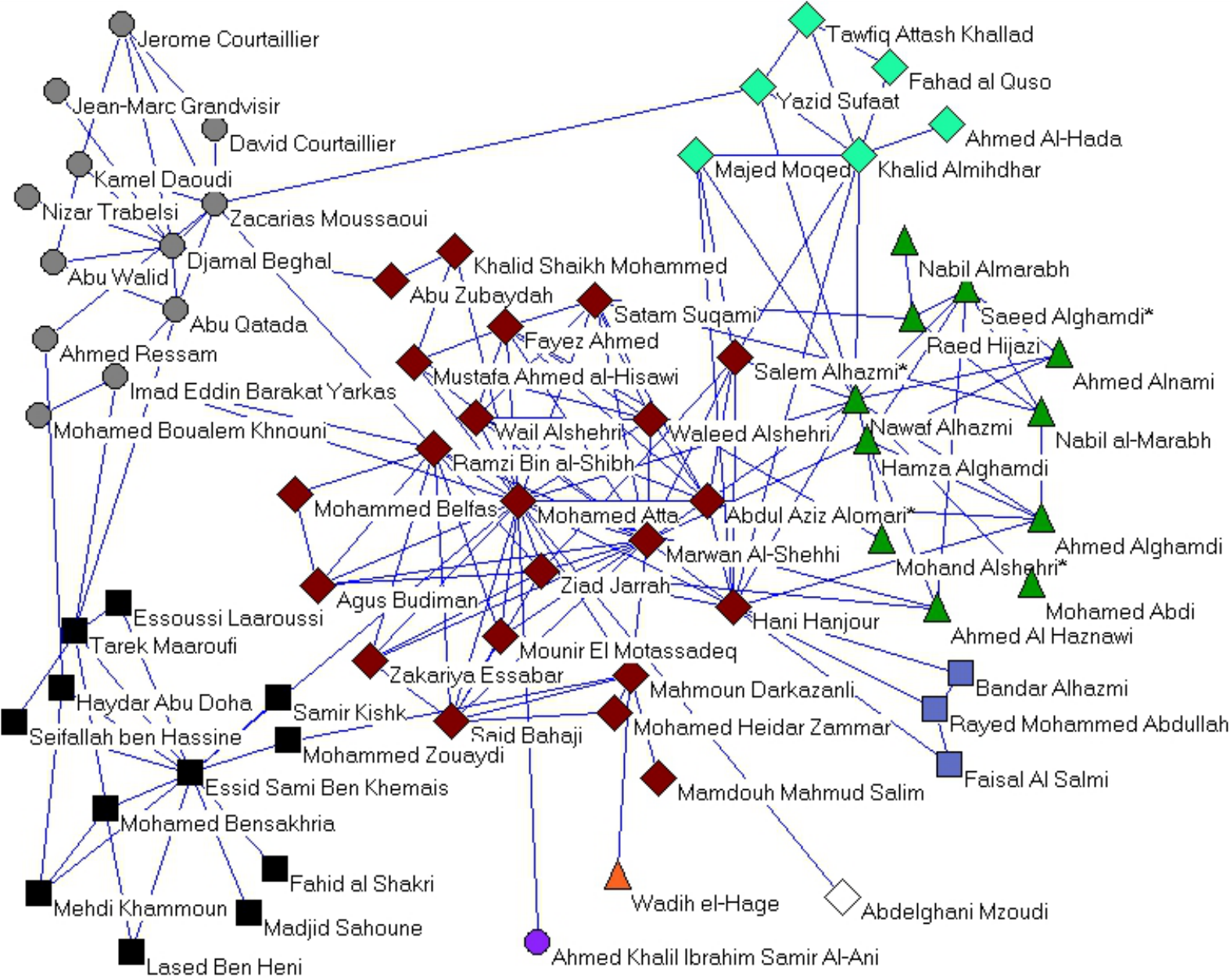}}
\subfigure[\ Expanding network around Mohamed Atta]
  {\includegraphics[width=1.0\columnwidth]{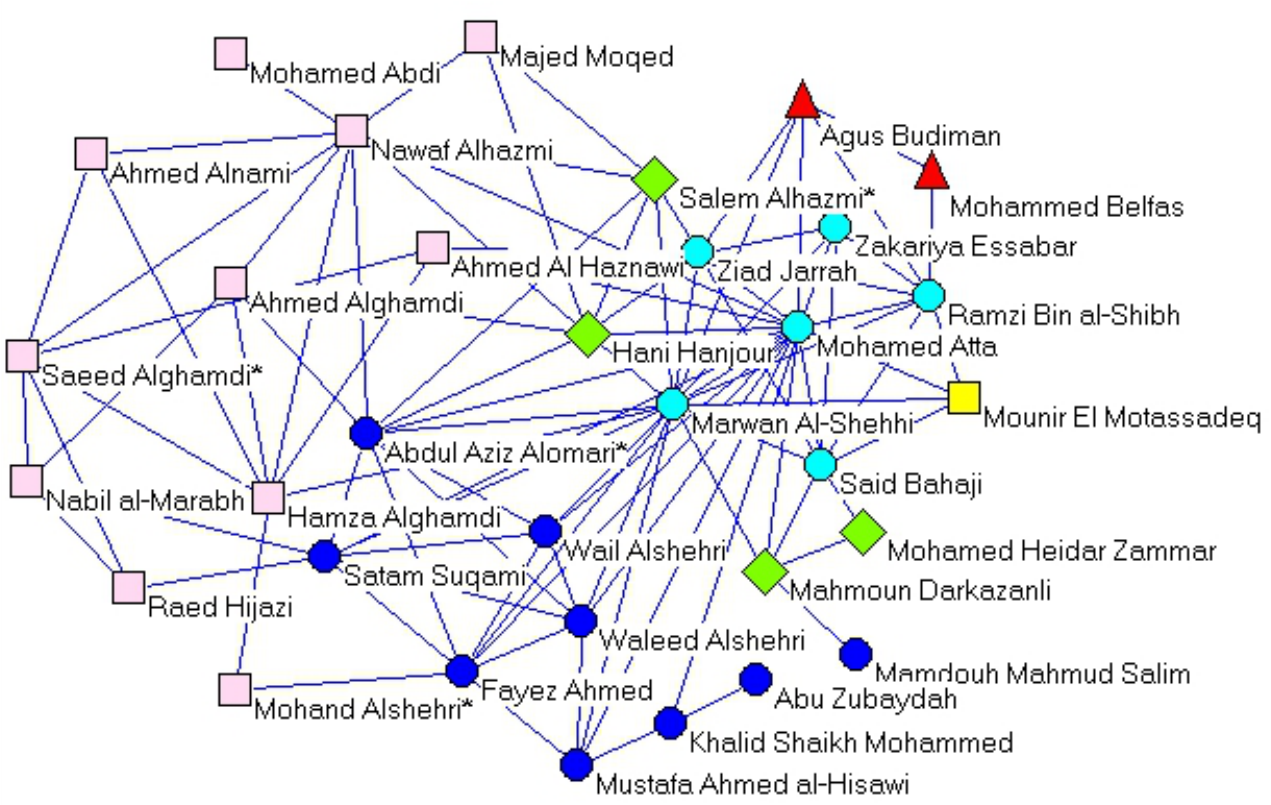}}
\caption{(Color online) The figure depicts a small terrorist network collected
from publicly available data \cite{ref:krebsconn}.
Panel (a) shows the overall network at $\gamma = 0.1$ in \eqnref{eq:ourmodel}
where distinct node shapes indicate separate communities.
Panel (b) shows an ``expanding'' community around Mohamed Atta where 
his ``local'' cluster grows roughly outward in the diagram.
Here, new node categories (shapes and colors) indicate nodes \emph{added} 
to the parent cluster (as opposed to new communities) as $\gamma$ is lowered 
to particular well-defined resolutions (see text). 
In this network, our local multiresolution algorithm indicates that these
communities are strongly defined on an individual basis with CVI $\vab =0$ 
in \subfigref{fig:terroristsplotlocal}{b} even at resolutions where the overall 
system structure is more vaguely defined in \figref{fig:terroristsplotglobal}.
This illustrates the main benefit of our local multiresolution approach.}
\label{fig:terrorists}
\end{figure*}

\begin{figure}
\centering
\subfigure[\ Global MRA: 9/11 Terrorists]{\includegraphics[width=0.9\columnwidth]{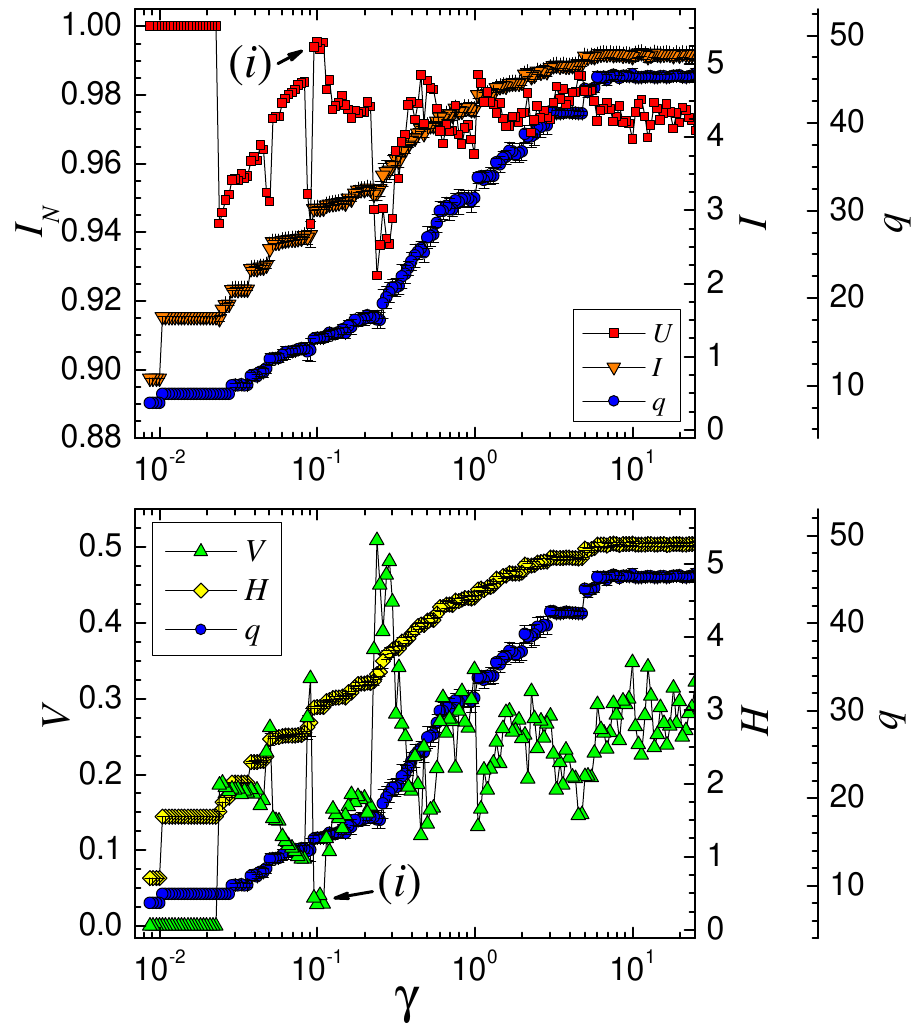}}
\caption{(Color online) We apply our multiresolution algorithm 
(see \Appref{app:localglobal} and \secref{app:MRA}) to a small terrorist network \cite{ref:krebsconn}.
Although the plot shows a best resolution at $\gamma\simeq 0.1$ (depicted 
in \figref{fig:terrorists}) as indicated by $V\simeq 0$, the remainder of the 
plot has a largely blurred multiresolution signature (high VI or low NMI).
The $V=0$ region on the far left is an essentially trivial partition into nearly
disjoint clusters.
In \figref{fig:terroristsplotlocal}, we show results from the local multiresolution 
algorithm in \secref{sec:LMRA} to three selected terrorists where we track the respective 
parent clusters over a range of resolutions [\ie{}, values of $\gamma$ in \eqnref{eq:ourmodel}]
and calculate the cluster correlations using the CVI and CNMI in \secref{sec:CNMVI}.}
\label{fig:terroristsplotglobal}
\end{figure}

\begin{figure*}
\centering
\subfigure[\ LMRA: Hani Hanjour]{\includegraphics[width=\widesubfigsize]{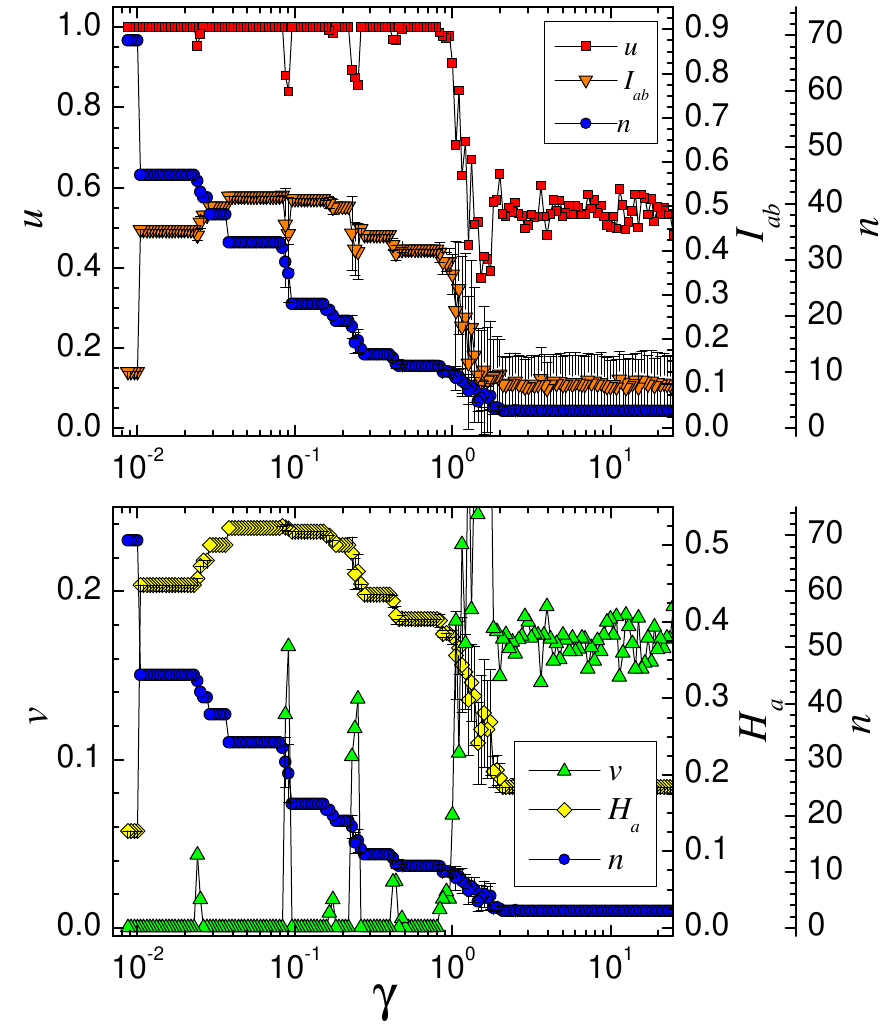}}
\subfigure[\ LMRA: Mohamed Atta]{\includegraphics[width=\widesubfigsize]{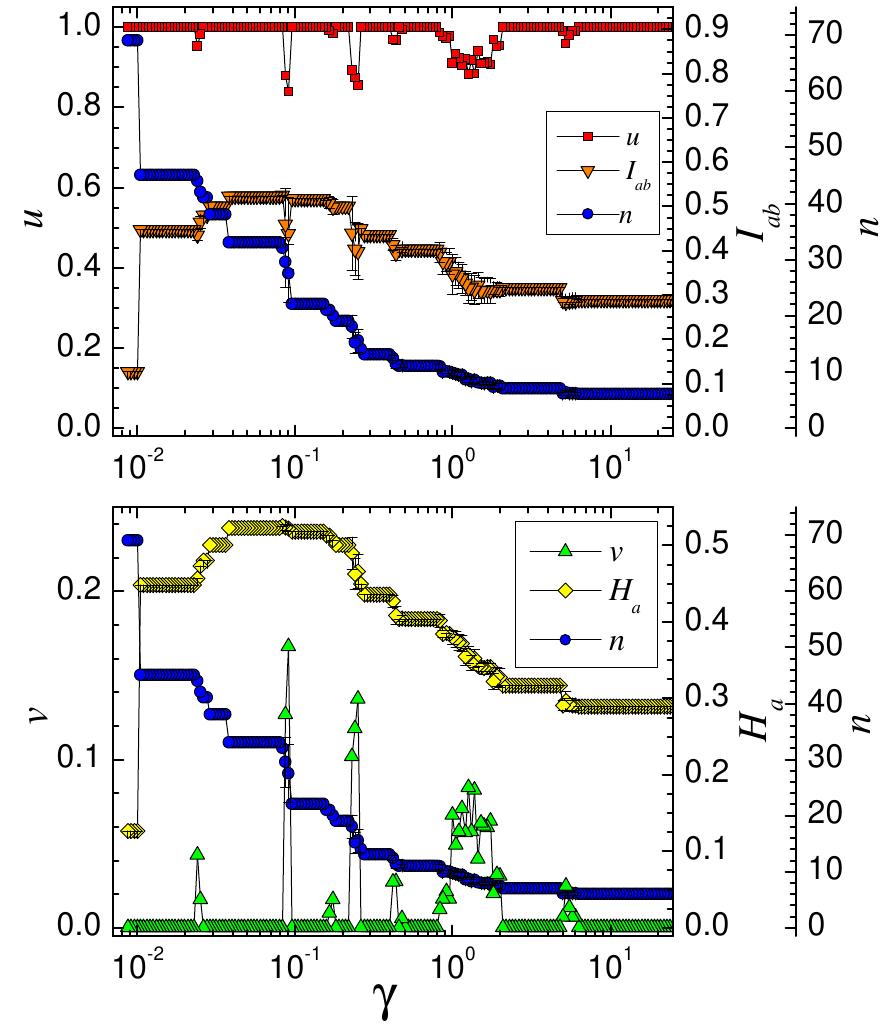}}
\subfigure[\ LMRA: Zacarias Moussaoui]
                                {\includegraphics[width=\widesubfigsize]{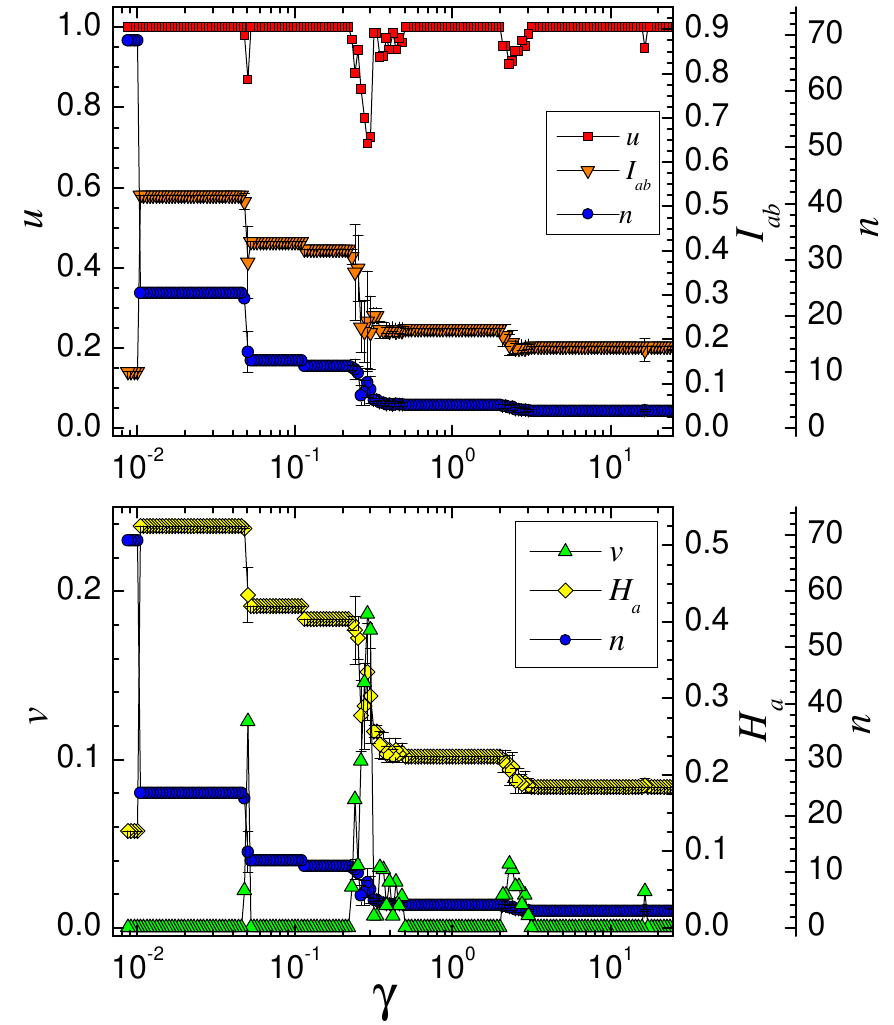}}
\caption{(Color online) In each panel, we apply our local multiresolution 
algorithm (LRMA, see \secref{sec:LMRA}) to a small terrorist network \cite{ref:krebsconn}.
We analyze three selected terrorists (see text) by tracking the respective parent 
clusters over a range of resolutions [\ie{}, values of $\gamma$ in \eqnref{eq:ourmodel}].
We then calculate the cluster correlations using the community comparison measures 
in \secref{sec:CNMVI}.
Note that the individual nodes possess certain strongly preferred resolutions 
with $\vab{}=0$ for their parent clusters whereas the global system in 
\figref{fig:terroristsplotglobal} is less well-defined for most values 
of $\gamma$.}
\label{fig:terroristsplotlocal}
\end{figure*}

\subsection{Branched hierarchy} \label{sec:staggered}

We construct a branched, strict hierarchy as depicted in 
\figref{fig:staggeredhierarchy} which we use to test the LMRA 
method of \secref{sec:LMRA}.
Level $1$ is the full system of $N=1024$ nodes;
level $2$ is the two-part branch split (groups of superclusters) 
with $N_L = 502$ and $N_R = 522$ nodes for the left (L) and right 
(R) sides, respectively; level $3$ is the set of superclusters;
level $4$ is the set of innermost clusters.

Level $1$ was defined by connecting nodes in the left and right branches 
(levels $2$L and $2$R) with an \emph{inter}community density $p_1 = 0.015$.
The approximate \emph{intra}community edge densities at level $4$ were 
$p_{4\mathrm{L}} = 0.9$ and $p_{4\mathrm{R}} = 0.6$ assigned randomly 
with a normal distribution of $\sigma_p = 0.02$.
We connected nodes \emph{between} the respective communities in the 
intermediate levels $2$ and $3$ with probabilities: $p_{3\mathrm{L}}=0.37$, 
$p_{3\mathrm{R}}=0.10$, $p_{2\mathrm{L}}=0.16$, and $p_{2\mathrm{R}}=0.03$.
These values were selected in order to demonstrate a somewhat ``blurred'' 
multiresolution signature in a controlled example where the underlying 
local structure is nevertheless strongly defined.

In \subfigref{fig:staggeredplotglobal}{a}, we show the global MRA 
algorithm from Ref.\ \cite{ref:rzmultires} (summarized in \secref{app:MRA})
applied to the full $N=1024$ node network using $r=20$ replicas and $t=10$ 
optimization trials per replica.
A more thorough discussion follows, but briefly, feature ($iv$) illustrates 
how poorly-correlated communities almost completely obscure the well-defined 
level $4$L structure.
Nevertheless, the local MRA algorithm in \secref{sec:LMRA} can 
\emph{fully extract this hidden section of the hierarchy}.

In \figref{fig:staggeredplotglobal}, the left axes plot NMI, $\INN$, and VI, $V$, 
from \secref{sec:NMVI} in the top and bottom sub-panels, respectively, averaged 
over all replica pairs.
On the right axes, we plot the average mutual information $I$ and the Shannon 
entropy $H$ for the top and bottom sub-panels, respectively.
The right offset axes in both sub-panels plot the average number of communities 
$q$.
Panels (b) and (c) show the MRA results applied to the separate left and right 
branches of the hierarchy, respectively, using the same $r$ and $t$ as in panel 
(a).

For completeness, features ($i$)--($iii$) in panel (a) illustrate how the global 
MRA signature can identify preferred or stable resolutions by low VI or high NMI 
correlations (or plateaus in $H$, $I$, and $q$ in this example) averaged over the 
independently-solved replica partitions.
Specifically, feature ($i$) corresponds to level the $2$ partition with $q_{i}=2$, 
and feature ($ii$) concurrently identifies levels $2$L and $3$R with $q_{ii}=11$
because the respective community edge densities are similar 
(see \secref{sec:resolution}).
Likewise, feature ($iii$) solves levels $3$L and $4$R with $q_{iii}=52$.
These specific partitions consist of combinations of well-resolved sub-graphs 
at different levels of the branched hierarchy, but it is the loss of level $4$L 
in the global MRA plot that is the main topic of this example.

At feature ($iv$) in panel (a), the poor correlations show that 
\emph{the global analysis of the full system misses level} $4$L.
This occurs because the well-defined local clusters conflict with more random 
partitions for the right-side subgraph in \figref{fig:staggeredhierarchy}.
In contrast, panels (b) and (c) show that the MRA method applied to the separate
left and right branches are perfectly defined with $V=0$ and $\INN=1$ 
[marked by ($2$L), ($3$L), $\ldots$, ($4$R), respectively].
That is, the structure clearly exists locally, but the global MRA method in panel 
(a) cannot resolve level $4$L.

In \subfigref{fig:staggeredplotlocal}{a--c}, we plot the results of the new LMRA 
method from \secref{sec:CNMVI} for the parent clusters of three randomly selected 
nodes $116$, $661$, and $951$, respectively, as identified within the full $N=1024$ 
node system.
On the left axes, we plot CNMI $\Inab$ in \eqnref{eq:CNMI} and CVI $\vab$  
in \eqnref{eq:CVI}, respectively, averaged over all community pairs in the 
respective replicas.
On the right axes, we plot the mutual information contribution $I_{ab}$
in \eqnref{eq:IabAB} and the Shannon entropy contribution $H_a$ in \eqnref{eq:Ha} 
averaged over all pairs of target communities in the replicas or all target 
communities, respectively.
The offset right axes plot the average number of nodes $n$ over all targeted 
communities.

Despite being buried within the full $N=1024$ node system, the parent cluster
of node $951$ corresponding to level $4$L is clearly present in the LMRA 
analysis in \subfigref{fig:staggeredplotlocal}{b,c}.
This illustrates how \emph{our LMRA algorithm can resolve well-defined local 
structure even when the global signature is obscured}.
In principle, we could further apply the LMRA algorithm to all clusters 
in the partitions and unambiguously identify the entire set of well-defined 
level $4$L communities.

\subsection{Small terrorist network} \label{sec:terrornetwork}

Even small networks can possess strongly-defined local clusters
within a more indistinct global partition.
We apply the LMRA method to a small terrorist network related to
the terrorism attacks of September $11$, $2001$, as constructed from 
publicly available data \cite{ref:krebsconn}.
Given that the highest quality intelligence would be classified, our purpose 
here is to demonstrate the practical application of the LMRA on real data 
as opposed to setting forth a rigorous study of the terrorist network.
Toward this end, we select Mohamed Atta for study as the leader of the operation.
We further consider Hani Hanjour, who was another pilot, and Zacarias Moussaoui, 
who was the only terrorist of the $20$ prevented from participating in the attack.

\subFigref{fig:terrorists}{a} depicts the network at $\gamma = 0.1$ 
in \eqnref{eq:ourmodel} corresponding to the minimum VI at feature ($i$)
in \figref{fig:terroristsplotglobal} with $V\simeq 0$ (see below). 
Here, distinct node shapes indicate separate \emph{communities}.
The community partitions with $V=0$ at the lowest $\gamma$'s are disjoint 
clusters where all nodes are completely collapsed into their respective 
connected groups, resulting in a trivial partition.
The left axes plot $\INN$ and $V$ (see \secref{sec:NMVI}) for top and bottom 
sub-panels, respectively, averaged over all replica pairs.
On the right axes, we plot $I$ and $H$ for the top and bottom sub-panels, 
respectively, and the offset axes in both sub-panels plot the average 
number of communities $q$.

\subFigref{fig:terrorists}{b} shows the expanding network core 
centered on Mohamed Atta at several strongly-defined resolutions as determined
from \subfigref{fig:terroristsplotlocal}{b} where $\vab = 0$.
In this panel, distinct node shapes and colors indicate \emph{added nodes} 
[as opposed to new communities in panel (a)], roughly spreading outward 
as $\gamma$ is lowered.  Specifically, the fixed resolutions correspond 
to $\gamma = 10$ (smallest, innermost cyan circles),
$\gamma = 3$ (yellow square), $\gamma = 0.6$ (green diamonds), 
$\gamma = 0.3$ (red triangles), $\gamma = 0.125$ (dark blue circles), 
and $\gamma = 0.05$ (largest, pink squares) with a few other small 
fluctuations not depicted.

On the left axes in \subfigref{fig:terroristsplotlocal}{a--c}, 
we plot CNMI $\Inab$ in \eqnref{eq:CNMI} and CVI $\vab$  
in \eqnref{eq:CVI}, respectively, averaged over all pairs of parent 
communities in the respective replicas.
Similarly, the right axes plot the mutual information contribution 
$I_{ab}$ in \eqnref{eq:IabAB} and the Shannon entropy contribution $H_a$ 
in \eqnref{eq:Ha} averaged over all pairs of parent communities or all 
parent communities, respectively.
The right offset axes display the average number of nodes $n$ over the 
parent communities.

Each panel shows distinct, but different, regions of $\gamma$ where the 
parent clusters are strongly defined, but the cluster correlations in the 
full network 
in \figref{fig:terroristsplotglobal} are more poorly defined at most resolutions.
Hani Hanjour has a LMRA signature distinct from Mohamed Atta for $\gamma\gtrsim 1$,
but they match at lower $\gamma$ because they are mutual members of the same 
communities.

\subsection{LFR Benchmark} \label{sec:LFRnetwork}

We also tested the LMRA method on a common CD benchmark by Lancichinetti
Fortunato and Radicchi \cite{ref:lancbenchmark},
which was designed to emulate a series of strongly defined networks 
with realistic distributions of community sizes and edge assignments.
Specifically, it defines a power-law distribution of community sizes
specified by an exponent $\beta$, 
minimum size $n_\mathrm{min}$, and maximum size $n_\mathrm{max}$.
It adds random unweighted, undirected edges to the network, defining both the communities 
and any intercommunity noise (extraneous edges outside of the well-defined communities), 
according to a power-law distribution of node degrees given
by an exponent $\alpha$, average power-law degree $\kpowavg$ 
(or minimum degree $k_\mathrm{min}$), and maximum degree $k_\mathrm{max}$.

In the current tests, we solve for each system using the algorithm in \secref{sec:LMRAalgorithm} 
using $20$ replicas to ensure that we have a good sample of possible partitions
and $t=4$ trials per CD solution attempt.
We used $N=10~\!000$, $\kpowavg =35$, 
$n_\mathrm{min}=10$, $n_\mathrm{max}=50$, $\alpha = -2$, $\beta = -1$, 
and $\mu = 0.1$ or $0.5$ as indicated in \figref{fig:LFRplotglobalmulow}.
The mixing parameter $\mu$ controls the level of intercommunity noise.
The results of the global MRA method of \secref{app:MRA} are shown in \figref{fig:LFRplotglobalmulow},
and the corresponding local LMRA data for three randomly selected nodes are in \figref{fig:LFRplotlocalmulow}.
As with previously demonstrated examples \cite{ref:rzmultires,ref:lancLFRcompare}, the global MRA 
algorithm correctly identifies the constructed global partition with exceptional accuracy 
in the presence low and high noise in panels (a) and (b).
The extreme noise case in panel $c$ was included for comparison purposes, since mixing parameters higher 
than $\mu\simeq 0.7$ present exceptional challenges for all tested CD algorithms in \cite{ref:lancLFRcompare}.
The planted partition may even exceed limits of well-defined communities (see \cite{ref:decelleKMZPT,ref:huCDPTsgd,ref:darstCDSBMdef} 
for some general discussion).

In \subfigref{fig:LFRplotlocalmulow}{a} where $\mu = 0.1$, the LRMA method identifies the
proper community of node $7603$, as indicated by the arrows at feature ($i$).
Panel (b) of \figref{fig:LFRplotlocalmulow} for node $1213$ shows a similar for $\mu=0.5$,
where the correct cluster is again identified perfectly.
In both cases, the cluster structures appear to have reasonably well-defined parent 
communities with $v$ being roughly near zero for $\gamma\gtrsim 4$ and $1$, respectively, 
but the CNMI measure clearly shows $u<1$ indicating poorly defined communities.
Thus, the proper interpretation of parameters $v$ and $u$ requires considering the relative 
values over the studied range of $\gamma$. 
Further, one should consider all of the measures together when identifying relevant 
community structure(s).

Feature ($i$) in panel (c) for node $2502$ indicates two $\gamma$ values which display perfect 
correlations in CNMI and CVI, respectively.
At $\gamma = 0.12$, the parent cluster of the target node forms a trivial split into a size $n=2$ 
community due to the excessive noise.
The result at $\gamma =0.095$ shows a partial-identification of the intended cluster with $u=0.88$ 
when compared with the intended community by construction.
This correlation represents a transient, partially-false-positive identification since $u=1$ among 
the algorithm replicas.  
Additional helps mitigate such identifications.
The large error bars to the left of feature ($i$) indicate wide disagreement between the different
candidate clusters among the replicas.

For other problems with very low noise (not depicted), both CVI and CNMI may indicate a strong 
partition with $u\simeq 1$ or $v\simeq 0$ across a wide range of $\gamma$'s without clearly marking 
the transitions by changes in $v$ or $u$.
This may imply that the communities in which the node resides are not blurred by significant noise 
so that many of the structural transitions as $\gamma$ is varied are sharply defined.
The supplemental measures of $H_a$ and $I_{ab}$ for communities $a$ and $b$ or the number of communities 
$q$ can help distinguish between these different structures, but in these cases, the base LMRA method may 
not be conclusive.

\begin{figure*}
\centering
\subfigure[\ $\mu=0.1$]{\includegraphics[width=\widesubfigsize]{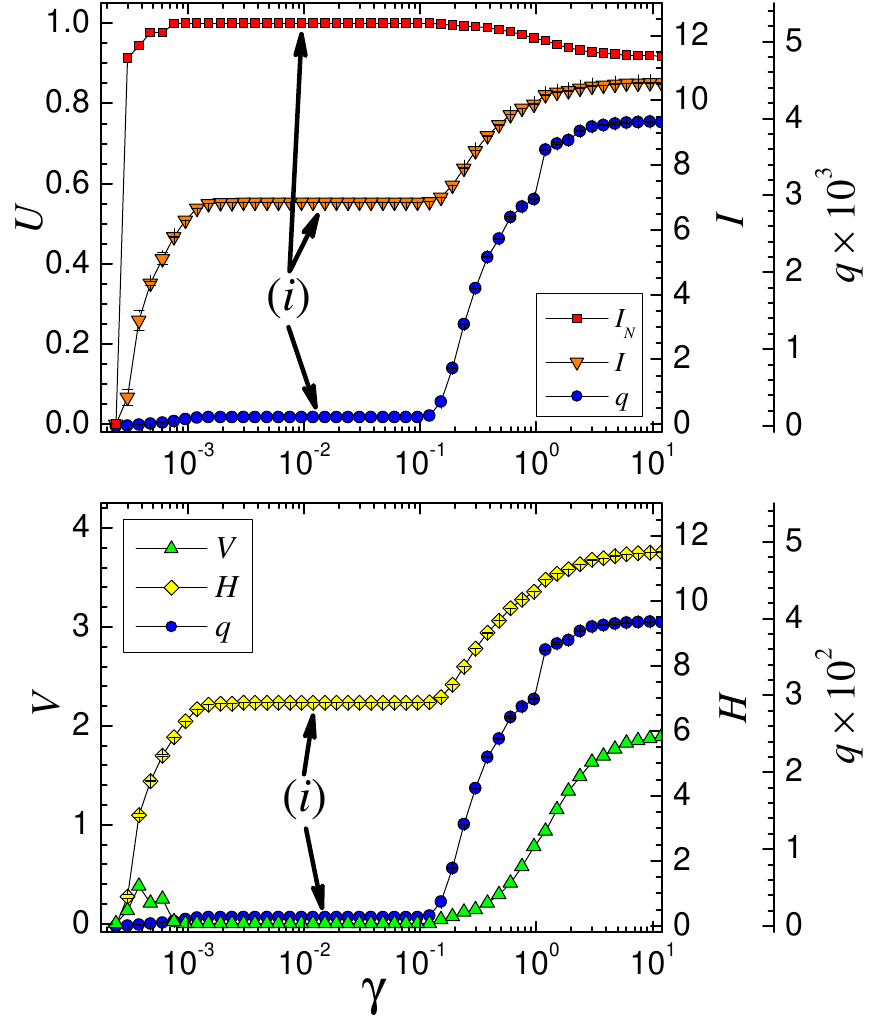}}
\subfigure[\ $\mu=0.5$]{\includegraphics[width=\widesubfigsize]{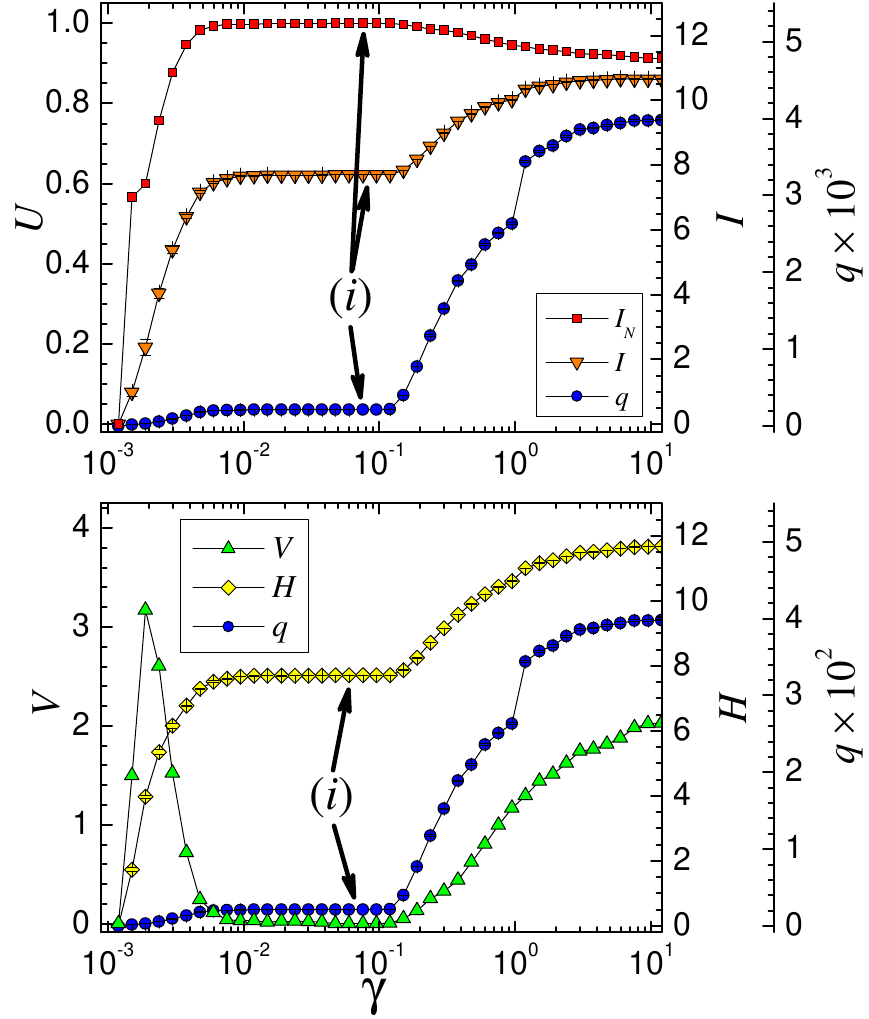}}
\subfigure[\ $\mu=0.8$]{\includegraphics[width=\widesubfigsize]{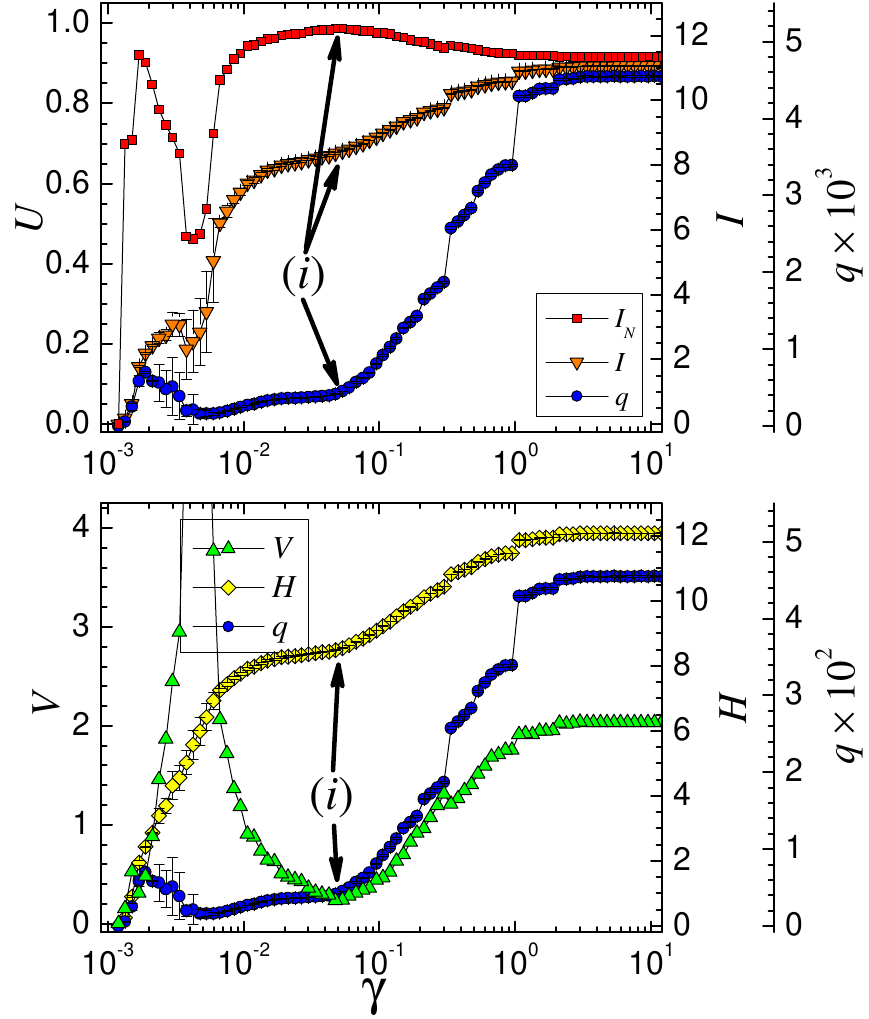}}
\caption{(Color online) We apply our multiresolution algorithm from 
\secref{app:MRA} \cite{ref:rzmultires,ref:lancLFRcompare} to the LFR benchmark \cite{ref:lancbenchmark} 
with mixing parameters $\mu = 0.1$,  $0.5$, and $0.8$. 
These values correspond to a low (panel a), moderately high (b), and extremely high (c) levels of network noise, respectively.  
See the text for other benchmark parameters.
In panels (a) and (b), the constructed partitions are correctly identified at feature ($i$) by low VI or high NMI.
In panel (c), a preferred partition resolution is implied by similar extrema at feature ($i$), but
the exact identity of the intended partition is probably beyond the capability of our algorithm to extract. 
Further, based on overall results in \cite{ref:lancLFRcompare}, the intended partition likely lies beyond the 
detection capability of current CD algorithms.
Other works \cite{ref:decelleKMZPT,ref:huCDPTsgd,ref:darstCDSBMdef} discuss maximum detectability limits 
and transitions, which may apply here.
In \figref{fig:LFRplotlocalmulow}, we show corresponding results from the local multiresolution 
algorithm in \secref{sec:LMRA} for randomly selected nodes where we track the respective 
parent clusters over a range of resolutions [\ie{}, values of $\gamma$ in \eqnref{eq:ourmodel}]
and calculate the cluster correlations using the CVI and CNMI in \secref{sec:CNMVI}.}
\label{fig:LFRplotglobalmulow}
\end{figure*}

\begin{figure*}
\centering
\subfigure[\ $\mu=0.1$]{\includegraphics[width=\widesubfigsize]{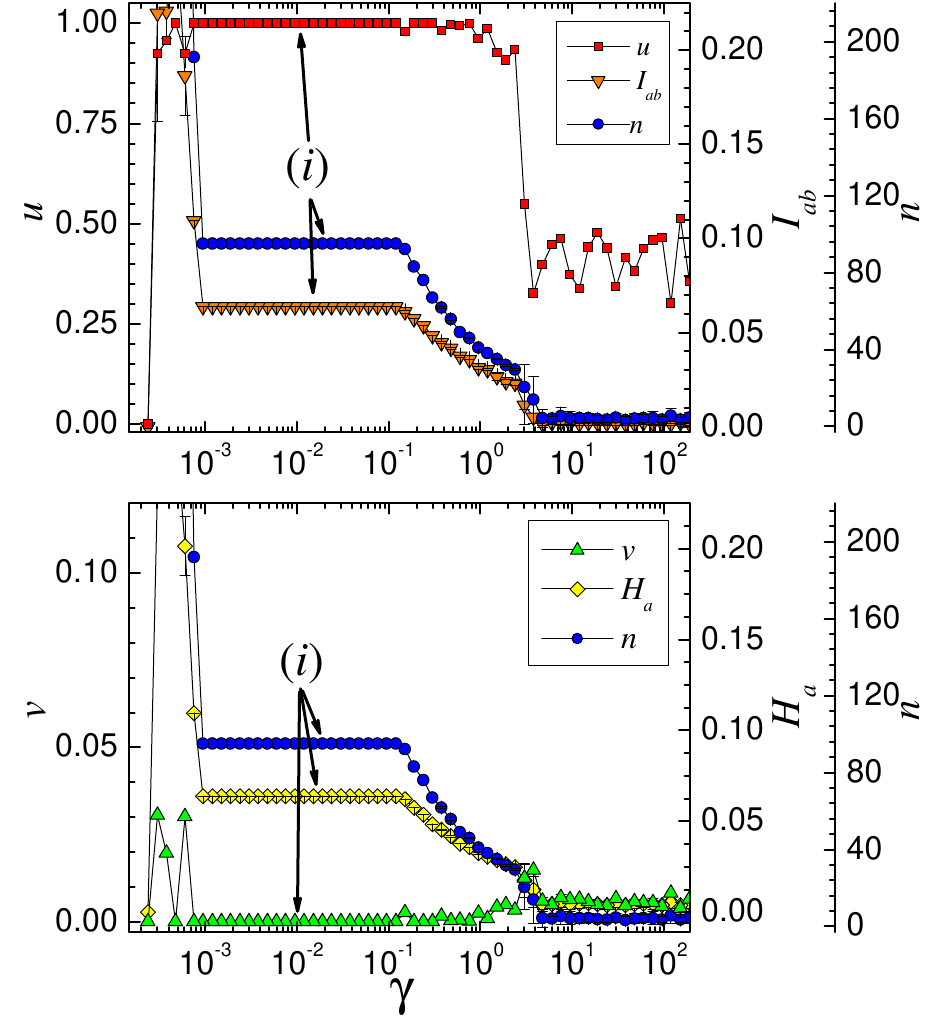}}
\subfigure[\ $\mu=0.5$]{\includegraphics[width=\widesubfigsize]{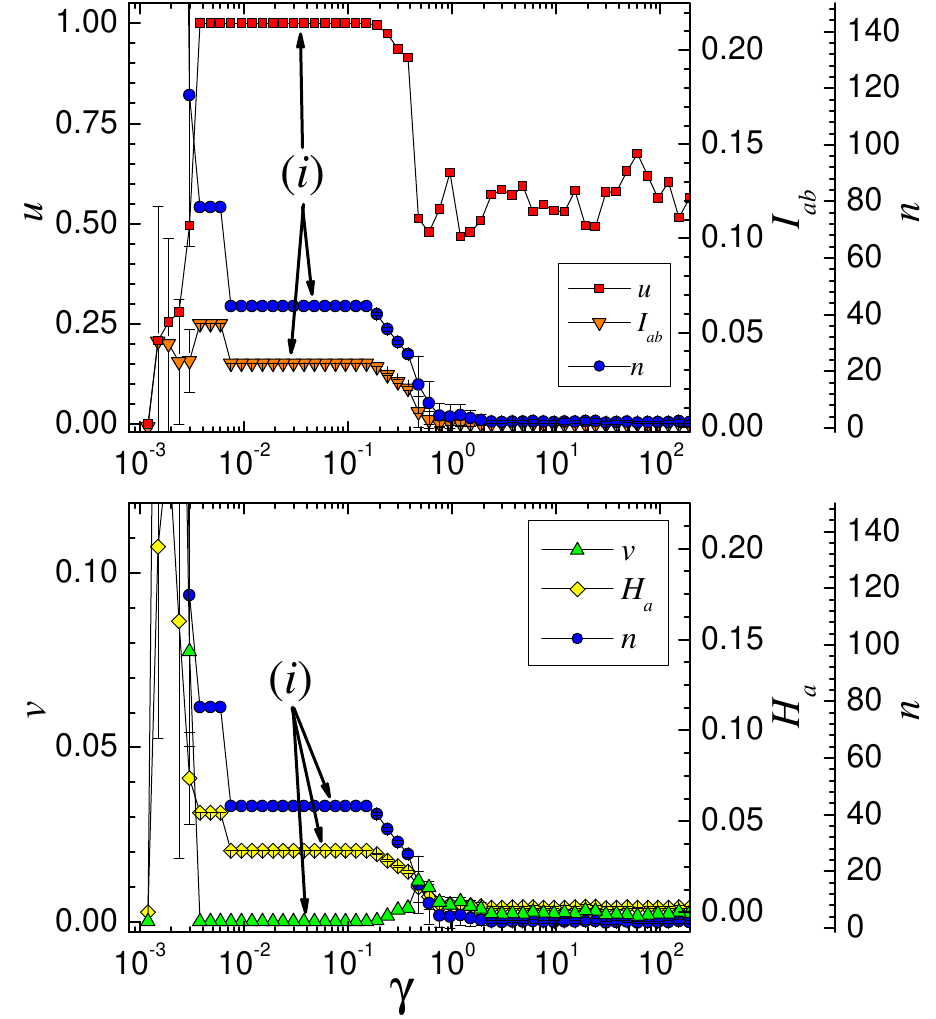}}
\subfigure[\ $\mu=0.8$]{\includegraphics[width=\widesubfigsize]{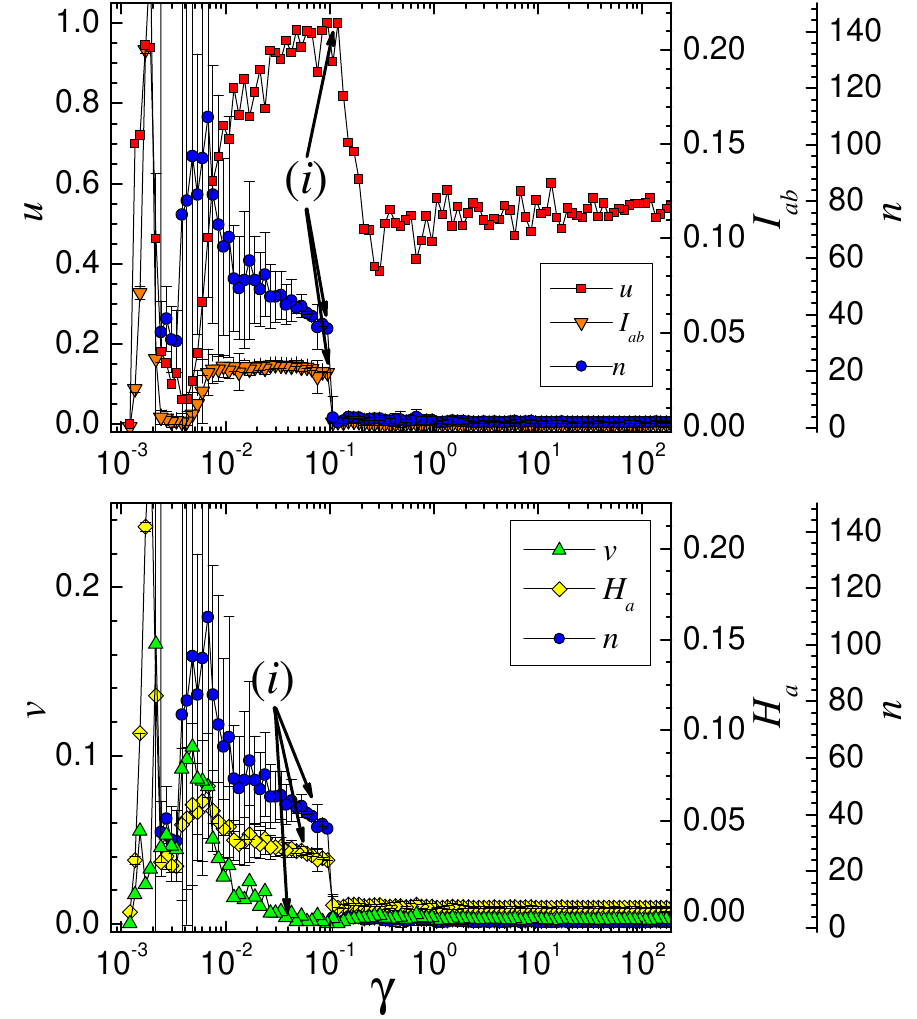}}
\caption{(Color online) In each panel, we apply our local multiresolution 
algorithm of \secref{sec:LMRA} to an implementation of the LFR benchmark \cite{ref:lancbenchmark}.
Mixing parameters are $\mu = 0.1$,  $0.5$, and $0.8$ which correspond to a low (panel a), moderately high (b), 
and extremely high (c) levels of network noise, respectively.  See the text for other benchmark parameters.
We analyze the parent communities for randomly selected nodes over a range of resolutions 
[\ie{}, values of $\gamma$ in \eqnref{eq:ourmodel}].  Other nodes display comparable multiresolution signatures.
We then calculate the cluster correlations using the community comparison measures in \secref{sec:CNMVI}.
In panels (a) and (b), the defined communities are correctly identified as feature ($i$)
with $v=0$ and $u=1$ for their parent clusters, corresponding to the correct solution 
for the global system in \figref{fig:LFRplotglobalmulow}.
The correct benchmark communities are clearly detected as stable regions in each of the various cluster measures.
In panel (b), note that $v$ is near zero for this plot beyond $\gamma\gtrsim 4$, but the CNMI 
measure clearly shows relatively low $u$ values, indicating that we should focus on the correct 
partition at feature ($i$).
Feature ($i$) in panel (c) has two perfect correlations in CNMI and CVI at $\gamma = 0.12$ and $0.095$, respectively.
The former is a trivial split of the target node into a size $n=2$ community due to the excessive noise,
and the latter is a partial-identification of the intended cluster ($u=0.88$ with the intended construction),
which represents a transient, partially-false-positive identification since $u=1$ among the algorithm replicas.
However, the exact intended partition likely lies beyond the detection capability of current CD algorithms
(see \cite{ref:lancLFRcompare} and also \cite{ref:decelleKMZPT,ref:huCDPTsgd,ref:darstCDSBMdef} which discuss 
maximum detectability limits and transitions), and there is some concern as to whether the intended partition 
is even well-defined.}
\label{fig:LFRplotlocalmulow}
\end{figure*}


\section{Conclusion} \label{sec:conclusion}

Multiresolution network analysis extends the basic notions of community 
detection to select the best resolution(s) for a given network over 
a range of network scales.
Certain networks may present situations where local clusters experience 
a lost-in-a-crowd effect.
Despite being strongly defined, the local structure may be lost among 
a collection of more poorly defined communities at a given resolution.
This may occur due to the sheer size of a network or because most clusters 
do not coalesce in their strongest state(s) at the same scale(s).

We presented an extension of an existing global multiresolution method 
\cite{ref:rzmultires} to detect and quantitatively assess local multiresolution 
structure.
We proposed cluster-level analogies to variation of information and normalized 
mutual information which evaluate the strength of local communities in the 
context of a pair of network partitions. 
After applying these measures to evaluate correlations among individual parent 
communities in multiple independent solutions (replicas), we demonstrated 
that the proposed local multiresolution algorithm is able to identify the
best resolutions and extract the local 
structure despite a blurred global multiresolution signature.
In addition, we demonstrated that the algorithm correctly identifies individual 
clusters in the single-layered structure of the LFR benchmark.
Our approach only requires output communities from a multiresolution capable 
CD algorithm, so it is independent of the search algorithm or CD 
model, making it suitable for use with any CD method 
that can identify partitions across different network scales.

\section*{ACKNOWLEDGMENTS}

This work was supported by NSF grant DMR-1106293 (ZN).
We wish to thank S. Chakrabarty, R. K. Darst, P. Johnson, and D. Hu for 
discussions and ongoing work.
ZN also thanks the Aspen Center for Physics and NSF Grant
\#1066293 for hospitality during the final stages of this work.

\appendix

\section{Local and global terminology}  \label{app:localglobal}

The meaning of the terms local and global depends on the context.
For our purposes, global \emph{cost functions} are those that \emph{require} 
network wide (global) parameters (e.g., number of edges $L$,
number of communities $q$, overall graph density $p$, etc.) in the 
quantitative evaluation of community structure \cite{ref:gn,ref:smcd}.
Global \emph{multiresolution methods} are those for which the best partition
is simultaneously determined for the entire system, effectively averaging 
the partition robustness over all communities.
This is true regardless of whether the cost function is itself local or global 
in nature.

Local \emph{cost functions} \cite{ref:rzmultires,ref:rzlocal,ref:traaglocalscope} 
or algorithms \cite{ref:LPA} utilize parameters only in the neighborhood 
of a community or node (e.g., size of community $a$, edges of node $i$, etc.) 
to evaluate the best community structure.
These can be subdivided into weak and strong local cost functions 
\cite{ref:rzlocal} where weakly-local cost functions may depend on the details
of the community structure.
Briefly, strongly local cost functions determine community membership for a given 
node based only on the node's own relations with candidate communities.
Local \emph{multiresolution methods}, such as the current work, seek to identify
the best communities based on their own robustness at a given resolution.
That is, each community determines whether it is strongly 
defined regardless of the community structure present in the remainder of the network.
The evaluation of the best resolution is not effectively averaged over all the 
communities in the graph. 
Further, each community may be strongly resolved at different network scales (often 
described in terms of certain model weight parameters).

\section{Semi-metric property of CVI} \label{app:semimetricproof}

A semi-metric possesses intuitive ``distance-like'' properties
for comparing cluster similarity.
The proof that CVI is a semi-metric is trivial.
A measure $S(a,b)$ on a set $X$ with two variables $a$ and $b$ 
in $X$ is a semi-metric if and only if it satisfies the following
conditions:
\begin{itemize}
	\item Non-negativity -- $S(a,b)\geq 0$ for all $a$ and $b$.
	\item Zero only for equality -- $S(a,b)=0$ if and only if $a=b$.
	\item Symmetry -- $S(a,b)=S(b,a)$ for all $a$ and $b$.
\end{itemize}
$S(a,b)$ is a metric if it additionally satisfies the triangle inequality 
$S(a,c) \leq S(a,b) + S(b,c)$ for three variables $a$, $b$, and $c$ in $X$.

\begin{theorem}
CVI in \eqnref{eq:CVI} is a semi-metric between two clusters $a$ and 
$b$ in partitions $A$ and $B$ of size $|A|=|B|=N$ in the space of 
possible partitions of the $N$ nodes:
\emph{($1$)} It is non-negative and equal to zero if and only if $a=b$.
\emph{($2$)} It is symmetric with respect to clusters $(a,A)$ and 
$(b,B)$, $\vab = \vba$.
\end{theorem}

\begin{proof} 
\ \newline
\noindent ($1$) It is non-negative and strictly equal to zero if and only if $a=b$.
From \eqnref{eq:CVI}
\begin{eqnarray}
  \vab  &    =   & -\frac{n_a}{N}\log\left(\frac{n_a}{N}\right) 
                     -\frac{n_b}{N}\log\left(\frac{n_b}{N}\right) \nonumber\\
        &        & - 2\frac{n_{ab}}{N} \log\left(\frac{n_{ab} N}{n_a n_b}\right) \nonumber\\
        &    =   &   \frac{n_a - n_{ab}}{N}\log\left(\frac{N}{n_a}\right) 
                       + \frac{n_b - n_{ab}}{N}\log\left(\frac{N}{n_b}\right) \nonumber\\
        &        & + \frac{n_{ab}}{N} \log\left(\frac{n_a}{n_{ab}}\right) 
                   + \frac{n_{ab}}{N} \log\left(\frac{n_b}{n_{ab}}\right) \nonumber\\
  \vab  &   \ge  & 0
\end{eqnarray}
since $n_a>0$, $n_b>0$, $n_{ab}\ge 0$, $n_a\ge n_{ab}$, $n_b\ge n_{ab}$, 
$N\ge n_a$, and $N\ge n_b$.

Furthermore, $a=b$ implies that $n_a=n_b=n_{ab}$, which trivially yields
\begin{eqnarray}
  \vaa  &    =   & -\frac{n_a}{N}\log\left(\frac{n_a}{N}\right) 
                     -\frac{n_a}{N}\log\left(\frac{n_a}{N}\right) \nonumber\\
        &        & - 2\frac{n_a}{N} \log\left(\frac{n_a N}{n_a n_a}\right) \nonumber\\
        &    =   & -2\frac{n_a}{N}\log\left(\frac{n_a}{N}\right) 
                   + 2\frac{n_a}{N} \log\left(\frac{n_a}{N}\right) \nonumber\\
  \vaa  &   =    & 0.
\end{eqnarray}
Now, if $v(a,b)=0$, this implies
\begin{eqnarray}
  \vab  &    =   & -\frac{n_a}{N}\log\left(\frac{n_a}{N}\right) 
                     -\frac{n_b}{N}\log\left(\frac{n_b}{N}\right) \nonumber\\
        &        & - 2\frac{n_{ab}}{N} \log\left(\frac{n_{ab} N}{n_a n_b}\right) \nonumber\\
   0    &    =   &   \frac{n_a - n_{ab}}{N}\log\left(\frac{N}{n_a}\right) 
                       + \frac{n_b - n_{ab}}{N}\log\left(\frac{N}{n_b}\right) \nonumber\\
        &        & + \frac{n_{ab}}{N} \log\left(\frac{n_a}{n_{ab}}\right) 
                   + \frac{n_{ab}}{N} \log\left(\frac{n_b}{n_{ab}}\right) \label{eq:vabzero}
\end{eqnarray}
Since $n_{ab}\ge 0$, $n_a\ge n_{ab}$, $n_b\ge n_{ab}$, $N\ge n_a$, and $N\ge n_b$, 
all terms are non-negative;  so they cannot cancel each other.  
Each term must be individually zero.

First, if $n_{ab}=0$ the last two terms in \eqnref{eq:vabzero} are zero. 
Then since $n_a>0$ and $n_b>0$, $N=n_a=n_b$, but this result requires 
that $n_{ab}=N$ which contradicts the assumption that $n_{ab}=0$.  
Thus, $n_{ab}>0$.

Now, the third term in \eqnref{eq:vabzero}
vanishes only if $n_a=n_{ab}$, and fourth term is zero only if $n_b=n_{ab}$.
Thus, $n_a=n_b=n_{ab}$ which means that $a=b$.

Thus, $v(a,b)=0$ if and only if $a=b$.

\bigskip

\noindent ($2$) It is symmetric with clusters $(a,A)$ and $(b,B)$, $\vab = \vba$.
\newline
\noindent Since $n_{ab}$ is necessarily equal to $n_{ba}$, $I_{ab}(A,B)$ is 
symmetric in clusters $(a,A)$ and $(b,B)$. 
The symmetry of $\vab$ is then immediately obvious.

Thus, CVI is a semi-metric.
\end{proof}
\ \\
We have not proved the triangle inequality for CVI, making it a metric, 
but the triangle inequality appears to be violated rarely, if at all.

\section{Alternative cluster measures} \label{app:altCVI}

A tempting alternate measure for CVI might be defined based on the 
individual terms of 
\begin{eqnarray} 
  V(A,B) & = & H(A|B) + H(B|A) \nonumber\\
         & = & \sum_{a,b} \left[ \frac{n_{ab}}{N}\log\frac{n_b}{n_{ab}}
                            +\frac{n_{ab}}{N}\log\frac{n_a}{n_{ab}} \right].
\end{eqnarray} 
where $A$ and $B$ are two partitions of a graph.
From this equivalent definition of VI, the natural CVI definition for two
clusters $a$ in $A$ and $b$ in $B$ would be  
\begin{equation}
  v_{ab}'(A,B) = \frac{n_{ab}}{N}\log\frac{n_a}{n_{ab}}
              + \frac{n_{ab}}{N}\log\frac{n_b}{n_{ab}}.
  \label{eq:altCVIHAB}
\end{equation}
Unlike CVI in \eqnref{eq:CVI}, \eqnref{eq:altCVIHAB} has the nice property 
that the individual cluster contributions sum to $V(A,B)$, 
$V(A,B) = \sum_{a}^{q_A}\sum_{b}^{q_B} v_{ab}'(A,B)$.

Unfortunately, this particular measure does not work for cluster 
comparisons.
While $v_{aa}'(A,A)=0$ as desired, it is also the case that $v_{ab}'(A,B)=0$ 
if $n_{ab}=0$.
That is, it is zero if \emph{no overlap} exists between $a$ and $b$
which violates the notion of a ``distance'' as well as one of the requirements 
for being a (semi)metric. 
VI is a metric on partitions $A$ and $B$ because it \emph{sums} over all $a$ 
and $b$ in $A$ and $B$, respectively.

We could also consider an alternate \emph{ad hoc} definition by 
redefining the CVI entropy terms in \eqnref{eq:vitoviab} according 
to $v_{ab}''(A,B) = H_a(A)/q_B + H_b(B)/q_A - 2 I_{ab}(A,B)$.
This variant would again yield the desirable property $V(A,B) = 
\sum_{a}^{q_A}\sum_{b}^{q_B} v_{ab}''(A,B)$, but the measure loses the 
semi-metric requirements $v_{ab}''(A,B)\ge 0$ and $v_{aa}''(A,A)=0$.


\begin{thebibliography}{87}%
\makeatletter
\providecommand \@ifxundefined [1]{%
 \@ifx{#1\undefined}
}%
\providecommand \@ifnum [1]{%
 \ifnum #1\expandafter \@firstoftwo
 \else \expandafter \@secondoftwo
 \fi
}%
\providecommand \@ifx [1]{%
 \ifx #1\expandafter \@firstoftwo
 \else \expandafter \@secondoftwo
 \fi
}%
\providecommand \natexlab [1]{#1}%
\providecommand \enquote  [1]{``#1''}%
\providecommand \bibnamefont  [1]{#1}%
\providecommand \bibfnamefont [1]{#1}%
\providecommand \citenamefont [1]{#1}%
\providecommand \href@noop [0]{\@secondoftwo}%
\providecommand \href [0]{\begingroup \@sanitize@url \@href}%
\providecommand \@href[1]{\@@startlink{#1}\@@href}%
\providecommand \@@href[1]{\endgroup#1\@@endlink}%
\providecommand \@sanitize@url [0]{\catcode `\\12\catcode `\$12\catcode
  `\&12\catcode `\#12\catcode `\^12\catcode `\_12\catcode `\%12\relax}%
\providecommand \@@startlink[1]{}%
\providecommand \@@endlink[0]{}%
\providecommand \url  [0]{\begingroup\@sanitize@url \@url }%
\providecommand \@url [1]{\endgroup\@href {#1}{\urlprefix }}%
\providecommand \urlprefix  [0]{URL }%
\providecommand \Eprint [0]{\href }%
\providecommand \doibase [0]{http://dx.doi.org/}%
\providecommand \selectlanguage [0]{\@gobble}%
\providecommand \bibinfo  [0]{\@secondoftwo}%
\providecommand \bibfield  [0]{\@secondoftwo}%
\providecommand \translation [1]{[#1]}%
\providecommand \BibitemOpen [0]{}%
\providecommand \bibitemStop [0]{}%
\providecommand \bibitemNoStop [0]{.\EOS\space}%
\providecommand \EOS [0]{\spacefactor3000\relax}%
\providecommand \BibitemShut  [1]{\csname bibitem#1\endcsname}%
\let\auto@bib@innerbib\@empty
\bibitem [{\citenamefont {Motter}\ and\ \citenamefont
  {Albert}(2012)}]{ref:motternetworks}%
  \BibitemOpen
  \bibfield  {author} {\bibinfo {author} {\bibfnamefont {A.~E.}\ \bibnamefont
  {Motter}}\ and\ \bibinfo {author} {\bibfnamefont {R.}~\bibnamefont
  {Albert}},\ }\href {\doibase 10.1063/PT.3.1518} {\bibfield  {journal}
  {\bibinfo  {journal} {Physics Today}\ }\textbf {\bibinfo {volume} {65}},\
  \bibinfo {pages} {43} (\bibinfo {year} {2012})}\BibitemShut {NoStop}%
\bibitem [{\citenamefont {Fortunato}(2010)}]{ref:fortunatophysrep}%
  \BibitemOpen
  \bibfield  {author} {\bibinfo {author} {\bibfnamefont {S.}~\bibnamefont
  {Fortunato}},\ }\href {\doibase 10.1016/j.physrep.2010.1.???} {\bibfield
  {journal} {\bibinfo  {journal} {Phys. Rep.}\ }\textbf {\bibinfo {volume}
  {486}},\ \bibinfo {pages} {75} (\bibinfo {year} {2010})}\BibitemShut
  {NoStop}%
\bibitem [{\citenamefont {Lancichinetti}\ \emph {et~al.}(2010)\citenamefont
  {Lancichinetti}, \citenamefont {Kivel{\"a}}, \citenamefont {Saram{\"a}ki},\
  and\ \citenamefont {Fortunato}}]{ref:lanccharacter}%
  \BibitemOpen
  \bibfield  {author} {\bibinfo {author} {\bibfnamefont {A.}~\bibnamefont
  {Lancichinetti}}, \bibinfo {author} {\bibfnamefont {M.}~\bibnamefont
  {Kivel{\"a}}}, \bibinfo {author} {\bibfnamefont {J.}~\bibnamefont
  {Saram{\"a}ki}}, \ and\ \bibinfo {author} {\bibfnamefont {S.}~\bibnamefont
  {Fortunato}},\ }\href {\doibase 10.1371/journal.pone.0011976} {\bibfield
  {journal} {\bibinfo  {journal} {PLoS ONE}\ }\textbf {\bibinfo {volume} {5}},\
  \bibinfo {pages} {e11976} (\bibinfo {year} {2010})}\BibitemShut {NoStop}%
\bibitem [{\citenamefont {Ronhovde}\ \emph {et~al.}(2011)\citenamefont
  {Ronhovde}, \citenamefont {Chakrabarty}, \citenamefont {Hu}, \citenamefont
  {Sahu}, \citenamefont {Sahu}, \citenamefont {Kelton},\ and\ \citenamefont
  {Nussinov}}]{ref:RCHNstructure}%
  \BibitemOpen
  \bibfield  {author} {\bibinfo {author} {\bibfnamefont {P.}~\bibnamefont
  {Ronhovde}}, \bibinfo {author} {\bibfnamefont {S.}~\bibnamefont
  {Chakrabarty}}, \bibinfo {author} {\bibfnamefont {D.}~\bibnamefont {Hu}},
  \bibinfo {author} {\bibfnamefont {M.}~\bibnamefont {Sahu}}, \bibinfo {author}
  {\bibfnamefont {K.~K.}\ \bibnamefont {Sahu}}, \bibinfo {author}
  {\bibfnamefont {K.~F.}\ \bibnamefont {Kelton}}, \ and\ \bibinfo {author}
  {\bibfnamefont {Z.}~\bibnamefont {Nussinov}},\ }\href {\doibase
  10.1140/epje/i2011-11105-9} {\bibfield  {journal} {\bibinfo  {journal} {Euro.
  Phys. J. E}\ }\textbf {\bibinfo {volume} {34}},\ \bibinfo {pages} {105}
  (\bibinfo {year} {2011})}\BibitemShut {NoStop}%
\bibitem [{\citenamefont {Palla}\ \emph {et~al.}(2005)\citenamefont {Palla},
  \citenamefont {Der{\'e}nyi}, \citenamefont {Farkas},\ and\ \citenamefont
  {Vicsek}}]{ref:palla}%
  \BibitemOpen
  \bibfield  {author} {\bibinfo {author} {\bibfnamefont {G.}~\bibnamefont
  {Palla}}, \bibinfo {author} {\bibfnamefont {I.}~\bibnamefont {Der{\'e}nyi}},
  \bibinfo {author} {\bibfnamefont {I.}~\bibnamefont {Farkas}}, \ and\ \bibinfo
  {author} {\bibfnamefont {T.}~\bibnamefont {Vicsek}},\ }\href {\doibase
  10.1038/nature03607} {\bibfield  {journal} {\bibinfo  {journal} {Nature
  (London)}\ }\textbf {\bibinfo {volume} {435}},\ \bibinfo {pages} {814}
  (\bibinfo {year} {2005})}\BibitemShut {NoStop}%
\bibitem [{\citenamefont {Arenas}\ \emph {et~al.}(2008)\citenamefont {Arenas},
  \citenamefont {Fern{\'a}ndez},\ and\ \citenamefont
  {G{\'o}mez}}]{ref:arenasmultires}%
  \BibitemOpen
  \bibfield  {author} {\bibinfo {author} {\bibfnamefont {A.}~\bibnamefont
  {Arenas}}, \bibinfo {author} {\bibfnamefont {A.}~\bibnamefont
  {Fern{\'a}ndez}}, \ and\ \bibinfo {author} {\bibfnamefont {S.}~\bibnamefont
  {G{\'o}mez}},\ }\href {\doibase 10.1088/1367-2630/10/5/053039} {\bibfield
  {journal} {\bibinfo  {journal} {New J. Phys.}\ }\textbf {\bibinfo {volume}
  {10}},\ \bibinfo {pages} {053039} (\bibinfo {year} {2008})}\BibitemShut
  {NoStop}%
\bibitem [{\citenamefont {Kumpula}\ \emph
  {et~al.}(2007{\natexlab{a}})\citenamefont {Kumpula}, \citenamefont
  {Saram{\"a}ki}, \citenamefont {Kaski},\ and\ \citenamefont
  {Kert{\'e}sz}}]{ref:kumpulamultires}%
  \BibitemOpen
  \bibfield  {author} {\bibinfo {author} {\bibfnamefont {J.~M.}\ \bibnamefont
  {Kumpula}}, \bibinfo {author} {\bibfnamefont {J.}~\bibnamefont
  {Saram{\"a}ki}}, \bibinfo {author} {\bibfnamefont {K.}~\bibnamefont {Kaski}},
  \ and\ \bibinfo {author} {\bibfnamefont {J.}~\bibnamefont {Kert{\'e}sz}},\
  }\href {\doibase 10.1142/S0219477507003854} {\bibfield  {journal} {\bibinfo
  {journal} {Fluct. Noise Lett.}\ }\textbf {\bibinfo {volume} {7}},\ \bibinfo
  {pages} {L209} (\bibinfo {year} {2007}{\natexlab{a}})}\BibitemShut {NoStop}%
\bibitem [{\citenamefont {Fenn}\ \emph {et~al.}(2009)\citenamefont {Fenn},
  \citenamefont {Porter}, \citenamefont {McDonald}, \citenamefont {Williams},
  \citenamefont {Johnson},\ and\ \citenamefont {Jones}}]{ref:fenndynamic}%
  \BibitemOpen
  \bibfield  {author} {\bibinfo {author} {\bibfnamefont {D.~J.}\ \bibnamefont
  {Fenn}}, \bibinfo {author} {\bibfnamefont {M.~A.}\ \bibnamefont {Porter}},
  \bibinfo {author} {\bibfnamefont {M.}~\bibnamefont {McDonald}}, \bibinfo
  {author} {\bibfnamefont {S.}~\bibnamefont {Williams}}, \bibinfo {author}
  {\bibfnamefont {N.~F.}\ \bibnamefont {Johnson}}, \ and\ \bibinfo {author}
  {\bibfnamefont {N.~S.}\ \bibnamefont {Jones}},\ }\href {\doibase
  10.1063/1.3184538} {\bibfield  {journal} {\bibinfo  {journal} {Chaos}\
  }\textbf {\bibinfo {volume} {19}},\ \bibinfo {pages} {033119} (\bibinfo
  {year} {2009})}\BibitemShut {NoStop}%
\bibitem [{\citenamefont {Rosvall}\ and\ \citenamefont
  {Bergstrom}(2011)}]{ref:rosvallmultires}%
  \BibitemOpen
  \bibfield  {author} {\bibinfo {author} {\bibfnamefont {M.}~\bibnamefont
  {Rosvall}}\ and\ \bibinfo {author} {\bibfnamefont {C.~T.}\ \bibnamefont
  {Bergstrom}},\ }\href {\doibase 10.1371/journal.pone.0018209} {\bibfield
  {journal} {\bibinfo  {journal} {PLoS ONE}\ }\textbf {\bibinfo {volume} {6}},\
  \bibinfo {pages} {e18209} (\bibinfo {year} {2011})}\BibitemShut {NoStop}%
\bibitem [{\citenamefont {Cheng}\ and\ \citenamefont
  {Shen}(2010)}]{ref:chengshen}%
  \BibitemOpen
  \bibfield  {author} {\bibinfo {author} {\bibfnamefont {X.-Q.}\ \bibnamefont
  {Cheng}}\ and\ \bibinfo {author} {\bibfnamefont {H.-W.}\ \bibnamefont
  {Shen}},\ }\href {\doibase 10.1088/1742-5468/2010/04/P04024} {\bibfield
  {journal} {\bibinfo  {journal} {J. Stat. Mech.}\ }\textbf {\bibinfo {volume}
  {04}},\ \bibinfo {pages} {P04024} (\bibinfo {year} {2010})}\BibitemShut
  {NoStop}%
\bibitem [{\citenamefont {Ronhovde}\ and\ \citenamefont
  {Nussinov}(2009)}]{ref:rzmultires}%
  \BibitemOpen
  \bibfield  {author} {\bibinfo {author} {\bibfnamefont {P.}~\bibnamefont
  {Ronhovde}}\ and\ \bibinfo {author} {\bibfnamefont {Z.}~\bibnamefont
  {Nussinov}},\ }\href {\doibase 10.1103/PhysRevE.80.016109} {\bibfield
  {journal} {\bibinfo  {journal} {Phys. Rev. E}\ }\textbf {\bibinfo {volume}
  {80}},\ \bibinfo {pages} {016109} (\bibinfo {year} {2009})}\BibitemShut
  {NoStop}%
\bibitem [{\citenamefont {Danon}\ \emph {et~al.}(2005)\citenamefont {Danon},
  \citenamefont {D{\'i}az-Guilera}, \citenamefont {Duch},\ and\ \citenamefont
  {Arenas}}]{ref:danon}%
  \BibitemOpen
  \bibfield  {author} {\bibinfo {author} {\bibfnamefont {L.}~\bibnamefont
  {Danon}}, \bibinfo {author} {\bibfnamefont {A.}~\bibnamefont
  {D{\'i}az-Guilera}}, \bibinfo {author} {\bibfnamefont {J.}~\bibnamefont
  {Duch}}, \ and\ \bibinfo {author} {\bibfnamefont {A.}~\bibnamefont
  {Arenas}},\ }\href {\doibase 10.1088/1742-5468/2005/09/P09008} {\bibfield
  {journal} {\bibinfo  {journal} {J. Stat. Mech.}\ }\textbf {\bibinfo {volume}
  {09}},\ \bibinfo {pages} {P09008} (\bibinfo {year} {2005})}\BibitemShut
  {NoStop}%
\bibitem [{\citenamefont {Meil{\u a}}(2007)}]{ref:vi}%
  \BibitemOpen
  \bibfield  {author} {\bibinfo {author} {\bibfnamefont {M.}~\bibnamefont
  {Meil{\u a}}},\ }\href {\doibase 10.1016/j.jmva.2006.11.013} {\bibfield
  {journal} {\bibinfo  {journal} {J. Multivariate Anal.}\ }\textbf {\bibinfo
  {volume} {98}},\ \bibinfo {pages} {873} (\bibinfo {year} {2007})}\BibitemShut
  {NoStop}%
\bibitem [{\citenamefont {Newman}\ and\ \citenamefont {Girvan}(2004)}]{ref:gn}%
  \BibitemOpen
  \bibfield  {author} {\bibinfo {author} {\bibfnamefont {M.~E.~J.}\
  \bibnamefont {Newman}}\ and\ \bibinfo {author} {\bibfnamefont
  {M.}~\bibnamefont {Girvan}},\ }\href {\doibase 10.1103/PhysRevE.69.026113}
  {\bibfield  {journal} {\bibinfo  {journal} {Phys. Rev. E}\ }\textbf {\bibinfo
  {volume} {69}},\ \bibinfo {pages} {026113} (\bibinfo {year}
  {2004})}\BibitemShut {NoStop}%
\bibitem [{\citenamefont {Raghavan}\ \emph {et~al.}(2007)\citenamefont
  {Raghavan}, \citenamefont {Albert},\ and\ \citenamefont {Kumara}}]{ref:LPA}%
  \BibitemOpen
  \bibfield  {author} {\bibinfo {author} {\bibfnamefont {U.~N.}\ \bibnamefont
  {Raghavan}}, \bibinfo {author} {\bibfnamefont {R.}~\bibnamefont {Albert}}, \
  and\ \bibinfo {author} {\bibfnamefont {S.}~\bibnamefont {Kumara}},\ }\href
  {\doibase 10.1103/PhysRevE.76.036106} {\bibfield  {journal} {\bibinfo
  {journal} {Phys. Rev. E}\ }\textbf {\bibinfo {volume} {76}},\ \bibinfo
  {pages} {036106} (\bibinfo {year} {2007})}\BibitemShut {NoStop}%
\bibitem [{\citenamefont {Barber}\ and\ \citenamefont
  {Clark}(2009)}]{ref:barberLPA}%
  \BibitemOpen
  \bibfield  {author} {\bibinfo {author} {\bibfnamefont {M.~J.}\ \bibnamefont
  {Barber}}\ and\ \bibinfo {author} {\bibfnamefont {J.~W.}\ \bibnamefont
  {Clark}},\ }\href {\doibase 10.1103/PhysRevE.80.026129} {\bibfield  {journal}
  {\bibinfo  {journal} {Phys. Rev. E}\ }\textbf {\bibinfo {volume} {80}},\
  \bibinfo {pages} {026129} (\bibinfo {year} {2009})}\BibitemShut {NoStop}%
\bibitem [{\citenamefont {Reichardt}\ and\ \citenamefont
  {Bornholdt}(2006)}]{ref:smcd}%
  \BibitemOpen
  \bibfield  {author} {\bibinfo {author} {\bibfnamefont {J.}~\bibnamefont
  {Reichardt}}\ and\ \bibinfo {author} {\bibfnamefont {S.}~\bibnamefont
  {Bornholdt}},\ }\href {\doibase 10.1103/PhysRevE.74.016110} {\bibfield
  {journal} {\bibinfo  {journal} {Phys. Rev. E}\ }\textbf {\bibinfo {volume}
  {74}},\ \bibinfo {pages} {016110} (\bibinfo {year} {2006})}\BibitemShut
  {NoStop}%
\bibitem [{\citenamefont {Reichardt}\ and\ \citenamefont
  {Bornholdt}(2004)}]{ref:reichardt}%
  \BibitemOpen
  \bibfield  {author} {\bibinfo {author} {\bibfnamefont {J.}~\bibnamefont
  {Reichardt}}\ and\ \bibinfo {author} {\bibfnamefont {S.}~\bibnamefont
  {Bornholdt}},\ }\href {\doibase 10.1103/PhysRevLett.93.218701} {\bibfield
  {journal} {\bibinfo  {journal} {Phys. Rev. Lett.}\ }\textbf {\bibinfo
  {volume} {93}},\ \bibinfo {pages} {218701} (\bibinfo {year}
  {2004})}\BibitemShut {NoStop}%
\bibitem [{\citenamefont {Fortunato}\ and\ \citenamefont
  {Barth{\'e}lemy}(2007)}]{ref:fortunato}%
  \BibitemOpen
  \bibfield  {author} {\bibinfo {author} {\bibfnamefont {S.}~\bibnamefont
  {Fortunato}}\ and\ \bibinfo {author} {\bibfnamefont {M.}~\bibnamefont
  {Barth{\'e}lemy}},\ }\href {\doibase 10.1073/pnas.0605965104} {\bibfield
  {journal} {\bibinfo  {journal} {Proc. Natl. Aca. Sci. U.S.A.}\ }\textbf
  {\bibinfo {volume} {104}},\ \bibinfo {pages} {36} (\bibinfo {year}
  {2007})}\BibitemShut {NoStop}%
\bibitem [{\citenamefont {Kumpula}\ \emph
  {et~al.}(2007{\natexlab{b}})\citenamefont {Kumpula}, \citenamefont
  {Saram{\"a}ki}, \citenamefont {Kaski},\ and\ \citenamefont
  {Kert{\'e}sz}}]{ref:kumpulaResLim}%
  \BibitemOpen
  \bibfield  {author} {\bibinfo {author} {\bibfnamefont {J.~M.}\ \bibnamefont
  {Kumpula}}, \bibinfo {author} {\bibfnamefont {J.}~\bibnamefont
  {Saram{\"a}ki}}, \bibinfo {author} {\bibfnamefont {K.}~\bibnamefont {Kaski}},
  \ and\ \bibinfo {author} {\bibfnamefont {J.}~\bibnamefont {Kert{\'e}sz}},\
  }\href {\doibase 10.1140/epjb/e2007-00088-4} {\bibfield  {journal} {\bibinfo
  {journal} {Euro. Phys. J. B}\ }\textbf {\bibinfo {volume} {56}},\ \bibinfo
  {pages} {41} (\bibinfo {year} {2007}{\natexlab{b}})}\BibitemShut {NoStop}%
\bibitem [{\citenamefont {Zhang}\ \emph
  {et~al.}(2009{\natexlab{a}})\citenamefont {Zhang}, \citenamefont {Wang},
  \citenamefont {Wang}, \citenamefont {Wang}, \citenamefont {Qiu},
  \citenamefont {Wang},\ and\ \citenamefont {Chen}}]{ref:zhangmodweak}%
  \BibitemOpen
  \bibfield  {author} {\bibinfo {author} {\bibfnamefont {X.~S.}\ \bibnamefont
  {Zhang}}, \bibinfo {author} {\bibfnamefont {R.~S.}\ \bibnamefont {Wang}},
  \bibinfo {author} {\bibfnamefont {Y.}~\bibnamefont {Wang}}, \bibinfo {author}
  {\bibfnamefont {J.}~\bibnamefont {Wang}}, \bibinfo {author} {\bibfnamefont
  {Y.}~\bibnamefont {Qiu}}, \bibinfo {author} {\bibfnamefont {L.}~\bibnamefont
  {Wang}}, \ and\ \bibinfo {author} {\bibfnamefont {L.}~\bibnamefont {Chen}},\
  }\href {\doibase 10.1209/0295-5075/87/38002} {\bibfield  {journal} {\bibinfo
  {journal} {Europhys. Lett.}\ }\textbf {\bibinfo {volume} {87}},\ \bibinfo
  {pages} {38002} (\bibinfo {year} {2009}{\natexlab{a}})}\BibitemShut {NoStop}%
\bibitem [{\citenamefont {Lancichinetti}\ and\ \citenamefont
  {Fortunato}(2011)}]{ref:lancfortunatomod}%
  \BibitemOpen
  \bibfield  {author} {\bibinfo {author} {\bibfnamefont {A.}~\bibnamefont
  {Lancichinetti}}\ and\ \bibinfo {author} {\bibfnamefont {S.}~\bibnamefont
  {Fortunato}},\ }\href {\doibase 10.1103/PhysRevE.84.066122} {\bibfield
  {journal} {\bibinfo  {journal} {Phys. Rev. E}\ }\textbf {\bibinfo {volume}
  {84}},\ \bibinfo {pages} {066122} (\bibinfo {year} {2011})}\BibitemShut
  {NoStop}%
\bibitem [{\citenamefont {Xiang}\ and\ \citenamefont
  {Hu}(2012)}]{ref:xiangmultireslimit}%
  \BibitemOpen
  \bibfield  {author} {\bibinfo {author} {\bibfnamefont {J.}~\bibnamefont
  {Xiang}}\ and\ \bibinfo {author} {\bibfnamefont {K.}~\bibnamefont {Hu}},\
  }\href {\doibase 10.1016/j.physa.2012.05.006} {\bibfield  {journal} {\bibinfo
   {journal} {Physica A}\ }\textbf {\bibinfo {volume} {391}},\ \bibinfo {pages}
  {4995} (\bibinfo {year} {2012})}\BibitemShut {NoStop}%
\bibitem [{\citenamefont {Holland}\ \emph {et~al.}(1983)\citenamefont
  {Holland}, \citenamefont {Laskey},\ and\ \citenamefont
  {Leinhardt}}]{ref:holland1983stochastic}%
  \BibitemOpen
  \bibfield  {author} {\bibinfo {author} {\bibfnamefont {P.~W.}\ \bibnamefont
  {Holland}}, \bibinfo {author} {\bibfnamefont {K.~B.}\ \bibnamefont {Laskey}},
  \ and\ \bibinfo {author} {\bibfnamefont {S.}~\bibnamefont {Leinhardt}},\
  }\href@noop {} {\bibfield  {journal} {\bibinfo  {journal} {Social networks}\
  }\textbf {\bibinfo {volume} {5}},\ \bibinfo {pages} {109} (\bibinfo {year}
  {1983})}\BibitemShut {NoStop}%
\bibitem [{\citenamefont {Snijders}\ and\ \citenamefont
  {Nowicki}(1997)}]{ref:snijders1997estimation}%
  \BibitemOpen
  \bibfield  {author} {\bibinfo {author} {\bibfnamefont {T.~A.}\ \bibnamefont
  {Snijders}}\ and\ \bibinfo {author} {\bibfnamefont {K.}~\bibnamefont
  {Nowicki}},\ }\href@noop {} {\bibfield  {journal} {\bibinfo  {journal} {J.
  classification}\ }\textbf {\bibinfo {volume} {14}},\ \bibinfo {pages} {75}
  (\bibinfo {year} {1997})}\BibitemShut {NoStop}%
\bibitem [{\citenamefont {Newman}\ and\ \citenamefont
  {Leicht}(2007)}]{ref:newman2007mixture}%
  \BibitemOpen
  \bibfield  {author} {\bibinfo {author} {\bibfnamefont {M.~E.}\ \bibnamefont
  {Newman}}\ and\ \bibinfo {author} {\bibfnamefont {E.~A.}\ \bibnamefont
  {Leicht}},\ }\href@noop {} {\bibfield  {journal} {\bibinfo  {journal} {PNAS}\
  }\textbf {\bibinfo {volume} {104}},\ \bibinfo {pages} {9564} (\bibinfo {year}
  {2007})}\BibitemShut {NoStop}%
\bibitem [{\citenamefont {Latouche}\ \emph {et~al.}(2012)\citenamefont
  {Latouche}, \citenamefont {Birmele},\ and\ \citenamefont
  {Ambroise}}]{ref:latouche2012variational}%
  \BibitemOpen
  \bibfield  {author} {\bibinfo {author} {\bibfnamefont {P.}~\bibnamefont
  {Latouche}}, \bibinfo {author} {\bibfnamefont {E.}~\bibnamefont {Birmele}}, \
  and\ \bibinfo {author} {\bibfnamefont {C.}~\bibnamefont {Ambroise}},\
  }\href@noop {} {\bibfield  {journal} {\bibinfo  {journal} {Statistical
  Modelling}\ }\textbf {\bibinfo {volume} {12}},\ \bibinfo {pages} {93}
  (\bibinfo {year} {2012})}\BibitemShut {NoStop}%
\bibitem [{\citenamefont {Decelle}\ \emph
  {et~al.}(2011{\natexlab{a}})\citenamefont {Decelle}, \citenamefont
  {Krzakala}, \citenamefont {Moore},\ and\ \citenamefont
  {Zdeborov{\'a}}}]{ref:decelleKMZPT}%
  \BibitemOpen
  \bibfield  {author} {\bibinfo {author} {\bibfnamefont {A.}~\bibnamefont
  {Decelle}}, \bibinfo {author} {\bibfnamefont {F.}~\bibnamefont {Krzakala}},
  \bibinfo {author} {\bibfnamefont {C.}~\bibnamefont {Moore}}, \ and\ \bibinfo
  {author} {\bibfnamefont {L.}~\bibnamefont {Zdeborov{\'a}}},\ }\href {\doibase
  10.1103/PhysRevLett.107.065701} {\bibfield  {journal} {\bibinfo  {journal}
  {Phys. Rev. Lett.}\ }\textbf {\bibinfo {volume} {107}},\ \bibinfo {pages}
  {065701} (\bibinfo {year} {2011}{\natexlab{a}})}\BibitemShut {NoStop}%
\bibitem [{\citenamefont {Decelle}\ \emph
  {et~al.}(2011{\natexlab{b}})\citenamefont {Decelle}, \citenamefont
  {Krzakala}, \citenamefont {Moore},\ and\ \citenamefont
  {Zdeborov{\'a}}}]{ref:decelle2011asymptotic}%
  \BibitemOpen
  \bibfield  {author} {\bibinfo {author} {\bibfnamefont {A.}~\bibnamefont
  {Decelle}}, \bibinfo {author} {\bibfnamefont {F.}~\bibnamefont {Krzakala}},
  \bibinfo {author} {\bibfnamefont {C.}~\bibnamefont {Moore}}, \ and\ \bibinfo
  {author} {\bibfnamefont {L.}~\bibnamefont {Zdeborov{\'a}}},\ }\href {\doibase
  10.1103/PhysRevE.84.066106} {\bibfield  {journal} {\bibinfo  {journal} {Phys.
  Rev. E}\ }\textbf {\bibinfo {volume} {84}},\ \bibinfo {pages} {066106}
  (\bibinfo {year} {2011}{\natexlab{b}})}\BibitemShut {NoStop}%
\bibitem [{\citenamefont {Hu}\ \emph {et~al.}(2012{\natexlab{a}})\citenamefont
  {Hu}, \citenamefont {Ronhovde},\ and\ \citenamefont
  {Nussinov}}]{ref:hu2012phase}%
  \BibitemOpen
  \bibfield  {author} {\bibinfo {author} {\bibfnamefont {D.}~\bibnamefont
  {Hu}}, \bibinfo {author} {\bibfnamefont {P.}~\bibnamefont {Ronhovde}}, \ and\
  \bibinfo {author} {\bibfnamefont {Z.}~\bibnamefont {Nussinov}},\ }\href@noop
  {} {\bibfield  {journal} {\bibinfo  {journal} {Philosophical Magazine}\
  }\textbf {\bibinfo {volume} {92}},\ \bibinfo {pages} {406} (\bibinfo {year}
  {2012}{\natexlab{a}})}\BibitemShut {NoStop}%
\bibitem [{\citenamefont {Hu}\ \emph {et~al.}(2012{\natexlab{b}})\citenamefont
  {Hu}, \citenamefont {Ronhovde},\ and\ \citenamefont
  {Nussinov}}]{ref:huCDPTlong}%
  \BibitemOpen
  \bibfield  {author} {\bibinfo {author} {\bibfnamefont {D.}~\bibnamefont
  {Hu}}, \bibinfo {author} {\bibfnamefont {P.}~\bibnamefont {Ronhovde}}, \ and\
  \bibinfo {author} {\bibfnamefont {Z.}~\bibnamefont {Nussinov}},\ }\href
  {\doibase 10.1103/PhysRevE.86.066106} {\bibfield  {journal} {\bibinfo
  {journal} {Phys. Rev. E}\ }\textbf {\bibinfo {volume} {86}},\ \bibinfo
  {pages} {066106} (\bibinfo {year} {2012}{\natexlab{b}})}\BibitemShut
  {NoStop}%
\bibitem [{\citenamefont {Darst}\ \emph {et~al.}(2014)\citenamefont {Darst},
  \citenamefont {Reichman}, \citenamefont {Ronhovde},\ and\ \citenamefont
  {Nussinov}}]{ref:darstCDSBMdef}%
  \BibitemOpen
  \bibfield  {author} {\bibinfo {author} {\bibfnamefont {R.~K.}\ \bibnamefont
  {Darst}}, \bibinfo {author} {\bibfnamefont {D.~R.}\ \bibnamefont {Reichman}},
  \bibinfo {author} {\bibfnamefont {P.}~\bibnamefont {Ronhovde}}, \ and\
  \bibinfo {author} {\bibfnamefont {Z.}~\bibnamefont {Nussinov}},\ }\href
  {\doibase 10.1093/comnet/cnu042} {\bibfield  {journal} {\bibinfo  {journal}
  {J. Complex Networks}\ }\textbf {\bibinfo {volume} {2}} (\bibinfo {year}
  {2014}),\ 10.1093/comnet/cnu042}\BibitemShut {NoStop}%
\bibitem [{\citenamefont {Karrer}\ and\ \citenamefont
  {Newman}(2011)}]{ref:karrer2011stochastic}%
  \BibitemOpen
  \bibfield  {author} {\bibinfo {author} {\bibfnamefont {B.}~\bibnamefont
  {Karrer}}\ and\ \bibinfo {author} {\bibfnamefont {M.~E.~J.}\ \bibnamefont
  {Newman}},\ }\href {\doibase 10.1103/PhysRevE.83.016107} {\bibfield
  {journal} {\bibinfo  {journal} {Phys. Rev. E}\ }\textbf {\bibinfo {volume}
  {83}},\ \bibinfo {pages} {016107} (\bibinfo {year} {2011})}\BibitemShut
  {NoStop}%
\bibitem [{\citenamefont {Zhu}\ \emph {et~al.}(2014)\citenamefont {Zhu},
  \citenamefont {Yan},\ and\ \citenamefont {Moore}}]{ref:zhu2014oriented}%
  \BibitemOpen
  \bibfield  {author} {\bibinfo {author} {\bibfnamefont {Y.}~\bibnamefont
  {Zhu}}, \bibinfo {author} {\bibfnamefont {X.}~\bibnamefont {Yan}}, \ and\
  \bibinfo {author} {\bibfnamefont {C.}~\bibnamefont {Moore}},\ }\href@noop {}
  {\bibfield  {journal} {\bibinfo  {journal} {J. Complex Networks}\ }\textbf
  {\bibinfo {volume} {2}},\ \bibinfo {pages} {1} (\bibinfo {year}
  {2014})}\BibitemShut {NoStop}%
\bibitem [{\citenamefont {Yan}\ \emph {et~al.}(2014)\citenamefont {Yan},
  \citenamefont {Shalizi}, \citenamefont {Jensen}, \citenamefont {Krzakala},
  \citenamefont {Moore}, \citenamefont {Zdeborov{\' a}}, \citenamefont
  {Zhang},\ and\ \citenamefont {Zhu}}]{ref:yan2014model}%
  \BibitemOpen
  \bibfield  {author} {\bibinfo {author} {\bibfnamefont {X.}~\bibnamefont
  {Yan}}, \bibinfo {author} {\bibfnamefont {C.}~\bibnamefont {Shalizi}},
  \bibinfo {author} {\bibfnamefont {J.~E.}\ \bibnamefont {Jensen}}, \bibinfo
  {author} {\bibfnamefont {F.}~\bibnamefont {Krzakala}}, \bibinfo {author}
  {\bibfnamefont {C.}~\bibnamefont {Moore}}, \bibinfo {author} {\bibfnamefont
  {L.}~\bibnamefont {Zdeborov{\' a}}}, \bibinfo {author} {\bibfnamefont
  {P.}~\bibnamefont {Zhang}}, \ and\ \bibinfo {author} {\bibfnamefont
  {Y.}~\bibnamefont {Zhu}},\ }\href@noop {} {\bibfield  {journal} {\bibinfo
  {journal} {J. Stat. Mech.: Theory and Exp.}\ }\textbf {\bibinfo {volume}
  {2014}},\ \bibinfo {pages} {P05007} (\bibinfo {year} {2014})}\BibitemShut
  {NoStop}%
\bibitem [{\citenamefont {Blatt}\ \emph {et~al.}(1996)\citenamefont {Blatt},
  \citenamefont {Wiseman},\ and\ \citenamefont {Domany}}]{ref:blatt}%
  \BibitemOpen
  \bibfield  {author} {\bibinfo {author} {\bibfnamefont {M.}~\bibnamefont
  {Blatt}}, \bibinfo {author} {\bibfnamefont {S.}~\bibnamefont {Wiseman}}, \
  and\ \bibinfo {author} {\bibfnamefont {E.}~\bibnamefont {Domany}},\ }\href
  {\doibase 10.1103/PhysRevLett.76.3251} {\bibfield  {journal} {\bibinfo
  {journal} {Phys. Rev. Lett.}\ }\textbf {\bibinfo {volume} {76}},\ \bibinfo
  {pages} {3251} (\bibinfo {year} {1996})}\BibitemShut {NoStop}%
\bibitem [{\citenamefont {Ispolatov}\ \emph {et~al.}(2006)\citenamefont
  {Ispolatov}, \citenamefont {Mazo},\ and\ \citenamefont
  {Yuryev}}]{ref:ispolatov}%
  \BibitemOpen
  \bibfield  {author} {\bibinfo {author} {\bibfnamefont {I.}~\bibnamefont
  {Ispolatov}}, \bibinfo {author} {\bibfnamefont {I.}~\bibnamefont {Mazo}}, \
  and\ \bibinfo {author} {\bibfnamefont {A.}~\bibnamefont {Yuryev}},\ }\href
  {\doibase 10.1088/1742-5468/2006/09/P09014} {\bibfield  {journal} {\bibinfo
  {journal} {J. Stat. Mech.}\ }\textbf {\bibinfo {volume} {09}},\ \bibinfo
  {pages} {P09014} (\bibinfo {year} {2006})}\BibitemShut {NoStop}%
\bibitem [{\citenamefont {Hastings}(2006)}]{ref:hastings}%
  \BibitemOpen
  \bibfield  {author} {\bibinfo {author} {\bibfnamefont {M.~B.}\ \bibnamefont
  {Hastings}},\ }\href {\doibase 10.1103/PhysRevE.74.035102} {\bibfield
  {journal} {\bibinfo  {journal} {Phys. Rev. E}\ }\textbf {\bibinfo {volume}
  {74}},\ \bibinfo {pages} {035102} (\bibinfo {year} {2006})}\BibitemShut
  {NoStop}%
\bibitem [{\citenamefont {Traag}\ and\ \citenamefont
  {Bruggeman}(2009)}]{ref:traagPRE}%
  \BibitemOpen
  \bibfield  {author} {\bibinfo {author} {\bibfnamefont {V.~A.}\ \bibnamefont
  {Traag}}\ and\ \bibinfo {author} {\bibfnamefont {J.}~\bibnamefont
  {Bruggeman}},\ }\href {\doibase 10.1103/PhysRevE.80.036115} {\bibfield
  {journal} {\bibinfo  {journal} {Phys. Rev. E}\ }\textbf {\bibinfo {volume}
  {80}},\ \bibinfo {pages} {036115} (\bibinfo {year} {2009})}\BibitemShut
  {NoStop}%
\bibitem [{\citenamefont {Traag}\ \emph {et~al.}(2011)\citenamefont {Traag},
  \citenamefont {Van~Dooren},\ and\ \citenamefont
  {Nesterov}}]{ref:traaglocalscope}%
  \BibitemOpen
  \bibfield  {author} {\bibinfo {author} {\bibfnamefont {V.~A.}\ \bibnamefont
  {Traag}}, \bibinfo {author} {\bibfnamefont {P.}~\bibnamefont {Van~Dooren}}, \
  and\ \bibinfo {author} {\bibfnamefont {Y.}~\bibnamefont {Nesterov}},\ }\href
  {\doibase 10.1103/PhysRevE.84.016114} {\bibfield  {journal} {\bibinfo
  {journal} {Phys. Rev. E}\ }\textbf {\bibinfo {volume} {84}},\ \bibinfo
  {pages} {016114} (\bibinfo {year} {2011})}\BibitemShut {NoStop}%
\bibitem [{\citenamefont {Ronhovde}\ and\ \citenamefont
  {Nussinov}(2010)}]{ref:rzlocal}%
  \BibitemOpen
  \bibfield  {author} {\bibinfo {author} {\bibfnamefont {P.}~\bibnamefont
  {Ronhovde}}\ and\ \bibinfo {author} {\bibfnamefont {Z.}~\bibnamefont
  {Nussinov}},\ }\href {\doibase 10.1103/PhysRevE.81.046114} {\bibfield
  {journal} {\bibinfo  {journal} {Phys. Rev. E}\ }\textbf {\bibinfo {volume}
  {81}},\ \bibinfo {pages} {046114} (\bibinfo {year} {2010})}\BibitemShut
  {NoStop}%
\bibitem [{\citenamefont {Bagrow}\ and\ \citenamefont
  {Bollt}(2005)}]{ref:bagrowboltlocal}%
  \BibitemOpen
  \bibfield  {author} {\bibinfo {author} {\bibfnamefont {J.~P.}\ \bibnamefont
  {Bagrow}}\ and\ \bibinfo {author} {\bibfnamefont {E.~M.}\ \bibnamefont
  {Bollt}},\ }\href {\doibase 10.1103/PhysRevE.72.046108} {\bibfield  {journal}
  {\bibinfo  {journal} {Phys. Rev. E}\ }\textbf {\bibinfo {volume} {72}},\
  \bibinfo {pages} {046108} (\bibinfo {year} {2005})}\BibitemShut {NoStop}%
\bibitem [{\citenamefont {Lancichinetti}\ \emph {et~al.}(2009)\citenamefont
  {Lancichinetti}, \citenamefont {Fortunato},\ and\ \citenamefont
  {Kert{\'e}sz}}]{ref:lanc}%
  \BibitemOpen
  \bibfield  {author} {\bibinfo {author} {\bibfnamefont {A.}~\bibnamefont
  {Lancichinetti}}, \bibinfo {author} {\bibfnamefont {S.}~\bibnamefont
  {Fortunato}}, \ and\ \bibinfo {author} {\bibfnamefont {J.}~\bibnamefont
  {Kert{\'e}sz}},\ }\href {\doibase doi:10.1088/1367-2630/11/3/033015}
  {\bibfield  {journal} {\bibinfo  {journal} {New J. Phys.}\ }\textbf {\bibinfo
  {volume} {11}},\ \bibinfo {pages} {033015} (\bibinfo {year}
  {2009})}\BibitemShut {NoStop}%
\bibitem [{\citenamefont {Havemann}\ \emph {et~al.}(2011)\citenamefont
  {Havemann}, \citenamefont {Heinz}, \citenamefont {Struck},\ and\
  \citenamefont {Gl{\"a}ser}}]{ref:havemannlocalmultires}%
  \BibitemOpen
  \bibfield  {author} {\bibinfo {author} {\bibfnamefont {F.}~\bibnamefont
  {Havemann}}, \bibinfo {author} {\bibfnamefont {M.}~\bibnamefont {Heinz}},
  \bibinfo {author} {\bibfnamefont {A.}~\bibnamefont {Struck}}, \ and\ \bibinfo
  {author} {\bibfnamefont {J.}~\bibnamefont {Gl{\"a}ser}},\ }\href {\doibase
  10.1088/1742-5468/2011/01/P01023} {\bibfield  {journal} {\bibinfo  {journal}
  {J. Stat. Mech.}\ }\textbf {\bibinfo {volume} {01}},\ \bibinfo {pages}
  {P01023} (\bibinfo {year} {2011})}\BibitemShut {NoStop}%
\bibitem [{\citenamefont {Zhao}\ \emph {et~al.}(2011)\citenamefont {Zhao},
  \citenamefont {Levina},\ and\ \citenamefont {Zhu}}]{ref:zhao2011community}%
  \BibitemOpen
  \bibfield  {author} {\bibinfo {author} {\bibfnamefont {Y.}~\bibnamefont
  {Zhao}}, \bibinfo {author} {\bibfnamefont {E.}~\bibnamefont {Levina}}, \ and\
  \bibinfo {author} {\bibfnamefont {J.}~\bibnamefont {Zhu}},\ }\href@noop {}
  {\bibfield  {journal} {\bibinfo  {journal} {PNAS}\ }\textbf {\bibinfo
  {volume} {108}},\ \bibinfo {pages} {7321} (\bibinfo {year}
  {2011})}\BibitemShut {NoStop}%
\bibitem [{\citenamefont {Clauset}(2005)}]{ref:clausetlocal}%
  \BibitemOpen
  \bibfield  {author} {\bibinfo {author} {\bibfnamefont {A.}~\bibnamefont
  {Clauset}},\ }\href {\doibase 10.1103/PhysRevE.72.026132} {\bibfield
  {journal} {\bibinfo  {journal} {Phys. Rev. E}\ }\textbf {\bibinfo {volume}
  {72}},\ \bibinfo {pages} {026132} (\bibinfo {year} {2005})}\BibitemShut
  {NoStop}%
\bibitem [{\citenamefont {Muff}\ \emph {et~al.}(2005)\citenamefont {Muff},
  \citenamefont {Rao},\ and\ \citenamefont {Caflisch}}]{ref:muff}%
  \BibitemOpen
  \bibfield  {author} {\bibinfo {author} {\bibfnamefont {S.}~\bibnamefont
  {Muff}}, \bibinfo {author} {\bibfnamefont {F.}~\bibnamefont {Rao}}, \ and\
  \bibinfo {author} {\bibfnamefont {A.}~\bibnamefont {Caflisch}},\ }\href
  {\doibase 10.1103/PhysRevE.72.056107} {\bibfield  {journal} {\bibinfo
  {journal} {Phys. Rev. E}\ }\textbf {\bibinfo {volume} {72}},\ \bibinfo
  {pages} {056107} (\bibinfo {year} {2005})}\BibitemShut {NoStop}%
\bibitem [{\citenamefont {Hu}\ \emph {et~al.}(2012{\natexlab{c}})\citenamefont
  {Hu}, \citenamefont {Ronhovde},\ and\ \citenamefont
  {Nussinov}}]{ref:huCDPTsgd}%
  \BibitemOpen
  \bibfield  {author} {\bibinfo {author} {\bibfnamefont {D.}~\bibnamefont
  {Hu}}, \bibinfo {author} {\bibfnamefont {P.}~\bibnamefont {Ronhovde}}, \ and\
  \bibinfo {author} {\bibfnamefont {Z.}~\bibnamefont {Nussinov}},\ }\href
  {\doibase 10.1080/14786435.2011.616547} {\bibfield  {journal} {\bibinfo
  {journal} {Phil. Mag.}\ }\textbf {\bibinfo {volume} {92}},\ \bibinfo {pages}
  {406} (\bibinfo {year} {2012}{\natexlab{c}})}\BibitemShut {NoStop}%
\bibitem [{\citenamefont {Good}\ \emph {et~al.}(2010)\citenamefont {Good},
  \citenamefont {de~Montjoye},\ and\ \citenamefont {Clauset}}]{ref:goodMC}%
  \BibitemOpen
  \bibfield  {author} {\bibinfo {author} {\bibfnamefont {B.~H.}\ \bibnamefont
  {Good}}, \bibinfo {author} {\bibfnamefont {Y.-A.}\ \bibnamefont
  {de~Montjoye}}, \ and\ \bibinfo {author} {\bibfnamefont {A.}~\bibnamefont
  {Clauset}},\ }\href {\doibase 10.1103/PhysRevE.81.046106} {\bibfield
  {journal} {\bibinfo  {journal} {Phys. Rev. E}\ }\textbf {\bibinfo {volume}
  {81}},\ \bibinfo {pages} {046106} (\bibinfo {year} {2010})}\BibitemShut
  {NoStop}%
\bibitem [{\citenamefont {Nadakuditi}\ and\ \citenamefont
  {Newman}(2012)}]{ref:nadakuditiSBM}%
  \BibitemOpen
  \bibfield  {author} {\bibinfo {author} {\bibfnamefont {R.~R.}\ \bibnamefont
  {Nadakuditi}}\ and\ \bibinfo {author} {\bibfnamefont {M.~E.~J.}\ \bibnamefont
  {Newman}},\ }\href {\doibase 10.1103/PhysRevLett.108.188701} {\bibfield
  {journal} {\bibinfo  {journal} {Phys. Rev. Lett.}\ }\textbf {\bibinfo
  {volume} {108}},\ \bibinfo {pages} {188701} (\bibinfo {year}
  {2012})}\BibitemShut {NoStop}%
\bibitem [{\citenamefont {Ronhovde}\ \emph {et~al.}(2012)\citenamefont
  {Ronhovde}, \citenamefont {Hu},\ and\ \citenamefont
  {Nussinov}}]{ref:rhzglobaldisorder}%
  \BibitemOpen
  \bibfield  {author} {\bibinfo {author} {\bibfnamefont {P.}~\bibnamefont
  {Ronhovde}}, \bibinfo {author} {\bibfnamefont {D.}~\bibnamefont {Hu}}, \ and\
  \bibinfo {author} {\bibfnamefont {Z.}~\bibnamefont {Nussinov}},\ }\href
  {\doibase 10.1209/0295-5075/99/38006} {\bibfield  {journal} {\bibinfo
  {journal} {EPL}\ }\textbf {\bibinfo {volume} {99}},\ \bibinfo {pages} {38006}
  (\bibinfo {year} {2012})}\BibitemShut {NoStop}%
\bibitem [{\citenamefont {Danon}\ \emph {et~al.}(2006)\citenamefont {Danon},
  \citenamefont {D{\'i}az-Guilera},\ and\ \citenamefont
  {Arenas}}]{ref:danonhetero}%
  \BibitemOpen
  \bibfield  {author} {\bibinfo {author} {\bibfnamefont {L.}~\bibnamefont
  {Danon}}, \bibinfo {author} {\bibfnamefont {A.}~\bibnamefont
  {D{\'i}az-Guilera}}, \ and\ \bibinfo {author} {\bibfnamefont
  {A.}~\bibnamefont {Arenas}},\ }\href {\doibase
  10.1088/1742-5468/2006/11/P11010} {\bibfield  {journal} {\bibinfo  {journal}
  {J. Stat. Mech.}\ }\textbf {\bibinfo {volume} {11}},\ \bibinfo {pages}
  {P11010} (\bibinfo {year} {2006})}\BibitemShut {NoStop}%
\bibitem [{\citenamefont {Everitt}\ \emph {et~al.}(2001)\citenamefont
  {Everitt}, \citenamefont {Landau},\ and\ \citenamefont
  {Leese}}]{ref:everittHC}%
  \BibitemOpen
  \bibfield  {author} {\bibinfo {author} {\bibfnamefont {B.}~\bibnamefont
  {Everitt}}, \bibinfo {author} {\bibfnamefont {S.}~\bibnamefont {Landau}}, \
  and\ \bibinfo {author} {\bibfnamefont {M.}~\bibnamefont {Leese}},\
  }\href@noop {} {\emph {\bibinfo {title} {Cluster analysis}}}\ (\bibinfo
  {year} {2001})\BibitemShut {NoStop}%
\bibitem [{\citenamefont {Clauset}\ \emph {et~al.}(2008)\citenamefont
  {Clauset}, \citenamefont {Moore},\ and\ \citenamefont
  {Newman}}]{ref:clausetmissinglinks}%
  \BibitemOpen
  \bibfield  {author} {\bibinfo {author} {\bibfnamefont {A.}~\bibnamefont
  {Clauset}}, \bibinfo {author} {\bibfnamefont {C.}~\bibnamefont {Moore}}, \
  and\ \bibinfo {author} {\bibfnamefont {M.~E.~J.}\ \bibnamefont {Newman}},\
  }\href {\doibase 10.1038/nature06830} {\bibfield  {journal} {\bibinfo
  {journal} {Nature}\ }\textbf {\bibinfo {volume} {453}},\ \bibinfo {pages}
  {98} (\bibinfo {year} {2008})}\BibitemShut {NoStop}%
\bibitem [{\citenamefont {Sales-Pardo}\ \emph {et~al.}(2007)\citenamefont
  {Sales-Pardo}, \citenamefont {Guimer{\`a}}, \citenamefont {Moreira},\ and\
  \citenamefont {Amaral}}]{ref:salespardo}%
  \BibitemOpen
  \bibfield  {author} {\bibinfo {author} {\bibfnamefont {M.}~\bibnamefont
  {Sales-Pardo}}, \bibinfo {author} {\bibfnamefont {R.}~\bibnamefont
  {Guimer{\`a}}}, \bibinfo {author} {\bibfnamefont {A.~A.}\ \bibnamefont
  {Moreira}}, \ and\ \bibinfo {author} {\bibfnamefont {L.~A.~N.}\ \bibnamefont
  {Amaral}},\ }\href {\doibase 10.1073/pnas.0703740104} {\bibfield  {journal}
  {\bibinfo  {journal} {PNAS}\ }\textbf {\bibinfo {volume} {104}},\ \bibinfo
  {pages} {15224} (\bibinfo {year} {2007})}\BibitemShut {NoStop}%
\bibitem [{\citenamefont {Blondel}\ \emph {et~al.}(2008)\citenamefont
  {Blondel}, \citenamefont {Guillaume}, \citenamefont {Lambiotte},\ and\
  \citenamefont {Lefebvre}}]{ref:blondel}%
  \BibitemOpen
  \bibfield  {author} {\bibinfo {author} {\bibfnamefont {V.~D.}\ \bibnamefont
  {Blondel}}, \bibinfo {author} {\bibfnamefont {J.-L.}\ \bibnamefont
  {Guillaume}}, \bibinfo {author} {\bibfnamefont {R.}~\bibnamefont
  {Lambiotte}}, \ and\ \bibinfo {author} {\bibfnamefont {E.}~\bibnamefont
  {Lefebvre}},\ }\href {\doibase 10.1088/1742-5468/2008/10/P10008} {\bibfield
  {journal} {\bibinfo  {journal} {J. Stat. Mech.}\ }\textbf {\bibinfo {volume}
  {10}},\ \bibinfo {pages} {P10008} (\bibinfo {year} {2008})}\BibitemShut
  {NoStop}%
\bibitem [{\citenamefont {Shen}\ \emph {et~al.}(2009)\citenamefont {Shen},
  \citenamefont {Cheng}, \citenamefont {Cai},\ and\ \citenamefont
  {Hu}}]{ref:shenmodhier}%
  \BibitemOpen
  \bibfield  {author} {\bibinfo {author} {\bibfnamefont {H.}~\bibnamefont
  {Shen}}, \bibinfo {author} {\bibfnamefont {X.}~\bibnamefont {Cheng}},
  \bibinfo {author} {\bibfnamefont {K.}~\bibnamefont {Cai}}, \ and\ \bibinfo
  {author} {\bibfnamefont {M.-B.}\ \bibnamefont {Hu}},\ }\href {\doibase
  10.1016/j.physa.2008.12.021} {\bibfield  {journal} {\bibinfo  {journal}
  {Physica A}\ }\textbf {\bibinfo {volume} {388}},\ \bibinfo {pages} {1706}
  (\bibinfo {year} {2009})}\BibitemShut {NoStop}%
\bibitem [{\citenamefont {Ahn}\ \emph {et~al.}(2010)\citenamefont {Ahn},
  \citenamefont {Bagrow},\ and\ \citenamefont
  {Lehmann}}]{ref:ahnlinkmultiscale}%
  \BibitemOpen
  \bibfield  {author} {\bibinfo {author} {\bibfnamefont {Y.-Y.}\ \bibnamefont
  {Ahn}}, \bibinfo {author} {\bibfnamefont {J.~P.}\ \bibnamefont {Bagrow}}, \
  and\ \bibinfo {author} {\bibfnamefont {S.}~\bibnamefont {Lehmann}},\ }\href
  {\doibase 10.1038/nature09182} {\bibfield  {journal} {\bibinfo  {journal}
  {Nature (London)}\ }\textbf {\bibinfo {volume} {466}},\ \bibinfo {pages}
  {761–764} (\bibinfo {year} {2010})}\BibitemShut {NoStop}%
\bibitem [{\citenamefont {Ravasz}\ and\ \citenamefont
  {Barab{\'a}si}(2003)}]{ref:ravaszbarabasi}%
  \BibitemOpen
  \bibfield  {author} {\bibinfo {author} {\bibfnamefont {E.}~\bibnamefont
  {Ravasz}}\ and\ \bibinfo {author} {\bibfnamefont {A.-L.}\ \bibnamefont
  {Barab{\'a}si}},\ }\href@noop {} {\bibfield  {journal} {\bibinfo  {journal}
  {Phys. Rev. E}\ }\textbf {\bibinfo {volume} {67}},\ \bibinfo {pages} {026112}
  (\bibinfo {year} {2003})}\BibitemShut {NoStop}%
\bibitem [{\citenamefont {Mucha}\ \emph {et~al.}(2010)\citenamefont {Mucha},
  \citenamefont {Richardson}, \citenamefont {Macon}, \citenamefont {Porter},\
  and\ \citenamefont {Onnella}}]{ref:muchaSCI}%
  \BibitemOpen
  \bibfield  {author} {\bibinfo {author} {\bibfnamefont {P.~J.}\ \bibnamefont
  {Mucha}}, \bibinfo {author} {\bibfnamefont {T.}~\bibnamefont {Richardson}},
  \bibinfo {author} {\bibfnamefont {K.}~\bibnamefont {Macon}}, \bibinfo
  {author} {\bibfnamefont {M.~A.}\ \bibnamefont {Porter}}, \ and\ \bibinfo
  {author} {\bibfnamefont {J.-P.}\ \bibnamefont {Onnella}},\ }\href {\doibase
  10.1126/science.1184819} {\bibfield  {journal} {\bibinfo  {journal}
  {Science}\ }\textbf {\bibinfo {volume} {328}},\ \bibinfo {pages} {876}
  (\bibinfo {year} {2010})}\BibitemShut {NoStop}%
\bibitem [{\citenamefont {Zhang}\ \emph
  {et~al.}(2009{\natexlab{b}})\citenamefont {Zhang}, \citenamefont {Zhang},
  \citenamefont {ke~Xu}, \citenamefont {Tse},\ and\ \citenamefont
  {Small}}]{ref:zhangsmall}%
  \BibitemOpen
  \bibfield  {author} {\bibinfo {author} {\bibfnamefont {J.}~\bibnamefont
  {Zhang}}, \bibinfo {author} {\bibfnamefont {K.}~\bibnamefont {Zhang}},
  \bibinfo {author} {\bibfnamefont {X.}~\bibnamefont {ke~Xu}}, \bibinfo
  {author} {\bibfnamefont {C.~K.}\ \bibnamefont {Tse}}, \ and\ \bibinfo
  {author} {\bibfnamefont {M.}~\bibnamefont {Small}},\ }\href {\doibase
  10.1088/1367-2630/11/11/113003} {\bibfield  {journal} {\bibinfo  {journal}
  {New J. Phys.}\ }\textbf {\bibinfo {volume} {11}},\ \bibinfo {pages} {113003}
  (\bibinfo {year} {2009}{\natexlab{b}})}\BibitemShut {NoStop}%
\bibitem [{\citenamefont {Lancichinetti}\ \emph {et~al.}(2011)\citenamefont
  {Lancichinetti}, \citenamefont {Radicchi}, \citenamefont {Ramasco},\ and\
  \citenamefont {Fortunato}}]{ref:lancstatCD}%
  \BibitemOpen
  \bibfield  {author} {\bibinfo {author} {\bibfnamefont {A.}~\bibnamefont
  {Lancichinetti}}, \bibinfo {author} {\bibfnamefont {F.}~\bibnamefont
  {Radicchi}}, \bibinfo {author} {\bibfnamefont {J.~J.}\ \bibnamefont
  {Ramasco}}, \ and\ \bibinfo {author} {\bibfnamefont {S.}~\bibnamefont
  {Fortunato}},\ }\href {\doibase 10.1371/journal.pone.0018961} {\bibfield
  {journal} {\bibinfo  {journal} {PLoS ONE}\ }\textbf {\bibinfo {volume} {6}},\
  \bibinfo {pages} {e18961} (\bibinfo {year} {2011})}\BibitemShut {NoStop}%
\bibitem [{\citenamefont {Shen}\ \emph
  {et~al.}(2010{\natexlab{a}})\citenamefont {Shen}, \citenamefont {Cheng},\
  and\ \citenamefont {Fang}}]{ref:shenccmmultiscale}%
  \BibitemOpen
  \bibfield  {author} {\bibinfo {author} {\bibfnamefont {H.-W.}\ \bibnamefont
  {Shen}}, \bibinfo {author} {\bibfnamefont {X.-Q.}\ \bibnamefont {Cheng}}, \
  and\ \bibinfo {author} {\bibfnamefont {B.-X.}\ \bibnamefont {Fang}},\ }\href
  {\doibase 10.1103/PhysRevE.82.016114} {\bibfield  {journal} {\bibinfo
  {journal} {Phys. Rev. E}\ }\textbf {\bibinfo {volume} {82}},\ \bibinfo
  {pages} {016114} (\bibinfo {year} {2010}{\natexlab{a}})}\BibitemShut
  {NoStop}%
\bibitem [{\citenamefont {Shen}\ \emph
  {et~al.}(2010{\natexlab{b}})\citenamefont {Shen}, \citenamefont {Cheng},\
  and\ \citenamefont {Fang}}]{ref:shenlaplacianPRE}%
  \BibitemOpen
  \bibfield  {author} {\bibinfo {author} {\bibfnamefont {H.-W.}\ \bibnamefont
  {Shen}}, \bibinfo {author} {\bibfnamefont {X.-Q.}\ \bibnamefont {Cheng}}, \
  and\ \bibinfo {author} {\bibfnamefont {B.-X.}\ \bibnamefont {Fang}},\ }\href
  {\doibase 10.1103/PhysRevE.82.016114} {\bibfield  {journal} {\bibinfo
  {journal} {Phys. Rev. E}\ }\textbf {\bibinfo {volume} {82}},\ \bibinfo
  {pages} {016114} (\bibinfo {year} {2010}{\natexlab{b}})}\BibitemShut
  {NoStop}%
\bibitem [{\citenamefont {Zhang}\ and\ \citenamefont
  {Zhao}(2012)}]{ref:zhangunbalanced}%
  \BibitemOpen
  \bibfield  {author} {\bibinfo {author} {\bibfnamefont {S.}~\bibnamefont
  {Zhang}}\ and\ \bibinfo {author} {\bibfnamefont {H.}~\bibnamefont {Zhao}},\
  }\href {\doibase 10.1103/PhysRevE.85.066114} {\bibfield  {journal} {\bibinfo
  {journal} {Phys. Rev. E}\ }\textbf {\bibinfo {volume} {85}},\ \bibinfo
  {pages} {066114} (\bibinfo {year} {2012})}\BibitemShut {NoStop}%
\bibitem [{\citenamefont {Shen}\ and\ \citenamefont
  {Cheng}(2010)}]{ref:shenchengspectral}%
  \BibitemOpen
  \bibfield  {author} {\bibinfo {author} {\bibfnamefont {H.-W.}\ \bibnamefont
  {Shen}}\ and\ \bibinfo {author} {\bibfnamefont {X.-Q.}\ \bibnamefont
  {Cheng}},\ }\href {\doibase 10.1088/1742-5468/2010/10/P10020} {\bibfield
  {journal} {\bibinfo  {journal} {J. Stat. Mech.}\ }\textbf {\bibinfo {volume}
  {10}},\ \bibinfo {pages} {P10020} (\bibinfo {year} {2010})}\BibitemShut
  {NoStop}%
\bibitem [{\citenamefont {Hofman}\ and\ \citenamefont
  {Wiggins}(2008)}]{ref:Hofman2008}%
  \BibitemOpen
  \bibfield  {author} {\bibinfo {author} {\bibfnamefont {J.~M.}\ \bibnamefont
  {Hofman}}\ and\ \bibinfo {author} {\bibfnamefont {C.~H.}\ \bibnamefont
  {Wiggins}},\ }\href {\doibase 10.1103/PhysRevLett.100.258701} {\bibfield
  {journal} {\bibinfo  {journal} {Phys. Rev. Lett.}\ }\textbf {\bibinfo
  {volume} {100}},\ \bibinfo {pages} {258701} (\bibinfo {year}
  {2008})}\BibitemShut {NoStop}%
\bibitem [{\citenamefont {Gudkov}\ \emph {et~al.}(2008)\citenamefont {Gudkov},
  \citenamefont {Montealegre}, \citenamefont {Nussinov},\ and\ \citenamefont
  {Nussinov}}]{ref:gudkov}%
  \BibitemOpen
  \bibfield  {author} {\bibinfo {author} {\bibfnamefont {V.}~\bibnamefont
  {Gudkov}}, \bibinfo {author} {\bibfnamefont {V.}~\bibnamefont {Montealegre}},
  \bibinfo {author} {\bibfnamefont {S.}~\bibnamefont {Nussinov}}, \ and\
  \bibinfo {author} {\bibfnamefont {Z.}~\bibnamefont {Nussinov}},\ }\href
  {\doibase 10.1103/PhysRevE.78.016113} {\bibfield  {journal} {\bibinfo
  {journal} {Phys. Rev. E}\ }\textbf {\bibinfo {volume} {78}},\ \bibinfo
  {pages} {016113} (\bibinfo {year} {2008})}\BibitemShut {NoStop}%
\bibitem [{\citenamefont {Boccaletti}\ \emph {et~al.}(2007)\citenamefont
  {Boccaletti}, \citenamefont {Ivanchenko}, \citenamefont {Latora},
  \citenamefont {Pluchino},\ and\ \citenamefont {Rapisarda}}]{ref:boccaletti}%
  \BibitemOpen
  \bibfield  {author} {\bibinfo {author} {\bibfnamefont {S.}~\bibnamefont
  {Boccaletti}}, \bibinfo {author} {\bibfnamefont {M.}~\bibnamefont
  {Ivanchenko}}, \bibinfo {author} {\bibfnamefont {V.}~\bibnamefont {Latora}},
  \bibinfo {author} {\bibfnamefont {A.}~\bibnamefont {Pluchino}}, \ and\
  \bibinfo {author} {\bibfnamefont {A.}~\bibnamefont {Rapisarda}},\ }\href
  {\doibase 10.1103/PhysRevE.75.045102} {\bibfield  {journal} {\bibinfo
  {journal} {Phys. Rev. E}\ }\textbf {\bibinfo {volume} {75}},\ \bibinfo
  {pages} {045102(R)} (\bibinfo {year} {2007})}\BibitemShut {NoStop}%
\bibitem [{\citenamefont {Kanungo}\ \emph {et~al.}(2002)\citenamefont
  {Kanungo}, \citenamefont {Mount}, \citenamefont {Netanyahu}, \citenamefont
  {Piatko}, \citenamefont {Silverman},\ and\ \citenamefont
  {Wu}}]{ref:kanungo2002}%
  \BibitemOpen
  \bibfield  {author} {\bibinfo {author} {\bibfnamefont {T.}~\bibnamefont
  {Kanungo}}, \bibinfo {author} {\bibfnamefont {D.}~\bibnamefont {Mount}},
  \bibinfo {author} {\bibfnamefont {N.}~\bibnamefont {Netanyahu}}, \bibinfo
  {author} {\bibfnamefont {C.}~\bibnamefont {Piatko}}, \bibinfo {author}
  {\bibfnamefont {R.}~\bibnamefont {Silverman}}, \ and\ \bibinfo {author}
  {\bibfnamefont {A.}~\bibnamefont {Wu}},\ }\href {\doibase
  10.1109/TPAMI.2002.1017616} {\bibfield  {journal} {\bibinfo  {journal}
  {Pattern Analysis and Machine Intelligence, IEEE Transactions on}\ }\textbf
  {\bibinfo {volume} {24}},\ \bibinfo {pages} {881} (\bibinfo {year}
  {2002})}\BibitemShut {NoStop}%
\bibitem [{\citenamefont {Ball}\ \emph {et~al.}(2011)\citenamefont {Ball},
  \citenamefont {Karrer},\ and\ \citenamefont {Newman}}]{ref:ballMLlinks}%
  \BibitemOpen
  \bibfield  {author} {\bibinfo {author} {\bibfnamefont {B.}~\bibnamefont
  {Ball}}, \bibinfo {author} {\bibfnamefont {B.}~\bibnamefont {Karrer}}, \ and\
  \bibinfo {author} {\bibfnamefont {M.~E.~J.}\ \bibnamefont {Newman}},\ }\href
  {\doibase 10.1103/PhysRevE.84.036103} {\bibfield  {journal} {\bibinfo
  {journal} {Phys. Rev. E}\ }\textbf {\bibinfo {volume} {84}},\ \bibinfo
  {pages} {036103} (\bibinfo {year} {2011})}\BibitemShut {NoStop}%
\bibitem [{\citenamefont {Radicchi}\ \emph {et~al.}(2004)\citenamefont
  {Radicchi}, \citenamefont {Castellano}, \citenamefont {Cecconi},
  \citenamefont {Loreto},\ and\ \citenamefont {Parisi}}]{ref:radicchi}%
  \BibitemOpen
  \bibfield  {author} {\bibinfo {author} {\bibfnamefont {F.}~\bibnamefont
  {Radicchi}}, \bibinfo {author} {\bibfnamefont {C.}~\bibnamefont
  {Castellano}}, \bibinfo {author} {\bibfnamefont {F.}~\bibnamefont {Cecconi}},
  \bibinfo {author} {\bibfnamefont {V.}~\bibnamefont {Loreto}}, \ and\ \bibinfo
  {author} {\bibfnamefont {D.}~\bibnamefont {Parisi}},\ }\href {\doibase
  10.1073/pnas.0400054101} {\bibfield  {journal} {\bibinfo  {journal} {Proc.
  Natl. Acad. Sci. U.S.A.}\ }\textbf {\bibinfo {volume} {101}},\ \bibinfo
  {pages} {2658} (\bibinfo {year} {2004})}\BibitemShut {NoStop}%
\bibitem [{\citenamefont {Cafieri}\ \emph {et~al.}(2012)\citenamefont
  {Cafieri}, \citenamefont {Caporossi}, \citenamefont {Hansen}, \citenamefont
  {Perron},\ and\ \citenamefont {Costa}}]{ref:cafieristrong}%
  \BibitemOpen
  \bibfield  {author} {\bibinfo {author} {\bibfnamefont {S.}~\bibnamefont
  {Cafieri}}, \bibinfo {author} {\bibfnamefont {G.}~\bibnamefont {Caporossi}},
  \bibinfo {author} {\bibfnamefont {P.}~\bibnamefont {Hansen}}, \bibinfo
  {author} {\bibfnamefont {S.}~\bibnamefont {Perron}}, \ and\ \bibinfo {author}
  {\bibfnamefont {A.}~\bibnamefont {Costa}},\ }\href {\doibase
  10.1103/PhysRevE.85.046113} {\bibfield  {journal} {\bibinfo  {journal} {Phys.
  Rev. E}\ }\textbf {\bibinfo {volume} {85}},\ \bibinfo {pages} {046113}
  (\bibinfo {year} {2012})}\BibitemShut {NoStop}%
\bibitem [{\citenamefont {Zachary}(1977)}]{ref:zachary}%
  \BibitemOpen
  \bibfield  {author} {\bibinfo {author} {\bibfnamefont {W.~W.}\ \bibnamefont
  {Zachary}},\ }\href@noop {} {\bibfield  {journal} {\bibinfo  {journal} {J.
  Anthropol. Res.}\ }\textbf {\bibinfo {volume} {33}},\ \bibinfo {pages} {452}
  (\bibinfo {year} {1977})}\BibitemShut {NoStop}%
\bibitem [{\citenamefont {Lancichinetti}\ and\ \citenamefont
  {Fortunato}(2009)}]{ref:lancLFRcompare}%
  \BibitemOpen
  \bibfield  {author} {\bibinfo {author} {\bibfnamefont {A.}~\bibnamefont
  {Lancichinetti}}\ and\ \bibinfo {author} {\bibfnamefont {S.}~\bibnamefont
  {Fortunato}},\ }\href {\doibase 10.1103/PhysRevE.80.056117} {\bibfield
  {journal} {\bibinfo  {journal} {Phys. Rev. E}\ }\textbf {\bibinfo {volume}
  {80}},\ \bibinfo {pages} {056117} (\bibinfo {year} {2009})}\BibitemShut
  {NoStop}%
\bibitem [{\citenamefont {Lancichinetti}\ \emph {et~al.}(2008)\citenamefont
  {Lancichinetti}, \citenamefont {Fortunato},\ and\ \citenamefont
  {Radicchi}}]{ref:lancbenchmark}%
  \BibitemOpen
  \bibfield  {author} {\bibinfo {author} {\bibfnamefont {A.}~\bibnamefont
  {Lancichinetti}}, \bibinfo {author} {\bibfnamefont {S.}~\bibnamefont
  {Fortunato}}, \ and\ \bibinfo {author} {\bibfnamefont {F.}~\bibnamefont
  {Radicchi}},\ }\href {\doibase 10.1103/PhysRevE.78.046110} {\bibfield
  {journal} {\bibinfo  {journal} {Phys. Rev. E}\ }\textbf {\bibinfo {volume}
  {78}},\ \bibinfo {pages} {046110} (\bibinfo {year} {2008})}\BibitemShut
  {NoStop}%
\bibitem [{\citenamefont {Fred}\ and\ \citenamefont
  {Jain}(2003)}]{ref:fredjain}%
  \BibitemOpen
  \bibfield  {author} {\bibinfo {author} {\bibfnamefont {A.~L.~N.}\
  \bibnamefont {Fred}}\ and\ \bibinfo {author} {\bibfnamefont {A.~K.}\
  \bibnamefont {Jain}},\ }in\ \href {\doibase 10.1109/CVPR.2003.1211462} {\emph
  {\bibinfo {booktitle} {Proceedings of the IEEE Computer Society Conference on
  Computer Vision Pattern Recognition}}},\ Vol.~\bibinfo {volume} {2}\
  (\bibinfo  {publisher} {IEEE Computer Society},\ \bibinfo {year} {2003})\
  pp.\ \bibinfo {pages} {128--133}\BibitemShut {NoStop}%
\bibitem [{\citenamefont {Topchy}\ \emph {et~al.}(2004)\citenamefont {Topchy},
  \citenamefont {Law}, \citenamefont {Jain},\ and\ \citenamefont
  {Fred}}]{ref:topchyconsensus}%
  \BibitemOpen
  \bibfield  {author} {\bibinfo {author} {\bibfnamefont {A.~P.}\ \bibnamefont
  {Topchy}}, \bibinfo {author} {\bibfnamefont {M.~H.~C.}\ \bibnamefont {Law}},
  \bibinfo {author} {\bibfnamefont {A.~K.}\ \bibnamefont {Jain}}, \ and\
  \bibinfo {author} {\bibfnamefont {A.~L.}\ \bibnamefont {Fred}},\ }in\ \href
  {\doibase 10.1109/ICDM.2004.10100} {\emph {\bibinfo {booktitle} {Data Mining,
  2004. ICDM '04. Fourth IEEE International Conference}}}\ (\bibinfo
  {publisher} {IEEE Computer Society},\ \bibinfo {year} {2004})\ pp.\ \bibinfo
  {pages} {225--232}\BibitemShut {NoStop}%
\bibitem [{\citenamefont {Topchy}\ \emph {et~al.}(2003)\citenamefont {Topchy},
  \citenamefont {Jain},\ and\ \citenamefont {Punch}}]{ref:topchyweak}%
  \BibitemOpen
  \bibfield  {author} {\bibinfo {author} {\bibfnamefont {A.~P.}\ \bibnamefont
  {Topchy}}, \bibinfo {author} {\bibfnamefont {A.~K.}\ \bibnamefont {Jain}}, \
  and\ \bibinfo {author} {\bibfnamefont {W.}~\bibnamefont {Punch}},\ }in\ \href
  {\doibase 10.1109/ICDM.2003.1250937} {\emph {\bibinfo {booktitle} {Data
  Mining, 2003. ICDM 2003. Third IEEE International Conference}}}\ (\bibinfo
  {publisher} {IEEE Computer Society},\ \bibinfo {year} {2003})\ pp.\ \bibinfo
  {pages} {331--338}\BibitemShut {NoStop}%
\bibitem [{\citenamefont {Hu}\ \emph {et~al.}(2012{\natexlab{d}})\citenamefont
  {Hu}, \citenamefont {Ronhovde},\ and\ \citenamefont
  {Nussinov}}]{ref:HRNimages}%
  \BibitemOpen
  \bibfield  {author} {\bibinfo {author} {\bibfnamefont {D.}~\bibnamefont
  {Hu}}, \bibinfo {author} {\bibfnamefont {P.}~\bibnamefont {Ronhovde}}, \ and\
  \bibinfo {author} {\bibfnamefont {Z.}~\bibnamefont {Nussinov}},\ }\href
  {\doibase 10.1103/PhysRevE.85.016101} {\bibfield  {journal} {\bibinfo
  {journal} {Phys. Rev. E}\ }\textbf {\bibinfo {volume} {85}},\ \bibinfo
  {pages} {016101} (\bibinfo {year} {2012}{\natexlab{d}})}\BibitemShut
  {NoStop}%
\bibitem [{\citenamefont {Villain}\ \emph {et~al.}(1980)\citenamefont
  {Villain}, \citenamefont {Bidaux}, \citenamefont {Carton},\ and\
  \citenamefont {Conte}}]{ref:villain1980JPF}%
  \BibitemOpen
  \bibfield  {author} {\bibinfo {author} {\bibfnamefont {J.}~\bibnamefont
  {Villain}}, \bibinfo {author} {\bibfnamefont {R.}~\bibnamefont {Bidaux}},
  \bibinfo {author} {\bibfnamefont {J.-P.}\ \bibnamefont {Carton}}, \ and\
  \bibinfo {author} {\bibfnamefont {R.}~\bibnamefont {Conte}},\ }\href
  {\doibase 10.1051/jphys:0198000410110126300} {\bibfield  {journal} {\bibinfo
  {journal} {J. Phys. France}\ }\textbf {\bibinfo {volume} {41}},\ \bibinfo
  {pages} {1263} (\bibinfo {year} {1980})}\BibitemShut {NoStop}%
\bibitem [{\citenamefont {Shender}(1982)}]{ref:shender1982SP}%
  \BibitemOpen
  \bibfield  {author} {\bibinfo {author} {\bibfnamefont {E.~F.}\ \bibnamefont
  {Shender}},\ }\href@noop {} {\bibfield  {journal} {\bibinfo  {journal} {Sov.
  Phys.}\ }\textbf {\bibinfo {volume} {56}},\ \bibinfo {pages} {178} (\bibinfo
  {year} {1982})}\BibitemShut {NoStop}%
\bibitem [{\citenamefont {Henley}(1989)}]{ref:henley1989PRL}%
  \BibitemOpen
  \bibfield  {author} {\bibinfo {author} {\bibfnamefont {C.~L.}\ \bibnamefont
  {Henley}},\ }\href {\doibase 10.1103/PhysRevLett.62.2056} {\bibfield
  {journal} {\bibinfo  {journal} {Phys. Rev. Lett.}\ }\textbf {\bibinfo
  {volume} {62}},\ \bibinfo {pages} {2056} (\bibinfo {year}
  {1989})}\BibitemShut {NoStop}%
\bibitem [{\citenamefont {Nussinov}\ \emph {et~al.}(2004)\citenamefont
  {Nussinov}, \citenamefont {Biskup}, \citenamefont {Chayes},\ and\
  \citenamefont {van~den Brink}}]{ref:nussinov2004EPL}%
  \BibitemOpen
  \bibfield  {author} {\bibinfo {author} {\bibfnamefont {Z.}~\bibnamefont
  {Nussinov}}, \bibinfo {author} {\bibfnamefont {M.}~\bibnamefont {Biskup}},
  \bibinfo {author} {\bibfnamefont {L.}~\bibnamefont {Chayes}}, \ and\ \bibinfo
  {author} {\bibfnamefont {J.}~\bibnamefont {van~den Brink}},\ }\href@noop {}
  {\bibfield  {journal} {\bibinfo  {journal} {Europhys. Lett.}\ }\textbf
  {\bibinfo {volume} {67}},\ \bibinfo {pages} {990} (\bibinfo {year}
  {2004})}\BibitemShut {NoStop}%
\bibitem [{\citenamefont {Trusina}\ \emph {et~al.}(2004)\citenamefont
  {Trusina}, \citenamefont {Maslov}, \citenamefont {Minnhagen},\ and\
  \citenamefont {Sneppen}}]{ref:trusinahiermeasure}%
  \BibitemOpen
  \bibfield  {author} {\bibinfo {author} {\bibfnamefont {A.}~\bibnamefont
  {Trusina}}, \bibinfo {author} {\bibfnamefont {S.}~\bibnamefont {Maslov}},
  \bibinfo {author} {\bibfnamefont {P.}~\bibnamefont {Minnhagen}}, \ and\
  \bibinfo {author} {\bibfnamefont {K.}~\bibnamefont {Sneppen}},\ }\href
  {\doibase 10.1103/PhysRevLett.92.178702} {\bibfield  {journal} {\bibinfo
  {journal} {Phys. Rev. Lett.}\ }\textbf {\bibinfo {volume} {92}},\ \bibinfo
  {pages} {178702} (\bibinfo {year} {2004})}\BibitemShut {NoStop}%
\bibitem [{\citenamefont {Mones}\ \emph {et~al.}(2012)\citenamefont {Mones},
  \citenamefont {Vicsek},\ and\ \citenamefont {Vicsek}}]{ref:moneshiermeasure}%
  \BibitemOpen
  \bibfield  {author} {\bibinfo {author} {\bibfnamefont {E.}~\bibnamefont
  {Mones}}, \bibinfo {author} {\bibfnamefont {L.}~\bibnamefont {Vicsek}}, \
  and\ \bibinfo {author} {\bibfnamefont {T.}~\bibnamefont {Vicsek}},\ }\href
  {\doibase 10.1371/journal.pone.0033799} {\bibfield  {journal} {\bibinfo
  {journal} {PLoS ONE}\ }\textbf {\bibinfo {volume} {7}},\ \bibinfo {pages}
  {e33799} (\bibinfo {year} {2012})}\BibitemShut {NoStop}%
\bibitem [{\citenamefont {Krebs}(2002)}]{ref:krebsconn}%
  \BibitemOpen
  \bibfield  {author} {\bibinfo {author} {\bibfnamefont {V.~E.}\ \bibnamefont
  {Krebs}},\ }\href@noop {} {\bibfield  {journal} {\bibinfo  {journal}
  {Connections}\ }\textbf {\bibinfo {volume} {24}},\ \bibinfo {pages} {43}
  (\bibinfo {year} {2002})}\BibitemShut {NoStop}%
\end{thebibliography}
%

\end{document}